\documentclass[preprint]{elsarticle} 
\usepackage{proof} 
\usepackage{amssymb}
\usepackage{url}
\usepackage{abbrevs}
\usepackage[margin=1in]{geometry}
\def\Dscr{{\mathcal D}}

\def\Fscr{{\mathcal F}}
\def\Gscr{{\mathcal G}}

\def\Lscr{{\mathcal L}}
\def\Mscr{{\mathcal M}}

\def\Pscr{{\mathcal P}}

\def\Rscr{{\mathcal R}}
\def\Sscr{{\mathcal S}}
\def\Tscr{{\mathcal T}}

\newcommand{\Momega}{\mathcal{M}_{\omega}}

\newcommand{\FOL   }{FO\lambda}

\newcommand{\FOLDN }{\FOL^{\Delta\N}}

\newcommand{\FOLNb }{\FOL^{\nabla}}

\newcommand{\Linc}{{\rm Linc}^-}
\newcommand{\fLinc}{{\rm Linc}}

\newcommand{\N}{{\rm I} \! {\rm N}}
\newcommand{\Seq}[2]{#1\longrightarrow #2}
\newcommand{\botL}{\bot{\cal L}}
\newcommand{\bulletL}{\bullet{\cal L}}

\newcommand{\cL}{\hbox{\sl c}{\cal L}}
\newcommand{\circL}{\circ{\cal L}}
\newcommand{\circR}{\circ{\cal R}}

\newcommand{\defeq}{\stackrel{\scriptscriptstyle\triangle}{=}}
\newcommand{\defmu}{\stackrel{\mu}{=}}
\newcommand{\defnu}{\stackrel{\nu}{=}}

\newcommand{\eqL}{{\rm eq}{\cal L}}
\newcommand{\eqR}{{\rm eq}{\cal R}}
\newcommand{\existsL}{\exists{\cal L}}
\newcommand{\existsR}{\exists{\cal R}}

\newcommand{\forallL}{\forall{\cal L}}
\newcommand{\forallR}{\forall{\cal R}}

\newcommand{\indR}{{\rm I}{\cal R}}
\newcommand{\indRP}{{\rm I}{\cal R}_p}
\newcommand{\indL}{{\rm I}{\cal L}}

\newcommand{\coindR}{{\rm CI}{\cal R}}
\newcommand{\coindL}{{\rm CI}{\cal L}}
\newcommand{\coindLP}{{\rm CI}{\cal L}_p}

\newcommand{\init}{\hbox{\sl init}}
\newcommand{\landL}{\land{\cal L}}
\newcommand{\landR}{\land{\cal R}}
\newcommand{\level}[1]{{\rm lvl}(#1)}
\newcommand{\lorL}{\lor{\cal L}}
\newcommand{\lorR}{\lor{\cal R}}

\newcommand{\lub}[1]{{\rm lub}(#1)}
\newcommand{\mc}{\hbox{\sl mc}}
\newcommand{\measure}[1]{{\rm ht}(#1)}

\newcommand{\oimpL}{\oimp{\cal L}}
\newcommand{\oimpR}{\oimp{\cal R}}
\newcommand{\oimp}{\supset}

\newcommand{\ra}{\to}

\newcommand{\topR}{\top{\cal R}}

\newcommand{\wL}{\hbox{\sl w}{\cal L}}

\def\RED{{\mathbf{RED}}}
\def\NM{{\mathbf{NM}}}
\def\idrv{{\mathrm{Id}}}

\newtheorem{thm}{Theorem}
\newdefinition{definition}{Definition }
\newproof{proof}{Proof}
\newtheorem{lemma}[thm]{Lemma}
\newtheorem{corollary}[thm]{Corollary}
\newdefinition{remark}{Remark }
\def\qed{$\hfill\Box$}

\newenvironment{lemmacp}[1]{\noindent {\bf Lemma~\ref{#1}.} \em}{}

                         {\endgroup]\end{quote}}
\long\def\ignore#1{}
\newlength{\infwidthi}

\journal{Journal of Applied Logic}
\begin{document}

\begin{frontmatter}

\title{Cut Elimination for a Logic with Induction and Co-induction}

\author{Alwen Tiu\corref{cor1}}

\ead{Alwen.Tiu@anu.edu.au }
\address{Logic and Computation Group\\
College of Engineering and Computer Science\\
The Australian National University}

\author{Alberto Momigliano\fnref{f1}}
\ead{amomigl1@inf.ed.ac.uk}
\address{Laboratory for Foundations of Computer Science\\
School of Informatics\\ The University of Edinburgh}

\cortext[cor1]{{Corresponding author. The Australian National
    University, Canberra, ACT 0200, Australia.  Phone: +61 (0)2 6125
    5992.  Fax: +61 (0)2 6125 8645}}

\fntext[f1]{Present address: Dipartimento di Scienze
  dell'Informazione, Universit\`{a} degli Studi di Milano, Italy.}

\begin{abstract}
  Proof search has been used to specify a wide range of computation
  systems.  In order to build a framework for reasoning about such
  specifications, we make use of a sequent calculus involving
  induction and co-induction.  These proof principles are based on a
  proof theoretic (rather than set-theoretic) notion of
  \emph{definition}~\cite{PID,eriksson91elp,schroeder-heister93lics,mcdowell00tcs}.
  Definitions are akin to logic programs, where the left and right
  rules for defined atoms allow one to view theories as ``closed'' or
  defining fixed points.  The use of definitions and free equality
  makes it possible to reason intentionally about syntax.  We add in a
  consistent way rules for pre and post fixed points, thus allowing
  the user to reason inductively and co-inductively about properties
  of computational system making full use of higher-order abstract
  syntax.  Consistency is guaranteed via cut-elimination, where we
  give the first, to our knowledge, cut-elimination procedure in the
  presence of general inductive and co-inductive definitions.
\end{abstract}

\end{frontmatter}

\section{Introduction}
\label{sec:intro}

A common approach to specifying computation systems is via deductive
systems. Those are used to specify and reason about various logics, as
well as aspects of programming languages such as operational
semantics, type theories, abstract machines \etc.  Such specifications
can be represented as logical theories in a suitably expressive formal
logic where \emph{proof-search} can then be used to model the
computation.  A logic used as a specification language is known as a
\emph{logical frameworks}~\cite{pfenning01handbook}, which comes
equipped with a representation methodology.  The encoding of the
syntax of deductive systems inside formal logic can benefit from the
use of \emph{higher-order abstract syntax} (HOAS) 
a high-level and declarative treatment of object-level bound variables
and substitution. At the same time, we want to use such a logic to
reason over the \emph{meta-theoretical} properties of object
languages, for example type preservation in operational
semantics~\cite{mcdowell02tocl}, soundness and completeness of
compilation~\cite{Momigliano03fos} or congruence of bisimulation in
transition systems~\cite{mcdowell03tcs}. Typically this involves
reasoning by (structural) induction and, when dealing with infinite
behavior, co-induction~\cite{Jacobs97}.

The need to support both inductive and co-inductive reasoning and some
form of HOAS requires some careful design decisions, since the two are
prima facie notoriously incompatible. While any meta-language based on
a $\lambda$-calculus can be used to specify and animate HOAS
encodings, meta-reasoning has traditionally involved (co)inductive
specifications both at the level of the syntax and of the judgements
--- which are of course unified at the type-theoretic level. The first
provides crucial freeness properties for datatypes constructors, while
the second offers principles of case analysis and (co)induction. This
is well-known to be problematic, since HOAS specifications may lead to
non-monotone (co)inductive operators, which by cardinality and
consistency reasons are not permitted in inductive logical
frameworks. Moreover, even when HOAS is weakened so as to be made
compatible with standard proof assistants~\cite{despeyroux94lpar} such
as HOL or Coq, the latter suffer the fate of allowing the existence of
too many functions and yielding the so called \emph{exotic}
terms. Those are canonical terms in the signature of an HOAS encoding
that do not correspond to any term in the deductive system under
study. This causes a loss of adequacy in HOAS specifications, which is
one of the pillar of formal verification, and it undermines the trust
in formal derivations. On the other hand, logics such as
LF~\cite{harper93jacm} that are weak by design 
in order to support this style of syntax are not directly endowed with
(co)induction principles.


The contribution of this paper lies in the design of a new logic,
called $\Linc$ (for a logic with $\lambda$-terms, induction and
co-induction),\footnote{The ``minus'' in the terminology refers to the
  lack of the $\nabla$-quantifier \wrt the eponymous logic in Tiu's
  thesis~\cite{tiu04phd}.} which carefully adds principles of
induction and co-induction to a higher-order intuitionistic logic
based on a proof theoretic notion of \emph{definition}, following on
work (among others) by Lars Halln{\"{a}}s~\cite{PID}, Eriksson
\cite{eriksson91elp}, Schroeder-Heister~\cite{schroeder-heister93lics}
and McDowell and Miller~\cite{mcdowell00tcs}.  Definitions are akin to
logic programs, but allow us to view theories as ``closed'' or
defining fixed points.  This alone permits to perform case analysis
independently from induction principles.  Our approach to formalizing
induction and co-induction is via the least and greatest solutions of
the fixed point equations specified by the definitions. 
The proof rules for induction and co-induction make use of the notion
of \emph{pre-fixed points} and \emph{post-fixed points}
respectively. In the inductive case, this corresponds to the induction
invariant, while in the co-inductive one to the so-called simulation.
Judgements are encoded as definitions accordingly to their informal
semantics, either inductive or co-inductive.

The simply typed language and the notion of free equality underlying
$\Linc$, enforced via (higher-order) unification in an inference rule,
make it possible to reason \emph{intensionally} about syntax.
In fact, we can support HOAS encodings of constants
and we can \emph{prove} the freeness properties of those constants,
namely injectivity, distinctness and case exhaustion, although they
cannot be the constructors of a (recursive)
datatype.  


$\Linc$ can be proved to be a conservative extension of
$\FOLDN$~\cite{mcdowell00tcs} and a generalization (with a term
language based on simply typed $\lambda$-calculus) of Martin-L\"of
first-order theory of iterated inductive
definitions~\cite{martin-lof71sls}. Moreover, to the best of our
knowledge, it is the first sequent calculus with a syntactical
cut-elimination theorem for co-inductive definitions.  In recent
years, several logical systems have been designed that build on the
core features of $\Linc$. In particular, one interesting, and
orthogonal, extension is the addition of the
$\nabla$-quantifier~\cite{miller05tocl,tiu04phd,Tiu07,gacek08lics},
which allows one to reason about the intentional aspects of
\emph{names and bindings} in object syntax specifications (see,
e.g.,~\cite{AbellaSOS,TiuMillerpi}).  The cut elimination
proof presented in this paper can be used as a springboard towards cut
elimination procedures for more expressive (conservative) extensions
of $\Linc$. 

In fact, the possibility of adapting the cut elimination proof for
$\Linc$ to various extensions of $\Linc$ with $\nabla$ is one of the
main reasons to introduce a \emph{direct} syntactic cut elimination
proof.  We note that there are at least a couple of indirect methods
to prove cut elimination in a logic with inductive and/or co-inductive
definitions.  The first of such methods relies on encodings of
inductive and co-inductive definitions as second-order (or
higher-order) formulae. This approach is followed in a recent work by
Baelde and Miller~\cite{baelde07lpar} where a  logic similar to $\Linc$
is considered. Cut elimination in their work is proved indirectly via
an encoding into higher-order linear logic. However, in the presence
of $\nabla$, the existence of such an encoding is presently unknown.
The second approach is via semantical methods. This approach is taken
in a recent work by Brotherston and Simpson~\cite{BrotherstonS07},
which provide a model for
a classical first-order logic with inductive definitions, hence, cut
elimination follows  by the semantical completeness of the cut
free fragment.  It is not obvious how such semantical methods can be
adapted to prove cut elimination for extensions of $\Linc$ with
$\nabla$.  This is 
 because the semantics of $\nabla$ itself is not yet
very well understood, although there have been some recent attempts,
see~\cite{Miculan05FOSSACS,Schoepp07LFMTP,GabbayFL}.

The present paper is an extended and revised version
of
~\cite{Momigliano03TYPES}.  In the conference paper, the co-inductive
rule had a technical side condition that is restrictive and
unnatural. The restriction was essentially imposed by the particular
cut elimination proof technique outlined in that paper.  This
restriction has been removed in the present version, and the
(co-)induction rules have been generalized. For the latter, the
formulation of the rules is inspired by a second-order encoding of
least and greatest fixed points.  Consequently, we now develop a new
cut elimination proof, which is radically different from the previous
proof, using a reducibility-candidate technique, which is influenced
by Girard's strong normalisation proof for System
F~\cite{girard89book}.  This paper is concerned only with the cut
elimination proof of $\Linc$.  For examples and applications of
$\Linc$ and its extensions with $\nabla$, we refer the interested reader
to~\cite{tiu04phd,Bedwyr,gacek08lics,Abella,AbellaSOS,TiuMillerpi}.


The rest of the paper is organized as follows. Section~\ref{sec:linc}
introduces the sequent calculus for the logic.  Section~\ref{sec:drv}
presents two transformations of derivations that are essential to the
cut reduction rules and the cut elimination proof in subsequent
sections. Section~\ref{sec:cut-elim} is the heart of the paper: we
first (Subsection~\ref{sec:reduc}) give a (sub)set of reduction rules
that transform a derivation ending with a cut rule to another
derivation. The complete set of reduction can be found in
Appendix~\ref{app:reduc}. We then introduce the crucial notions of
\emph{normalizability} (Subsection~\ref{sec:norm}) and of
\emph{parametric reducibility} after Girard
(Subsection~\ref{sec:red}). Detailed proofs of the main lemma related
to reducibility candidates are in Appendix~\ref{app:red}.  The
central result of this paper, i.e., cut elimination, is proved in
details in Subsection~\ref{sec:ceproof}. 
Section~\ref{sec:lrel} surveys the related work and
concludes the paper.

\section{The Logic $\Linc$}
\label{sec:linc}

\begin{figure}
$$
\begin{array}{cc}
  \multicolumn{2}{c}{
    \infer[init]{\Seq{C}{C}}{} 
    \quad
    \infer[\cL]{\Seq{B,\Gamma}{C}}
    {\Seq{B,B,\Gamma}{C}}
    \quad
    \infer[\wL]{\Seq{B,\Gamma}{C}}{\Seq{\Gamma}{C}}
  }
  \\ \\
  \multicolumn{2}{c}{       
    \infer[\begin{array}{l}
      \mc, 
      \mbox{where } n > 0 
    \end{array}]
    {\Seq{\Delta_1,\dots,\Delta_n, \Gamma}{C}}
    {\Seq{\Delta_1}{B_1}
      & \cdots &
      \Seq{\Delta_n}{B_n} &
      \Seq{B_1,\dots,B_n, \Gamma}{C}}}        
  \\ \\
  \infer[\botL]{\Seq{\bot,\Gamma}{B}}{\rule{0pt}{6pt}}
  & \infer[\topR]{\Seq{\Gamma}{\top}}{}
  \\ \\
  \infer[\landL, i \in \{1,2\}]{\Seq{B_1 \land B_2,\Gamma}{D}}
  {\Seq{B_i,\Gamma}{D}}
  & 
  \infer[\landR]{\Seq{\Gamma}{B \land C}}
  {\Seq{\Gamma}{B}
    & \Seq{\Gamma}{C}}
  \\ \\
  \infer[\lorL]{\Seq{B \lor C,\Gamma}{D}}
  {\Seq{B,\Gamma}{D}
    & \Seq{C,\Gamma}{D}}
  & 
  \infer[\lorR, i \in \{1,2\}]{\Seq{\Gamma}{B_1 \lor B_2}}
  {\Seq{\Gamma}{B_i}}
  \\ \\
  \infer[\oimpL]{\Seq{B \oimp C,\Gamma}{D}}
  {\Seq{\Gamma}{B}
    & \Seq{C,\Gamma}{D}}
  & \infer[\oimpR]{\Seq{\Gamma}{B \oimp C}}
  {\Seq{B,\Gamma}{C}}
  \\ \\
  \infer[\forallL]{\Seq{\forall x.B\,x,\Gamma}{C}}
  {\Seq{B\,t,\Gamma}{C}}
  & \infer[\forallR]{\Seq{\Gamma}{\forall x.B\,x}}
  {\Seq{\Gamma}{B\,y}}
  \\ \\
  \infer[\existsL]{\Seq{\exists x.B\,x,\Gamma}{C}}
  {\Seq{B\,y,\Gamma}{C}}
  & 
  \infer[\existsR]{\Seq{\Gamma}{\exists x.B\,x}}
  {\Seq{\Gamma}{B\,t}}
\end{array}
$$
\dotfill\\
\paragraph{Equality rules}
$$
\infer[\eqL] {\Seq{s = t, \Gamma}{C}} {
  \{\Seq{\Gamma\rho}{C\rho}~\mid~s\rho =_{\beta\eta} t\rho \} } \qquad
\infer[\eqR] {\Seq{\Gamma}{t = t}} {}
$$
\dotfill\\
\paragraph{Induction rules}
$$
\infer[\indL, p\,\vec x \defmu B\,p\,\vec x] {\Seq{\Gamma,
    p\,\vec{t}}{C}} {\Seq{B\,S \, \vec{y}}{S\,\vec{y}} & \Seq{\Gamma,
    S\,\vec{t}}{C} }
$$
$$
\infer[\indR, p\,\vec x \defmu B\,p\,\vec x]
{\Seq{\Gamma}{p\,\vec{t}}} {\Seq{\Gamma}{B\,X^p\,\vec{t}}}
\qquad\qquad \infer[\indRP, p\,\vec x \defmu B\,p\,\vec x]
{\Seq{\Gamma}{X^p\,\vec{t}}} {\Seq{\Gamma}{B\,X^p\,\vec{t}}}
$$
\dotfill\\[1em]
\paragraph{Co-induction rules}
$$
\infer[\coindL, p\,\vec x \defnu B\,p\,\vec x]
{\Seq{p\,\vec{t}, \Gamma}{C}} {\Seq{B\,X^p\,\vec{t}, \Gamma}{C} }
\qquad\qquad \infer[\coindLP, p\,\vec x \defnu B\,p\,\vec x]
{\Seq{X^p\,\vec{t}, \Gamma}{C}} {\Seq{B\,X^p\,\vec{t}, \Gamma}{C} }
$$
$$
\infer[\coindR, p\,\vec x \defnu B\,p\,\vec x]
{\Seq{\Gamma}{p\,\vec{t}}} {\Seq{\Gamma}{S\,\vec{t}} &
  \Seq{S\,\vec{y}}{B\,S\,\vec{y}} }
$$
\caption{The inference rules of $\Linc$}
\label{fig:linc}
\end{figure}

The logic $\Linc$ shares the core fragment of $\FOLDN$, which is an
intuitionistic version of Church's Simple Theory of Types.  We shall
assume that the reader is familiar with Church's simply typed $\lambda$-calculus
(with both $\beta$ and $\eta$ rules), so we shall recall only the basic
syntax of the calculus here. A simple type is either a \emph{base type}
or a compound type formed using the function-type constructor $\rightarrow$.
Types are ranged over by $\alpha$, $\beta$ and $\tau$.
We assume an infinite set of typed variables, written $x_\alpha$, $y_\beta$, etc.  
The syntax of $\lambda$-terms is given by the following grammar:
$$
s, t ::= x_\tau \mid (\lambda x_\tau.\ t) \mid (s~t)
$$
To simplify presentation, in the following we shall often omit the type 
index in variables and $\lambda$-abstraction. The notion of free and bound
variables are defined as usual.

Following Church, we distinguish a base type $o$ to denote formulae,
and we shall represent formulae as simply typed $\lambda$-terms of
type $o$.  We assume a set of typed constants that correspond to
logical connectives.  The constants $\top : o$ and $\bot : o$ denote
`true' and `false', respectively.  Propositional binary connectives,
i.e., $\land$, $\lor$, and $\oimp$, are assigned the type $o
\rightarrow o \rightarrow o$.  Quantifiers are represented by indexed
families of constants: $\forall_\tau$ and $\exists_\tau$, both are of
type $(\tau \rightarrow o) \rightarrow o$.  We also assume a family of
typed equality symbols $=_\tau : \tau \rightarrow \tau \rightarrow o$.
Although we adopt a representation of formulae as $\lambda$-terms, we
shall use a more traditional notation when writing down formulae. For
example, instead of writing $(\land~A~B)$, we shall use an infix
notation $(A \land B)$.  Similarly, we shall write $\forall_\alpha
x. P$ instead of $\forall_\alpha~(\lambda x_\alpha.P)$. Again, we
shall omit the type annotation when it can be inferred from the
context of the discussion.

The type $\tau$ in quantifiers and the equality predicate are 
restricted to those simple types that do not
contain occurrences of $o$. Hence our logic is essentially first-order, since
we do not allow quantification over predicates. 
As we shall often refer to this kind of restriction to types, we give
the following definition:

\begin{definition}
A simple type $\tau$ is \emph{essentially first-order} (efo) if it
is generated by the following grammar:
$$
  \tau \mathrel{::= }   k \mid \tau \rightarrow \tau\\
$$
where $k$ is a base type other than $o$. 
\end{definition}


For technical reasons (for presenting (co-)inductive proof rules), we
introduce a notion of \emph{parameter} into the syntax of formulae.
Intuitively, they play the role of eigenvariables ranging over the
recursive call in a fixed point expression.  More precisely, to each
predicate symbol $p$, we associate a countably infinite set $\Pscr_p$,
called the \emph{parameter set for $p$}. Elements of $\Pscr_p$ are
ranged over by $X^p$, $Y^p$, $Z^p$, etc, and have the same type as
$p$. 
When we refer to formulae of $\Linc$, we have in mind simply-typed
$\lambda$-terms of type $o$ \emph{in $\beta\eta$-long normal
  form}. Thus formulae of the logic $\Linc$ can be equivalently
defined via the following grammar:
$$
\begin{array}{rcl}
  F &\mathrel{::= } &  X^p\,\vec t \mid s =_\tau t \mid p\,\vec t   \mid 
  \bot \mid \top \mid F \land F \mid F \lor F \mid F \oimp T
  \mid \forall_\tau x. F \mid \exists_\tau x. F.
\end{array}
$$
where $\tau$ is an efo-type. 
We shall omit the type annotation in $s =_\tau t$ when it is not
important to the discussion.

A \emph{substitution} is a type-preserving mapping from variables to terms.  
We assume the usual notion of capture-avoiding substitutions. Substitutions are
ranged over by lower-case Greek letters, e.g., $\theta$, $\rho$ and
$\sigma$. Application of substitution is written in postfix notation,
\eg $t\theta$ denotes the term resulting from an application of
substitution $\theta$ to $t$. Composition of substitutions, denoted by
$\circ$, is defined as $t (\theta \circ \rho) = (t\theta)\rho$.

The whole logic is presented in the sequent calculus in
Figure~\ref{fig:linc}, including rules for equality and fixed points,
as we discuss in Section~\ref{ssec:eq} and~\ref{ssec:coind}. A sequent
is denoted by $\Seq{\Gamma}{C}$ where $C$ is a formula in
$\beta\eta$-long normal form and $\Gamma$ is a multiset of formulae,
also in $\beta\eta$-long normal form.  Notice that in the presentation
of the rule schemes, we make use of HOAS, e.g., in the application
$B\,x$ it is implicit that $B$ has no free occurrence of
$x$. Similarly for the (co)induction rules.  
We work modulo $\alpha$-conversion without further notice.  In the
$\forallR$ and $\existsL$ rules, $y$ is an eigenvariable that is not
free in the lower sequent of the rule.  The $\mc$ rule is a
generalization of the cut rule that simplifies the presentation of the
cut-elimination proof.

Whenever we write a sequent, it is assumed implicitly that the
formulae are well-typed: the type
context, i.e., the types of the constants and the eigenvariables used
in the sequent, is left implicit as well as they can be inferred
from the type annotations of the (eigen)variables.


In some inference rules, reading them bottom up, new eigenvariables
and parameters may be introduced in the premises of the rules, for
instance, in $\existsL$ and $\forallR$, as typical in sequent
calculus.  However, unusually, we shall also allow $\existsR$,
$\forallL$ and $mc$ to possibly introduce new eigenvariables (and new
parameters, in the case of $mc$), again reading the rules bottom-up.
Thus the term $t$ in the premise of the $\existsR$-rule may contain a
free occurrence of an eigenvariable not already occuring in the
conclusion of the rule. The implication of this is that $\exists_\tau
x. \top$ is provable for any type $\tau$; in other words, there is an
implicit assumption that all types are non-empty. Hence the
quantifiers in our setting behave more classically than
intuitionistically.
The reason for this rather awkward treatment of quantifiers is merely
a technical convenience. We could forgo the non-emptiness assumption
on types by augmenting sequents with an explicit signature acting as a
typing environment, and insisting that the term $t$ in $\existsR$ to
be well-formed under the typing environment of the conclusion of the
rule.  However, adding explicit typing contexts into sequents
introduces another layer of bureaucracy in the proof of cut
elimination, which is not especially illuminating. And since our
primary goal is to show the central arguments in cut elimination
involving (co-)induction, we opt to present a slightly simplified
version of the logic so that the main technical arguments (which are
already quite complicated) in the cut elimination proof, related to
(co-)induction rules, can be seen more clearly.  The cut elimination
proof presented in the paper can be adapted to a different
presentation of $\Linc$ with explicit typing contexts;
see~\cite{tiu04phd,Tiu07} for an idea of how such an adaptation may be
done.


We extend the logical fragment with a proof theoretic notion of
equality and fixed points.

\subsection{Equality}
\label{ssec:eq}
The right introduction rule for equality is reflexivity, that is,
it recognizes that two terms are syntactically equal.
The left introduction rule is more interesting.  The substitution
$\rho$ in $\eqL$ is a \emph{unifier} of $s$ and $t$. Note that we
specify the premise of $\eqL$ as a set, with the intention that every
sequent in the set is a premise of the rule.  This set is of course
infinite; every unifier of $(s,t)$ can be extend to another one (e.g.,
by adding substitution pairs for variables not in the terms).
However, in many cases, it is sufficient to consider a particular set
of unifiers, which is often called a \emph{complete set of unifiers
  (CSU)}~\cite{BaaderS01}, from which any unifier can be obtained by
composing a member of the CSU set with a substitution.  In the case
where the terms are first-order terms, or higher-order terms with the
pattern restriction~\cite{Miller91elp}, the set CSU is a singleton,
i.e., there exists a most general unifier (MGU) for the terms.


Our rules for equality actually encompasses the notion of \emph{free
  equality} as commonly found in logic programming, in the form of
Clark's equality theory~\cite{clark78}: injectivity of function
symbols, inequality between distinct function symbols, and the
``occur-check'' follow from rule $\eqL$-rule.  For instance, given a
base type $nt$ (for natural numbers) and the constants $z : nt$ (zero)
and $s : nt \rightarrow nt$ (successor), we can derive $\forall x.\ z
= (s~x) \oimp \bot$ as follows:
$$
\infer[\forallR] {\Seq {} {\forall x.\ z = (s~x) \oimp \bot}} {
  \infer[\oimpR] {\Seq{}{z = (s~y) \oimp \bot}} {\infer[\eqL] {\Seq
      {z = (s~y)} {\bot}} {} } }
$$
Since $z$ and $s~y$ are not unifiable, the $\eqL$ rule above has empty
premise, thus concluding the derivation. A similar derivation
establishes the occur-check property, e.g., $\forall x.\ x = (s~x)
\oimp \bot$. We can also prove the injectivity of the successor
function, \ie $\forall x \forall y. (s~x) = (s~y) \oimp x = y$.

This proof theoretic notion of equality has been considered in several
previous work \eg by Schroeder-Heister~\cite{schroeder-heister93lics},
and McDowell and Miller~\cite{mcdowell00tcs}.

\subsection{Induction and co-induction}
\label{ssec:coind}
One way of adding induction and co-induction to a logic is to
introduce fixed point expressions and their associated introduction
and elimination rules, \ie using the $\mu$ and $\nu$ operators of the
(first-order) $\mu$-calculus. This is essentially what we shall follow
here, but with a different notation. Instead of using a ``nameless''
notation with $\mu$ and $\nu$ to express fixed points, we associate a
fixed point equation with an atomic formula. That is, we associate
certain designated predicates with a \emph{definition}.  This notation
is clearer and more convenient as far as our 
applications are concerned. For a proof system using nameless notation
for (co)inductive predicates, the interested reader is referred to
recent work by Baelde and Miller~\cite{baelde07lpar}.

\begin{definition}
  \label{def:def-clause}
  An \emph{inductive definition clause} is written $\forall\vec{x}.\ p
  \, \vec{x} \defmu B\,\vec{x}$, where $p$ is a predicate constant. 
  The atomic formula $p \, \vec{x}$ is called the \emph{head} of the
  clause, and the formula $B\,\vec{x}$, where $B$ is a closed term
  containing no occurrences of parameters, is called the \emph{body}.
  Similarly, a \emph{ co-inductive definition clause} is written
  $\forall\vec{x}.\ p \, \vec{x} \defnu B\,\vec{x}$.  The symbols
  $\defmu$ and $\defnu$ are used simply to indicate a definition
  clause: they are not a logical connective.  We shall write $\forall
  \vec x.\ p\,\vec x \defeq B\,\vec x$ to denote a definition clause
  generally, i.e., when we are not interested in the details of
  whether it is an inductive or a co-inductive definition. A
  \emph{definition} is a finite set of definition clauses.  A predicate may
  occur only at most once in the heads of the clauses of a definition.
  We shall restrict to \emph{non-mutually recursive} definitions.
  That is, given two  clauses $\forall \vec x.\ p\, \vec x
  \defeq B\,\vec x$ and $\forall \vec y.\ q\, \vec y \defeq C\,\vec y$
  in a definition, where $p \not = q$, if $p$ occurs in $C$ then $q$
  does not occur in $B$, and vice versa.
\end{definition}
Note that the above restriction to non-mutual recursion is immaterial,
since in the first-order case it is well known how one can easily
encode mutually recursive predicates as a single predicate with an
extra argument. The rationale behind that restriction is merely to
simplify the presentation of inference rules and the cut elimination
proof.  Were we  to allow mutually recursive definitions, the
introduction rules $\indL$ and $\coindR$ for a predicate $p$ would
have possibly more than two premises, depending on the number of
predicates which are mutually dependent on $p$
(see~\cite{BrotherstonS07} for a presentation of introduction rules
for mutually dependent definitions).

For technical convenience we also bundle up all the
definitional clause for a given predicate in a single clause, 
following the same principles of the \emph{iff-completion} in logic
programming
. Further, in order to simplify the presentation of rules that involve
predicate substitutions, we denote a definition using an abstraction
over predicates, that is
$$
\forall \vec x.\ p\,\vec x \defeq B\,p\,\vec x
$$
where $B$ is an abstraction with no free occurrence of predicate
symbol $p$ and variables $\vec x$.  Substitution of $p$ in the body of
the clause with a formula $S$ can then be written simply as
$B\,S\,\vec x$.  When writing definition clauses, we often omit the
outermost universal quantifiers, with the assumption that free
variables in a clause are universally quantified. For example even
numbers are defined as follows:
$$ev~x \defmu (x = z) \lor (\exists y.\ x = (s~(s~y)) \land ev~y)
$$
where in this case $B$ is of the form $\lambda p w.\ (w = z) \lor
(\exists y. w = (s~(s~y)) \land p~y)$.


The left and right rules for (co-)inductively defined atoms are given
at the bottom of Figure~\ref{fig:linc}.  In rules $\indL$ and
$\coindR$, the abstraction $S$ is an invariant of the (co-)induction
rule. 
The variables $\vec{y}$ are
new eigenvariables and $X^p$ is a new parameter not already occuring
in the lower sequent.  For the induction rule $\indL$, $S$ denotes a
pre-fixed point of the underlying fixed point operator. Similarly, for
the co-induction rule $\coindR$, $S$ can be seen as denoting a
{post-fixed point} of the same operator.  Here, we use a
characterization of induction and co-induction proof rules as,
respectively, the least and the greatest solutions to a fixed point
equation.


Notice that the right-introduction rules for inductive predicates and
parameters (dually, the left-introduction rules for co-inductive
predicates and parameters) are slightly different from the
corresponding rules in $\fLinc$-like logics (see
Remark~\ref{rem:unf}).  These rules can be better understood by the
usual interpretation of (co-)inductive definitions in second-order
logic~\cite{Pfenning89mfps,Paulson97} (to simplify presentation, we
show only the propositional case here):
$$
p \defmu B\,p \qquad \leadsto \qquad \forall p. (B\,p \oimp p) \oimp p
$$
$$
p \defnu B\,p \qquad \leadsto \qquad \exists p. p \land (p \oimp
B\,p).
$$
Then the right-introduction rule for inductively defined predicate
will involve an implicit universal quantification over predicates.  As
standard in sequent calculus, such a universal quantified predicate
will be replaced by a new eigenvariable (in this case, a new
parameter), reading the rule bottom up.  Note that if we were to
follow the above second-order interpretation literally, an
alternative rule for inductive predicates could be:
$$
\infer[\indR, p \defmu B\,p] {\Seq{\Gamma}{p} } {\Seq{B\,X^p \oimp
    X^p, \Gamma}{X^p} }
$$
Then there would be no need to add the $\indRP$-rule since it would be
derivable, using the clause $B\,X^p \oimp X^p$ in the left hand side of
the sequent. (This, of course, is true only when such an $\indRP$
instance appears above an $\indR$ instance for $p$.)  Our presentation
has the advantage that it simplifies the cut-elimination arguments in
the subsequent sections.  The left-introduction rule for
co-inductively defined predicate can be explained dually.

A similar encoding of (co-)inductive definitions as second-order
formulae is used in~\cite{baelde07lpar}, where cut-elimination is
indirectly proved  by appealing to a \emph{focused} proof system for
higher-order linear logic. A similar approach can be followed for
$\Linc$, but we prefer to develop a direct cut-elimination proof,
since such a proof can serve as the basis of cut-elimination for
extensions of $\Linc$, for example, with the
$\nabla$-quantifier~\cite{miller05tocl,gacek08lics}.

\begin{remark}[Fixed point unfolding]

\label{rem:unf}
  A commonly used form of introduction rules for definitions, or fixed
  points, uses an unfolding of the definitions.  This form of rules is
  followed in several related logics, e.g.,
  $\FOLDN$~\cite{mcdowell00tcs},
  $\fLinc$~\cite{Momigliano03TYPES,tiu04phd} and
  $\mu$-MALL~\cite{baelde07lpar}.  The right-introduction rule for
  inductive definitions, for instance, takes the form:
$$
\infer[\indR', p\,\vec x \defmu B\,p\,\vec x] {\Seq{\Gamma}{p\,\vec
    t}} {\Seq{\Gamma}{B\,p\,\vec t} }
$$
That is, in the premise, the predicate $p$ is replaced with the body
of the definition.  The logic $\fLinc$, like $\FOLDN$, imposes a
stratification on definitions, which amounts to a strict
positivity condition: the head of a definition can only appear in a
strictly positive position in the body, i.e., it never appears to the
left of an implication. Let us call such a definition a \emph{stratified definition}.  
For stratified definitions, the rule $\indR'$ can be derived as follows:
$$
\infer[mc] {\Seq{\Gamma}{p\,\vec t}} {\Seq{\Gamma}{B\,p\,\vec t} &
  \infer[\indR] {\Seq{B\,p\,\vec t}{p\,\vec t}} {\infer[]
    {\Seq{B\,p\,\vec t}{B\,X^p\,\vec t}} {\deduce{\vdots} {
        \infer[\indL] {\Seq{p\,\vec u}{X^p\,\vec u}} {\infer[\indRP]
          {\Seq{B\,X^p\,\vec x}{X^p\,\vec x}} {\infer[init]
            {\Seq{B\,X^p\,\vec x}{B\,X^p\,\vec x}} {} } & \infer[init]
          {\Seq{X^p\,\vec u}{X^p\,\vec u}}{} } } } } }
$$
where the `dots' are a derivation composed using left and right
introduction rules for logical connectives in $B$.  Notice that all
leaves of the form $\Seq{p\,\vec u}{X^p\,\vec u}$ can be proved by
using the $\indL$ rule, with $X^p$ as the inductive invariant.
Conversely, given a stratified definition, any proof in
$\Linc$ using that definition can be transformed into a proof of
$\fLinc$ simply by replacing $X^p$ with $p$.  Note that once $\indR'$
is shown admissible, one can also prove admissibility of unfolding of
inductive definitions on the left of a sequent; see~\cite{tiu04phd}
for a proof.
\end{remark}

Since a defined atomic formula can be unfolded via its introduction rules, 
the notion of size of a formula as simply the number of connectives
in it would not take into account this possible unfolding. We shall
define a more general notion 
 assigning a positive integer to each predicate symbol, which we refer 
to as its \emph{level}. A similar notion of level of a predicate was introduced for
$\FOLDN$~\cite{mcdowell00tcs}.  However, in $\FOLDN$, the level of
a predicate is only used to guarantee monotonicity of definitions.

%
\begin{definition}[Size of formulae]
  \label{def:level}
  To each predicate $p$ we associate a natural number $\level{p}$, the
  \emph{level} of $p$.  Given a formula $B$, its \emph{size} $|B|$ is
  defined as follows:
  \begin{enumerate}
  \item $|X^p\,\vec t| = 1$, for any $X^p$ and any $\vec t$. 
  \item $|p \, \vec{t}| = \level{p}$. 
  \item $|\bot| = |\top| = |(s = t)| = 1$. 
  \item $|B \land C| = |B \lor C| = |B \oimp C| = |B| + |C| + 1$. 
  \item $|\forall x.\ B\,x| = |\exists x.\ B\,x| = |B\,x| + 1$. 
  \end{enumerate}
\end{definition}
%
%

Note that in this definition, we do not specify precisely
any particular level assignment to predicates. We show next that
there is a level assignment that has a property that will be useful
later in proving cut elimination. 

\begin{lemma}[Level assignment]
\label{lm:level}
Given any definition $\Dscr$, there is a level assignment
to every predicate $p$ occuring in $\Dscr$ such that
if $\forall \vec x. p\,\vec x \defeq B\,p\,\vec x$ is in $\Dscr$,
then $|p\,\vec x| >  |B\,X^p\,\vec x|$ for every parameter $X^p \in \Pscr_p$. 
\end{lemma}
\begin{proof}
  Let $\prec$ be a binary relation on predicate symbols defined as
  follows: $ q \prec p $ iff $q$ occurs in the body of the definition
  clause for $p$. Let $\prec^*$ be the reflexive-transitive closure of
  $\prec$.  Since we restrict to non-mutually recursive definitions
  and there are only finitely many definition clauses
  (Definition~\ref{def:def-clause}), it follows that $\prec^*$ is a
  well-founded partial order.  We now compute a level assignment to
  predicate symbols by induction on $\prec^*$. This is simply done by
  letting $\level p = 1$, if $p$ is undefined, and $\level p =
  |B\,X^p\,\vec x| + 1$, for some parameter $X^p$, if $\forall \vec
  x.\ p\,\vec x \defeq B\,p\,\vec x.$ Note that in the latter case, by
  induction hypothesis, every predicate symbol $q$, other than $p$, in
  $B$ has already been assigned a level, so $|B\,X^p\,\vec x|$ is
  already defined at this stage.  Note also that it does not matter
  which $X^p$ we choose since all parameters have the same size.  \qed
\end{proof}

We shall assume from now on that predicates are assigned levels
satisfying the condition of Lemma~\ref{lm:level}, so whenever we have
a definition clause of the form $\forall \vec x.p\,\vec x \defeq
B\,p\,\vec x$, we shall implicitly assume that $|p\,\vec x| >
|B\,X^p\,\vec x|$ for every parameter $X^p \in \Pscr_p.$



\begin{remark}[Non-monotonicity]
  \label{rem:strat}
  In $\FOLDN$, a notion of stratification is used to rule out
  non-monotone (or in Haln\"{a}s' terminology
  \emph{partial}~\cite{PID}) definitions, such as, $p \defeq p \oimp
  \bot$, for which cut-elimination is problematic. \footnote{Other
    ways beyond stratification of recovering cut-elimination are
    disallowing contraction or restricting to an \textit{init} rule
    for undefined atoms.}  
  In fact, from the above definition both $p$ and $p \oimp \bot$ are
  provable, but there is no direct proof of $\bot$. This can be traced
  back to the fact that unfolding of definitions in $\fLinc$ and
  $\FOLDN$ is allowed on both the left and the right hand side of
  sequent. 
  In $\Linc$, inconsistency does not arise even allowing a non-monotone
  definition as above, due to the fact that arbitrary unfolding of
  fixed points is not permitted. 
  Instead, only a limited form of unfolding is allowed, i.e., in the
  form of unfolding of inductive parameters on the right, and
  co-inductive parameters on the left.  As a consequence of this
  restrictive unfolding, in $\Linc$ one cannot reason about some
  well-founded inductive definitions which are not 
  stratified. For example, consider the non-stratified definition:
$$
\forall x.\ ev~x \defmu (x = z) \lor (\exists y. x = (s~y) \land (ev~y
\oimp \bot))
$$
If this definition were to be interpreted as a logic program (with
negation-as-failure), for example, then its least fixed point is
exactly the set of even natural numbers. However, the above encoding
in $\Linc$ is incomplete with respect to this interpretation, since
not all even natural numbers can be derived using the above
definition.  For example, it is easy to see that $ev~(s~(s~z))$ is not
derivable, since this would require a derivation of
$\Seq{X^{ev}\,(s~z)} \bot$, for some inductive parameter $X^{ev}$,
which is impossible because no unfolding of inductive parameter is
allowed on the left of a sequent. The same idea prevents the
derivability of $\Seq{} p$ given the definition $p \defeq p \oimp
\bot$. So while inconsistency in the presence of non-monotone
definitions is avoided in $\Linc$, its reasoning power does not extend
that of $\fLinc$ significantly.
\end{remark}

\section{Eigenvariables and parameters instantiations}
\label{sec:drv}

We now discuss some properties of derivations in $\Linc$ which involve
instantiations of eigenvariables and parameters. These properties will
be used in the cut-elimination proof in subsequent sections.  

Before we proceed, it will be useful to introduce the following 
derived rule in $\Linc$: 
$$
\infer[subst.]
{\Seq{\Gamma}{C}}
{
\{ \Seq{\Gamma\theta}{C\theta}\}_\theta
}
$$
This rule is just a `macro' for the following derivation:
$$
\infer[mc]
{\Seq{\Gamma}{C}}
{
\infer[\eqR]
{\Seq{}{t=t}}{}
&
\infer[\eqL]
 {\Seq{t=t,\Gamma}{C}}
 {
  \{\Seq{\Gamma\theta}{C\theta}\}_\theta
 }
}
$$
where $t$ is some arbitrary term. The motivation behind the 
rule $subst$ is purely technical; it allows us to prove that
a derivation transformation (i.e., substitutions of eigenvariables
in derivations in Section~\ref{sec:subst}) commutes with cut reduction (see Lemma~\ref{lm:reduct_subst}).
Since the rule $subst$ hides a simple form of cut, to prove cut-elimination
of $\Linc$, we have to show that $subst$, in addition to $mc$, 
is admissible. In the following, $\epsilon$ denotes the identity substitution,
i.e., $\epsilon(x)=x$ for every variable $x$.

\begin{lemma}[$subst$-elimination]
\label{lm:subst-elimination}
For every $\Gamma$ and $C$, if the sequent $\Seq{\Gamma}{C}$
is (cut-free) derivable in $\Linc$ with $subst$ then it is (cut-free)
derivable in $\Linc$ without $subst$.
\end{lemma}
\begin{proof}
Given a derivation $\Pi$ of $\Seq{\Gamma}{C}$ with occurrences of $subst$,
obtain a $subst$-free derivation by simply replacing 
any subderivation in $\Pi$ of the form:
$$
\infer[subst]
{\Seq{\Delta}{B}}
{
\left\{
\raisebox{-1.3ex}{
\deduce{\Seq{\Delta\theta}{B\theta}}{\Pi^\theta}
}
\right\}_\theta
}
$$
with its premise $\Pi^\epsilon$.
\qed
\end{proof}

Following~\cite{mcdowell00tcs}, we define a \emph{measure} which
corresponds to the height of a derivation:
\begin{definition}
  \label{def:mu}
  Given a derivation $\Pi$ with premise derivations $\{\Pi_i\}_{i \in
    I}$, for some index set $I$, the measure $\measure{\Pi}$ is the
  least upper bound $\lub {\{\measure{\Pi_i}\}_i\in I} + 1$.
\end{definition}
Note that given the possible infinite branching of $\eqL$ rule, these
measures can in general be (countable) ordinals.  Therefore proofs and
definitions on those measures require transfinite induction and
recursion. However, in most of the proofs to follow, we do case
analysis on the last rule of a derivation. In such a situation, the
inductive cases for both successor and limit ordinals are basically
covered by the case analysis on the inference figures involved, and we
shall not make explicit use of transfinite principles.

With respect to the use of eigenvariables and parameters in a
derivation, there may be occurrences of the formers that are not free
in the end sequent. We refer to these variables and parameters as the
\emph{internal variables and parameters}, respectively.  We view the
choices of those variables and parameters as arbitrary and therefore
identify derivations which differ on the choice of internal variables
and parameters. In other terms, we quotient derivations modulo
injective renaming of internal eigenvariables and parameters.

\subsection{Instantiating eigenvariables}
\label{sec:subst}

The following definition extends eigenvariable substitutions to apply
to derivations.  Since we identify derivations that differ only in the
choice of internal eigenvariables, we will assume that such variables
are chosen to be distinct from the variables in the domain of the
substitution and from the free variables of the range of the
substitution.  Thus applying a substitution to a derivation will only
affect the variables free in the end-sequent.


\begin{definition}
  \label{def:subst}
  If $\Pi$ is a derivation of $\Seq{\Gamma}{C}$ and $\theta$ is a
  substitution, then we define the derivation $\Pi\theta$ of
  $\Seq{\Gamma\theta}{C\theta}$ as follows:
  \begin{enumerate}
  \item Suppose $\Pi$ ends with the $\eqL$ rule
    \begin{displaymath}
      \infer[\eqL]{\Seq {s = t,\Gamma'}{C} }
      {\left\{\raisebox{-1.5ex}
          {\deduce{\Seq{\Gamma'\rho}{C\rho}}
            {\Pi^{\rho}}}
        \right\}_{\rho}}
      \enspace 
    \end{displaymath}
    where each $\rho$ satisfies $s\rho =_{\beta\eta} t\rho$. Observe
    that any unifier for the pair $(s\theta, t\theta)$ can be
    transformed to another unifier for $(s, t)$, by composing the
    unifier with $\theta$.  Thus $\Pi\theta$ is
    \begin{displaymath}
      \infer[\eqL]{\Seq{s\theta = t\theta,\Gamma'\theta}{C\theta}}
      {\left\{\raisebox{-1.5ex}
          {\deduce{\Seq{\Gamma'\theta\rho'}{C\theta\rho'}}
            {\Pi^{\theta\circ\rho'}}}
	\right\}_{\rho'}}
      \enspace ,
    \end{displaymath}
    where $s\theta\rho' =_{\beta\eta} t\theta\rho'$.

  \item If $\Pi$ ends with $subst$ with premise derivations
    $\{\Pi^\rho\}_\rho$ then $\Pi\theta$ also ends with the same rule
    and has premise derivations $\{\Pi^{\theta \circ \rho'}
    \}_{\rho'}$.
  \item If $\Pi$ ends with any other rule and has premise derivations
    $\Pi_1, \ldots, \Pi_n$, then $\Pi\theta$ also ends with the same
    rule and has premise derivations $\Pi_1\theta, \ldots,
    \Pi_n\theta$.
  \end{enumerate}
\end{definition}

Among the premises of the inference rules of $\Linc$ (with the exception
of $\coindR$), certain premises share the same right-hand side formula 
with the sequent in the conclusion.  We refer to such
premises as major premises. This notion of major premise will be
useful in proving cut-elimination, as certain proof transformations
involve only major premises.
\begin{definition}
  \label{def:major-premise}
  Given an inference rule $R$ with one or more premise sequents, we
  define its major premise sequents as follows.
  \begin{enumerate}
  \item If $R$ is either $\oimpL, \mc$ or $\indL$, then its rightmost
    premise is the major premise
  \item If $R$ is $\coindR$ then its left premise is the major
    premise.
  \item Otherwise, all the premises of $R$ are major premises.
  \end{enumerate}
  A \emph{minor premise} of a rule $R$ is a premise of $R$ which is
  not a major premise.  The definition extends to derivations by
  replacing premise sequents with premise derivations.
\end{definition}

The proofs of the following two lemma are straightforward from Definition~\ref{def:subst}
and induction on the height of derivations.
\begin{lemma}
  \label{lm:subst}
  For any substitution $\theta$ and derivation $\Pi$ of
  $\Seq{\Gamma}{C}$, $\Pi\theta$ is a derivation of
  $\Seq{\Gamma\theta}{C\theta}$.
\end{lemma}

\begin{lemma}
  \label{lm:subst-height}
  For any derivation $\Pi$ and substitution $\theta$, $\measure{\Pi}
  \geq \measure{\Pi\theta}$.
\end{lemma}

\begin{lemma}
  \label{lm:subst-drv-comp}
  For any derivation $\Pi$ and substitutions $\theta$ and $\rho$, the
  derivations $(\Pi\theta)\rho$ and $\Pi(\theta \circ \rho)$ are the
  same derivation.
\end{lemma}


\subsection{Instantiating parameters}
\label{sec:unfolding}

\begin{definition}
  \label{def:param subst}
  A \emph{parameter substitution} $\Theta$ is a partial map from
  parameters to pairs of proofs and \emph{closed} terms such that  whenever
$$
\Theta(X^p) = (\Pi_S, S)
$$ 
then $S$ has the same type as $p$ and either one of the following
holds:
\begin{itemize}
\item $p\,\vec x \defmu B\,p\,\vec x$, for some $B$ and $\vec x$, and
  $\Pi_S$ is a derivation of $\Seq {B\,S\,\vec x}{S\,\vec x}$,
  or
\item $p\,\vec x \defnu B\,p\,\vec x$, for some $B$ and $\vec x$, and
  $\Pi_S$ is a derivation of $\Seq {S\,\vec x}{B\,S\,\vec x}$.
\end{itemize}
The \emph{support} of $\Theta$ is the set
$$
supp(\Theta) = \{X^p \mid \Theta(X^p) \hbox{ is defined} \}.
$$
We consider only parameter substitutions with finite support.

We say that $X^p$ is \emph{fresh for $\Theta$}, written $X^p \#
\Theta$, if for each $Y^q \in supp(\Theta)$, $X^p \not = Y^q$ and
$X^p$ does not occur in $S$ whenever $\Theta(Y^q) = (\Pi_S,S)$.

\end{definition}

We shall often enumerate a parameter substitution using a similar
notation to (eigenvariables) substitution, e.g.,
$$
[(\Pi_1,S_1)/X^{p_1}, \ldots, (\Pi_n, S_n)/ X^{p_n}]
$$
denotes a parameter substitution $\Theta$ with support $\{X^{p_1},
\ldots, X^{p_n}\}$ and $\Theta(X^{p_i}) = (\Pi_i, S_i)$.

Given a formula $C$ and a parameter substitution $\Theta$ as above, we
write $C\Theta$ to denote the formula
$$
C[S_1/X^{p_1}, \ldots, S_n/X^{p_n}].
$$

\begin{definition}
  Let $\Pi$ be a derivation of $\Seq{\Gamma}{C}$ and let $\Theta$ be a
  parameter substitution.  Define the derivation $\Pi\Theta$ of
  $\Seq{\Gamma\Theta}{\Theta}$ by induction on the height of $\Pi$
  as follows:
  \begin{itemize}
  \item Suppose $C = X^p\,\vec t$ for some $X^p$ such that $\Theta(X^p) =
(\Pi_S, S)$ and $\Pi$ ends with $\indRP$, as shown below left. 
Then $\Pi\Theta$ is as shown below right. 
$$
\infer[\indRP] 
{\Seq \Gamma {X^p \vec t}} 
{\deduce{\Seq \Gamma {B\,X^p \vec t}}{\Pi'}}
\qquad
\infer[mc] {\Seq {\Gamma\Theta}{S\,\vec t}} { \deduce{\Seq
    {\Gamma\Theta}{B\,S\,\vec t}}{\Pi'\Theta} & \deduce{\Seq
    {B\,S\,\vec t}{S\,\vec t}}{\Pi_S[\vec t/\vec x]} }
$$


\item 
Similarly, suppose  $\Pi$ ends with $\coindLP$ on $X^p\,\vec t$
and $X^p \in supp(\Theta)$:
$$
\infer[\coindLP] {\Seq {X^p \vec t, \Gamma'}{C}} {
  \deduce{\Seq{B\,X^p\vec t, \Gamma'}{C}}{\Pi'} }
$$
where $p\,\vec x \defnu B\,p\,\vec x$ and $\Theta(X^p) = (\Pi_S, S)$.  
Then $\Pi\Theta$ is
$$
\infer[mc] {\Seq{S\,\vec t, \Gamma'\Theta}{C\Theta}} { \infer[mc]
  {\Seq {S\,\vec t}{B\,S\,\vec t}} { \infer[init] {\Seq {S\,\vec
        t}{S\,\vec t}}{} & \deduce{\Seq {S\,\vec t}{B\,S\,\vec
        t}}{\Pi_S[\vec t/\vec x]} } & \deduce{\Seq{B\,S\,\vec t,
      \Gamma'\Theta}{C\Theta}}{\Pi'\Theta} }
$$ 


\item In all other cases, suppose $\Pi$ ends with a rule $R$ with
  premise derivations $\{\Pi_i \}_{i \in I}$ for some index set $I$.
  Since we identify derivations up to renaming of internal parameters,
  we assume without loss of generality that the internal
  eigenvariables in the premises of $R$ (if any) do not appear in
  $\Theta$.  Then $\Pi\Theta$ ends with the same rule, with premise
  derivations $\{\Pi_i\Theta \}_{i \in I}$.
\end{itemize}

\end{definition}

\begin{remark}
  Notice that the definition of application of parameter substitution
  in derivations in Definition~\ref{def:param subst} is asymmetric in
  the treatment of inductive and co-inductive parameters, i.e., in the
  cases where $\Pi$ ends with $\indRP$ and $\coindLP$.  In the latter
  case, the substituted derivation uses a seemingly unnecessary cut
$$
\infer[mc.]  {\Seq {S\,\vec t}{B\,S\,\vec t}} { \infer[init] {\Seq
    {S\,\vec t}{S\,\vec t}}{} & \deduce{\Seq {S\,\vec t}{B\,S\,\vec
      t}}{\Pi_S[\vec t/\vec x]} }
$$
The reason behind this is rather technical; in our main cut
elimination proof, we need to establish that $\Pi_S[\vec t/\vec x]$ is
``reducible'' (i.e., all the cuts in it can be eventually eliminated),
given that the above cut is reducible. In a typical cut elimination
procedure, say Gentzen's proof for LK, one would have expected that
the above cut reduces to $\Pi_S[\vec t/\vec x]$, hence reducibility of
$\Pi_S$ would follow from reducibility of the above cut.  However,
according to our cut reduction rules (see Section~\ref{sec:reduc}),
the above cut does not necessarily reduce to $\Pi_S[\vec t/\vec x]$.
However, if the instance of $init$ appears instead on the right
premise of the cut, e.g., as in
$$
\infer[mc] {\Seq {B\,S\,\vec t}{S\,\vec t}} { \deduce{\Seq {B\,S\,\vec
      t}{S\,\vec t}}{\Pi_S[\vec t/\vec x]} & \infer[init] {\Seq
    {S\,\vec t}{S\,\vec t}}{} }
$$
the cut elimination procedure does reduce this to $\Pi_S[\vec t/\vec
x]$, so it is not necessary to introduce explicitly this cut instance
in the case involving inductive parameters.  It is possible to define
a symmetric notion of parameter substitution, but that would require
different cut reduction rules than the ones we proposed in this
paper. Another possibility would be to push the asymmetry to the
definition of \emph{reducibility} (see Section~\ref{sec:cut-elim}). We
have explored these alternative options, but for the purpose of
proving cut elimination, we found that the current definition yields a
simpler proof.\footnote{ But we conjecture that in the classical case
  a fully symmetric definition of parameter substitution and cut
  reduction would be needed. But this is outside the scope of the
  current paper.}
\end{remark}

The following lemma states that the derivation $\Pi\Theta$ is
well-formed.

\begin{lemma}
  \label{lm:param subst}
  Let $\Theta$ be a parameter substitution and $\Pi$  a
  derivation of $\Seq \Gamma C$. Then $\Pi\Theta$ is a derivation of
  $\Seq {\Gamma\Theta}{C\Theta}$.
\end{lemma}

Note that since parameter substitutions replace parameters with closed
terms, they commute with (eigenvariable) substitutions.

\begin{lemma}
  \label{lm:param subst commutes}
  For every derivation $\Pi$, substitution $\delta$, parameter
  substitution $\Theta$, the derivation $(\Pi\Theta)\delta$ is the
  same as the derivation $(\Pi\delta)\Theta$.
\end{lemma}

In the following, we denote with $[\Theta, (\Pi_S,S)/X^p]$, where $X^p
\# \Theta$, a parameter substitution obtained by extending $\Theta$
with the map $X^p \mapsto (\Pi_S,S)$.

\begin{lemma}
  \label{lm:param subst vacuous}
  Let $\Pi$ be a derivation of $\Seq \Gamma C$, $\Theta$  a
  parameter substitution and  $X^p$  a parameter such that $X^p
  \not \in supp(\Theta)$ and $X^p$ does not occur in $\Seq \Gamma C$.
  Then $\Pi[\Theta, (\Pi_S,S)/X^p] = \Pi\Theta$ for every $\Pi_S$ and
  $S$.
\end{lemma}

\section{Cut elimination for $\Linc$}
\label{sec:cut-elim}

The central result of our work is cut-elimination, from which
consistency of the logic follows.  Gentzen's classic proof of
cut-elimination for first-order logic uses an induction on the size of
the cut formula
. The cut-elimination procedure consists of a set of reduction
rules that reduces a cut of a compound formula to cuts on its
sub-formulae of smaller size.  In the case of $\Linc$, the use of
induction/co-induction complicates the reduction of cuts.  Consider
for example a cut involving the induction rules:
$$
\infer[\mc] {\Seq{\Delta, \Gamma}{C} } { \infer[\indR]
  {\Seq{\Delta}{p\,\vec t}} { \deduce{\Seq{\Delta}{B\,X^p\,\vec t}}{\Pi_1} } &
  \infer[\indL] {\Seq{p\,\vec t, \Gamma}{C}} {
    \deduce{\Seq{B\,S\,\vec y}{S\,\vec y}}{\Pi_B} & \deduce{\Seq{S\,\vec t,
        \Gamma}{C}}{\Pi} } }
$$
There are at least two problems in reducing this cut. First, any
permutation upwards of the cut will necessarily involve a cut with $S$
that can be of larger size than $p$, and hence a simple induction on
the size of the cut formula will not work.  Second, the invariant $S$ does
not appear in the conclusion of the left premise of the cut. The
latter means that we need to transform the left premise so that its
end sequent will agree with the right premise. Any such transformation
will most likely be \emph{global}, and hence simple induction on the
height of derivations will not work either.

We shall use the \emph{reducibility} technique to prove cut
elimination.  More specifically, we shall build on the notion of
reducibility introduced by Martin-L\"of to prove normalization of an
intuitionistic logic with iterative inductive definition
\cite{martin-lof71sls}.  Martin-L\"of's proof has been adapted to
sequent calculus by McDowell and Miller~\cite{mcdowell00tcs}, but in a
restricted setting where only natural number induction is
allowed. Since our logic involves arbitrary stratified inductive
definitions, which also includes iterative inductive definitions, we
shall need different, and more general, cut reductions. But the real
difficulty in our case is in establishing cut elimination in the
presence of co-inductive definitions, for which there is no known
direct cut elimination proof (prior to our work~\cite{Momigliano03TYPES} on which
this article is based on), at the best of our knowledge, as far as
the sequent calculus is concerned.

The main part of the reducibility technique is a definition of the
family of reducible sets of derivations.  In Martin-L\"of's theory of
iterative inductive definition, this family of sets is defined
inductively by the ``type'' of the derivations they contain, i.e., the
formula in the right-hand side of the end sequent in a
derivation. Extending this definition of reducibility to $\Linc$ is
not obvious.  In particular, in establishing the reducibility of a
derivation 
of type $p\,\vec t$ ending with a $\coindR$ rule
one must first establish the reducibility of its premise derivations,
which may have larger types, since $S\,\vec t$ could be any formula.
Therefore a simple inductive definition based on types of derivations
would not be well-founded.

The key to properly ``stratify'' the definition of reducibility is to
consider reducibility under parameter substitutions.  This notion of
reducibility, called \emph{parametric reducibility}, was originally
developed by Girard to prove strong normalisation of
System~F, i.e., in the interpretation of universal
types.  As with strong normalisation of System F, (co-)inductive
parameters are substituted with some ``reducibility candidates'',
which in our case are certain sets of derivations satisfying closure
conditions similar to those for System F, but which additionally
satisfy certain closure conditions related to (co-)inductive
definitions.

The remainder of this section is structured as follows.  In
Section~\ref{sec:reduc} we define a set of cut reduction rules that
are used to elimination the applications of the cut rule.  For the
cases involving logical operators, the cut-reduction rules used to
prove the cut-elimination for $\Linc$ are the same as those of
$\FOLDN$~\cite{mcdowell00tcs}.  The crucial differences are, of
course, in the reduction rules involving induction and co-induction
rules, where we use the transformation described in
Definition~\ref{def:param subst}.  We then proceed to define two
notions essential to our cut elimination proof: \emph{normalizability}
(Section~\ref{sec:norm}) and \emph{parametric reducibility}
(Section~\ref{sec:red}).  These can be seen as counterparts for
Martin-L\"of's notions of normalizability and
\emph{computability}~\cite{martin-lof71sls}, respectively.
Normalizability of a derivation implies that all the cuts in it can be
eventually eliminated (via the cut reduction rules defined earlier).
Reducibility is a stronger notion, in that it implies normalizability.
The main part of the cut elimination proof is presented in
Section~\ref{sec:ceproof}, where we show that every derivation is
reducible, hence it can be turned into a cut-free derivation.


\subsection{Cut reduction}
\label{sec:reduc}

We now define a reduction relation on derivations ending with $mc$.
This reduction relation is an extension of the similar cut reduction
relation used in McDowell and Miller's cut elimination
proof~\cite{mcdowell00tcs}.  In particular, the reduction rules
involving introduction rules for logical connectives are the same. The
main differences are, of course, in the reduction rules involving
induction and co-induction rules.  There is also slight difference in
one reduction rule involving equality, which in our case utilises the
derived rule $subst$.  Therefore in the following definition, we shall
highlight only those reductions that involve (co-)induction and
equality rules. The complete list of reduction rules can be found in
Appendix~\ref{app:reduc}.

To ease presentation, we shall use the following notations to
denote certain forms of derivations.  The derivation
$$
\infer[mc] {\Seq{\Delta_1, \ldots, \Delta_n, \Gamma }{C}} {
  \deduce{\Seq {\Delta_1}{B_1}}{\Pi_1} & \cdots & \deduce{\Seq
    {\Delta_n}{B_n}}{\Pi_n} & \deduce{\Seq {\Gamma}{C}}{\Pi} }
\enspace
$$
is abbreviated as $mc(\Pi_1, \ldots, \Pi_n, \Pi )$. Whenever we write
$mc(\Pi_1,\ldots,\Pi_n,\Pi)$ we assume implicitly that the derivation
is well-formed, i.e., $\Pi$ is a derivation ending with some
sequent $\Seq \Gamma C$ and the right-hand side of the end sequent of each $\Pi_i$ 
is a formula $F \in \Gamma$.
Similarly, we abbreviated as $\idrv_B$ the derivation 
$$
\infer[init]
{\Seq {B}{B}}{}
$$
and $subst(\{\Pi^\theta\}_\theta)$ denotes a derivation ending with
the rule $subst$ with premise derivations $\{\Pi^\theta\}_\theta$.

\begin{definition}
  \label{def:reduct}
  We define a \emph{reduction} relation between derivations.  The redex
  is always a derivation $\Xi$ ending with the multicut rule
  \begin{displaymath}
    \infer[\mc]{\Seq{\Delta_1,\ldots,\Delta_n,\Gamma}{C}}
    {\deduce{\Seq{\Delta_1}{B_1}}
      {\Pi_1}
      & \cdots
      & \deduce{\Seq{\Delta_n}{B_n}}
      {\Pi_n}
      & \deduce{\Seq{B_1,\ldots,B_n,\Gamma}{C}}
      {\Pi}}
    \enspace 
  \end{displaymath}
  We refer to the formulas $B_1,\dots,B_n$ produced by the $\mc$ as
  \emph{cut formulas}.

  If $n=0$, $\Xi$ reduces to the premise derivation $\Pi$.  For $n >
  0$ we specify the reduction relation based on the last rule of the
  premise derivations.  If the rightmost premise derivation $\Pi$ ends
  with a left rule acting on a cut formula $B_i$, then the last rule
  of $\Pi_i$ and the last rule of $\Pi$ together determine the
  reduction rules that apply.  Following McDowell and
  Miller~\cite{mcdowell00tcs}, we classify these rules according to
  the following criteria: we call the rule an \emph{essential} case
  when $\Pi_i$ ends with a right rule; if it ends with a left rule or $subst$, it
  is a \emph{left-commutative} case; if $\Pi_i$ ends with the $\init$
  rule, then we have an \emph{axiom} case; a \emph{multicut} case
  arises when it ends with the $\mc$ rule.  When $\Pi$ does not end
  with a left rule acting on a cut formula, then its last rule is
  alone sufficient to determine the reduction rules that apply.  If
  $\Pi$ ends with $subst$ or a rule acting on a formula other than a cut formula,
  then we call this a \emph{right-commutative} case.  A
  \emph{structural} case results when $\Pi$ ends with a contraction or
  weakening on a cut formula.  If $\Pi$ ends with the $\init$ rule,
  this is also an axiom case; similarly a multicut case arises if
  $\Pi$ ends in the $\mc$ rule.
  For simplicity of presentation, we always show $i = 1$.

  We show here  the cases involving (co-)induction rules. 

\paragraph{Essential cases:}

\begin{trivlist}

\item[\fbox{$\eqL/\eqR$}]
Suppose $\Pi_1$ and $\Pi$ are
  \begin{displaymath}
    \infer[\eqR]{\Seq{\Delta_1}{s = t}}
    {}
    \qquad\qquad\qquad
    \infer[\eqL]{\Seq{s=t,B_2,\ldots,B_n,\Gamma}{C}}
    {\left\{\raisebox{-1.5ex}
        {\deduce{\Seq{B_2\rho,\ldots,B_n\rho,\Gamma\rho}
            {C\rho}}
          {\Pi^\rho}}
      \right\}_\rho}
    \enspace 
  \end{displaymath}
  Note that in this case, $\rho$ in $\Pi$ ranges over all substitution, as any
  substitution is a unifier of $s$ and $t$. 
  Let $\Xi_1$ be the derivation 
  $mc(\Pi_2,\ldots,\Pi_n,subst(\{\Pi^\rho\}_\rho$.
  In this case, $\Xi$ reduces to
  $$
  \infer=[\wL]
  {\Seq{\Delta_1,\Delta_2,\ldots,\Delta_n,\Gamma}{C}}
  {\deduce{\Seq{\Delta_2,\ldots,\Delta_n,\Gamma}{C}}{\Xi_1}}
  $$
  We use the double horizontal lines to indicate that the relevant
  inference rule (in this case, $\wL$) may need to be applied zero
  or more times.

\item[\fbox{$\indR/\indL$}] 
Suppose $\Pi_1$ and $\Pi$ are, respectively,  
$$
\infer[\indR]
{\Seq {\Delta_1}{p\,\vec t}}
{
 \deduce{\Seq{\Delta_1}{D\,X^p\,\vec t}}{\Pi_1'}
}
\qquad
\infer[\indL] {\Seq{p\,\vec{t}, B_2,\dots,B_n,\Gamma}{C}} {
  \deduce{\Seq{D\,S\,\vec{y}}{S\,\vec{y}}}{\Pi_S} &
  \deduce{\Seq{S\,\vec{t}, B_2,\dots,B_n, \Gamma} {C}}{\Pi'} }
$$
where $p\,\vec{x} \defmu D\,p\,\vec{x}$ and $X^p$ is a new
parameter.  Then $\Xi$ reduces to
$$
mc(mc(\Pi_1'p[(\Pi_S,S)/X^p], \Pi_S[\vec t/\vec y]), \Pi_2,\ldots,\Pi_n,\Pi').
$$

\item[\fbox{$\coindR/\coindL$}] Suppose $\Pi_1$ and $\Pi$ are 
$$
\infer[\coindR] 
{\Seq{\Delta_1}{p\,\vec{t}}} 
{
  \deduce{\Seq{\Delta_1}{S\,\vec{t}}}{\Pi_1'} &
  \deduce{\Seq{S\,\vec{y}}{D\,S\,\vec{y}}}{\Pi_S} 
} 
\qquad \qquad
\infer[\coindL] 
{\Seq{p\,\vec{t}, \dots, \Gamma}{C}}
{
  \deduce{\Seq{D\,X^p\,\vec{t},\dots, \Gamma}{C}}{\Pi'}
}
$$
where $p\,\vec x \defnu D\,p\,\vec x$ and $X^p$ is a new parameter.
Then $\Xi$ reduces to
$$
mc(mc(\Pi_1', \Pi_S[\vec t/\vec y]), \Pi_2,\ldots,\Pi_n,\Pi'[(\Pi_S,S)/X^p]).
$$
\end{trivlist}

\paragraph{Left-commutative cases:}

In the following, we suppose that $\Pi$ ends with a left rule,
other than $\{\cL, \wL\}$, acting on $B_1$.

\begin{description}

\item[\fbox{$\indL/\circL$}] Suppose $\Pi_1$ is 
$$
\infer[\indL] {\Seq{p\,\vec{t}, \Delta_1'}{B_1}} {
  \deduce{\Seq{D\,S\,\vec{y}}{S\,\vec{y}}}{\Pi_S} &
  \deduce{\Seq{S\,\vec{t}, \Delta_1'}{B_1}}{\Pi_1'} }
$$      
where $p\,\vec{x} \defmu D\,p\,\vec{x}$.  
Let $\Xi_1 = mc(\Pi_1',\Pi_2,\ldots,\Pi_n, \Pi$.
Then $\Xi$ reduces to
$$
\infer[\indL] {\Seq{p\,\vec{t}, \Delta_1',\dots,\Delta_n}{C}} {
  \deduce{\Seq{D\,S\,\vec{y}}{S\,\vec{y}}}{\Pi_S} &
  \deduce{\Seq{S\,\vec{t}, \Delta_1',\dots,\Delta_n,\Gamma}{C}}{\Xi_1}
}
$$      
\end{description}

\paragraph{Right-commutative cases:}

\begin{description}

\item[\fbox{$-/\indL$}] Suppose $\Pi$ is 
$$
\infer[\indL] {\Seq{B_1,\dots,B_n, p\,\vec{t},\Gamma'} {C}} {
  \deduce{\Seq{D\,S\,\vec{y}}{S\,\vec{y}}}{\Pi_S} &
  \deduce{\Seq{B_1,\dots,B_n, S\,\vec{t}, \Gamma'}{C}}{\Pi'} }
\enspace ,
$$      
where $p\,\vec{x} \defmu D\,p\,\vec{x}$.  
Let $\Xi_1 = mc(\Pi_1,\ldots,\Pi_n, \Pi'$. 
Then $\Xi$ reduces to
$$
\infer[\indL] {\Seq{\Delta_1,\dots,\Delta_n, p\,\vec{t},\Gamma'}{C}} {
  \deduce{\Seq{D\,S\,\vec{y}}{S\,\vec{y}}}{\Pi_S} &
  \deduce{\Seq{\Delta_1,\dots,\Delta_n, S\,\vec{t}, \Gamma'}{C}}{\Xi_1}
} \enspace 
$$

\item[\fbox{$-/\coindR$}] Suppose $\Pi$ is 
$$
\infer[\coindR] {\Seq{B_1,\dots,B_n,\Gamma}{p\,\vec{t}}} {
  \deduce{\Seq{B_1,\dots,B_n,\Gamma} {S\,\vec{t}}}{\Pi'} &
  \deduce{\Seq{S\,\vec{y}}{D\,S\,\vec{y}}}{\Pi_S} } \enspace ,
$$
where $p\,\vec{x} \defnu D\,p\,\vec{x}$.  
Let $\Xi_1 = mc(\Pi_1,\ldots,\Pi_n,\Pi'$.
Then $\Xi$ reduces to
$$
\infer[\coindR] {\Seq{\Delta_1,\dots,\Delta_n,\Gamma}{p\,\vec{t}}} {
  \deduce{\Seq{\Delta_1,\dots,\Delta_n,\Gamma} {S\,\vec{t}}}{\Xi_1} &
  \deduce{\Seq{S\,\vec{y}}{D\,S\,\vec{y}}}{\Pi_S} } \enspace
$$
\end{description}

\end{definition}

It is clear from an inspection of the inference rules in
Figure~\ref{fig:linc} and the definition of cut reduction (see
Appendix~\ref{app:reduc}) that every
derivation ending with a multicut has a reduct. 
Note that since the left-hand
side of a sequent is a multiset, the same formula may 
occur more than once in the multiset. In the cut reduction rules,
we should view these occurrences as distinct so that no ambiguity
arises as to which occurrence of a formula is subject to
the $mc$ rule. 


The following lemma shows that the reduction relation is preserved by
eigenvariable substitution. The proof is given in Appendix~\ref{app:red}.

\begin{lemma}
  \label{lm:reduct_subst}
  Let $\Pi$ be a derivation ending with a $\mc$
  and let $\theta$ be a substitution.  If $\Pi\theta$ reduces to $\Xi$
  then there exists a derivation $\Pi'$ such that $\Xi = \Pi'\theta$
  and $\Pi$ reduces to $\Pi'$.
\end{lemma}

\subsection{Normalizability}
\label{sec:norm}

\begin{definition}
  \label{def:norm}
  We define the set of \emph{normalizable} derivations to be the
  smallest set that satisfies the following conditions:
  \begin{enumerate}
  \item If a derivation $\Pi$ ends with a multicut, then it is
    normalizable if every reduct of $\Pi$ is normalizable.
  \item If a derivation ends with any rule other than a multicut, then
    it is normalizable if the premise derivations are normalizable.
  \end{enumerate}
\end{definition}
The set of all normalizable derivations is denoted by $\NM$.

Each clause in the definition of normalizability asserts that a
derivation is normalizable if certain (possibly infinitely many) other
derivations are normalizable. We call the latter the
\emph{predecessors} of the former.  Thus a derivation is normalizable
if the tree of its successive predecessors is well-founded.  We refer
to this well-founded tree as its \emph{normalization}.  Since a
normalization is well-founded, it has an associated induction
principle: for any property $P$ of derivations, if for every
derivation $\Pi$ in the normalization, $P$ holds for every predecessor
of $\Pi$ implies that $P$ holds for $\Pi$, then $P$ holds for every
derivation in the normalization.  We shall define explicitly a measure
on a normalizable derivation based on its normalization tree.

\begin{definition}[Normalization Degree]
  \label{def:deg-norm}
  Let $\Pi$ be a normalizable derivation. The \emph{normalization
    degree of $\Pi$}, denoted by $nd(\Pi)$, is defined by induction on
  the normalization of $\Pi$ as follows:
$$
nd(\Pi) = 1 + \lub {\{nd(\Pi') \mid \Pi' \hbox{is a predecessor of
    $\Pi$}\}}
$$
\end{definition}
The normalization degree of $\Pi$ is basically the height of its
normalization tree. Note that $nd(\Pi)$ can be an ordinal in general,
due to the possibly infinite-branching rule $\eqL$.

\begin{lemma}
  \label{lm:norm-cut-free}
  If there is a normalizable derivation of a sequent, then there is a
  cut-free derivation of the sequent.
\end{lemma}
\begin{proof} Similarly to~\cite{mcdowell00tcs}.  \qed
\end{proof}

In the proof of the main lemma for cut elimination
(Lemma~\ref{lm:comp}) we shall use induction on the normalization degree,
instead of using directly the normalization ordering. The reason is
that in some inductive cases in the proof, we need to compare a
(normalizable) derivation with its instances, but the normalization
ordering does not necessarily relate the two, e.g., $\Pi$ and
$\Pi\theta$ may not be related by the normalization ordering, although
their normalization degrees are (see Lemma~\ref{lm:norm-degree}).
Later, we shall define a stronger ordering called \emph{reducibility},
which implies normalizability. In the cut elimination proof for
$\FOLDN$~\cite{mcdowell00tcs}, in one of the inductive cases, an
implicit reducibility ordering is assumed to hold between derivation
$\Pi$ and its instance $\Pi\theta$. As the reducibility ordering in
their setting is a subset of the normalizability ordering, this
assumption may not hold in all cases, and as a consequence there is a
gap in the proof in~\cite{mcdowell00tcs}.\footnote{ This gap was fixed
  in~\cite{tiu04phd} by strengthening the main lemma for cut
  elimination. Recently, Andrew Gacek and Gopalan Nadathur proposed
  another fix by assigning an explicit ordinal to each reducible
  derivation, and using the ordering on ordinals to replace the
  reducibility ordering in the lemma.  A discussion of these fixes can
  be found in the errata page of the paper~\cite{mcdowell00tcs}:
  \mbox{\url{http://www.lix.polytechnique.fr/Labo/Dale.Miller/papers/tcs00.errata.html}}.
  We essentially follow Gacek and Nadathur's approach here, although we
  assign ordinals to normalizable derivations rather than to reducible
  derivations.  }

The next lemma states that normalization is closed under
substitutions.
\begin{lemma}
  \label{lm:subst-norm}
  If $\Pi$ is a normalizable derivation, then for any substitution
  $\theta$, $\Pi\theta$ is normalizable.
\end{lemma}
\begin{proof} By induction on $nd(\Pi)$.
  \begin{enumerate}
  \item If $\Pi$ ends with a multicut, then $\Pi\theta$ also ends with
    a multicut. By Lemma~\ref{lm:reduct_subst} every reduct of
    $\Pi\theta$ corresponds to a reduct of $\Pi$, therefore by
    induction hypothesis every reduct of $\Pi\theta$ is normalizable,
    and hence $\Pi\theta$ is normalizable.
 
  \item Suppose $\Pi$ ends with a rule other than multicut and has
    premise derivations $\{\Pi_i\}$.  By Definition~\ref{def:subst}
    each premise derivation in $\Pi\theta$ is either $\Pi_i$ or
    $\Pi_i\theta$.  Since $\Pi$ is normalizable, $\Pi_i$ is
    normalizable, and so by the induction hypothesis $\Pi_i\theta$ is
    also normalizable.  Thus $\Pi\theta$ is normalizable.  \qed
 
  \end{enumerate}
\end{proof}

The normalization degree is non-increasing under eigenvariable
substitution.

\begin{lemma}
  \label{lm:norm-degree}
  Let $\Pi$ be a normalizable derivation. Then $nd(\Pi) \geq
  nd(\Pi\theta)$ for every substitution $\theta$.
\end{lemma}
\begin{proof}
  By induction on $nd(\Pi)$ using Definition~\ref{def:subst} and
  Lemma~\ref{lm:reduct_subst}.  Note that $nd(\Pi\theta)$ can be
  smaller than $nd(\Pi)$ because substitution may reduces the number
  of premises in $\eqL$, \ie, if $\Pi$ ends with an $\eqL$ acting on,
  say $x = y$ (which are unifiable), and $\theta$ is a substitution
  that maps $x$ and $y$ to distinct constants then $\Pi\theta$ ends
  with $\eqL$ with empty premise.  \qed
\end{proof}

\subsection{Parametric reducibility}
\label{sec:red}

In the following, we shall use the term ``type'' in two different
settings: in categorizing terms and in categorizing derivations.  To
avoid confusion, we shall refer to the types of terms as
\emph{syntactic types}, and the term ``type'' is reserved for types of
derivations.

Our notion of a type of a set of derivations may abstract from
particular first-order terms in a formula. This is because our
definition of reducibility (candidates) will have to be closed under
eigenvariable substitutions, which is in turn imposed by the fact that
our proof rules allow instantiation of eigenvariables in the
derivations (i.e., the $\eqL$ and the $subst$ rules).

\begin{definition}[Types of derivations]
  \label{def:type-of-drv}
  We say that \emph{a derivation $\Pi$ has type $C$} if the end
  sequent of $\Pi$ is of the form $\Seq{\Gamma}{C}$ for some $\Gamma$.
  Let $F$ be a term with syntactic type $\alpha_1 \to \cdots \to
  \alpha_n \to o$, where each $\alpha_i$ is a syntactic
  efo-type. \footnote{From now on, we shall assume that the $\alpha_i$
  are always efo-types.}  A
  set of derivations $\Sscr$ is said to be \emph{of type $F$} if every
  derivation in $\Sscr$ has type $F\,u_1 \ldots u_n$ for some terms
  $u_1,\ldots,u_n$.  Given a list of terms $\vec u = u_1 : \alpha_1,
  \ldots, u_n : \alpha_n$ and a set of derivations $\Sscr$ of type $F
  : \alpha_1 \to \cdots \to \alpha_n \to o$, we denote with
  $\Sscr\,\vec u$ the set
$$
\Sscr\,\vec u = \{\Pi \in \Sscr \mid \hbox{$\Pi$ has type $F\,\vec u$
} \}
$$
\end{definition}

\vskip10pt

\begin{definition}[Reducibility candidate]
  \label{def:candidates}
  Let $F$ be a \emph{closed term} having the syntactic type $\alpha_1
  \to \cdots \to \alpha_n \to o$.  A set of derivations $\Rscr$ of
  type $F$ is said to be a \emph{reducibility candidate of type $F$}
  if the following hold:
  \begin{description}
  \item[CR0] If $\Pi \in \Rscr$ then $\Pi\theta \in \Rscr$, for every
    $\theta$.
  \item[CR1] If $\Pi \in \Rscr$ then $\Pi$ is normalizable.
  \item[CR2] If $\Pi \in \Rscr$ and $\Pi$ reduces to $\Pi'$ then $\Pi'
    \in \Rscr$.
  \item[CR3] If $\Pi$ ends with $mc$ and all its reducts are in
    $\Rscr$, then $\Pi \in \Rscr$.
  \item[CR4] If $\Pi$ ends with $init$, then $\Pi \in \Rscr$.
  \item[CR5] If $\Pi$ ends with a left-rule or $subst$, then all its minor premise
    derivations are normalizable, and all its major premise
    derivations are in $\Rscr$, then $\Pi \in \Rscr$.
  \end{description}
  We shall write $\Rscr : F$ to denote a reducibility candidate
  $\Rscr$ of type $F$.
\end{definition}


The conditions {\bf CR1} and {\bf CR2} are similar to the eponymous 
conditions in Girard's definition of reducibility candidates in his
strong normalisation proof for System F (see~\cite{girard89book},
Chapter 14).  Girard's {\bf CR3} is expanded in our definition to {\bf
  CR3, CR4} and {\bf CR5}.  These conditions deal with what Girard
refers to as ``neutral'' proof term (or, in our setting,
derivations). Neutrality corresponds to derivations ending in $mc$,
$init$, $subst$, or a left rule.

The condition {\bf CR0} is needed because our cut reduction rules involve
substitution of eigenvariables in some cases (i.e., those that involve
permutation of $\eqL$ and $subst$ in the left/right commutative cases), and consequently,
the notion of reducibility (candidate) needs to be preserved
under eigenvariable substitution. 

Let $\Sscr$ be a set of derivations of type $B$ and let $\Tscr$ be a
set of derivations of type $C$.  Then $\Sscr \Rightarrow \Tscr$
denotes the set of derivations such that $ \Pi \in \Sscr \Rightarrow
\Tscr $ if and only if $\Pi$ ends with a sequent $\Seq{\Gamma}{C}$
such that $B \in \Gamma$ and for every $\Xi \in \Sscr$, we have
$mc(\Xi, \Pi) \in \Tscr$.

Let $S$ be a closed term. Define $\NM_S$ to be the set
$$
\NM_S = \{\Pi \mid \Pi\in\NM\ \hbox{and is of type $S\,\vec u$ for
  some $\vec u$} \}.
$$
It can be shown that $\NM_S$ is a reducibility candidate of type $S$. 

\begin{lemma}
\label{lm:norm red}
Let $S$ be a term of syntactic type $\alpha_1\ra \cdots \ra \alpha_n
\ra o$. 
Then the set $\NM_S$ is a reducibility candidate of type $S$. 
\end{lemma}
\begin{proof}
{\bf CR0} follows from Lemma~\ref{lm:subst-norm}, {\bf CR1} follows
from the definition of $\NM_S$, and the rest follow
from Definition~\ref{def:norm}. \qed
\end{proof}

\begin{definition}[Candidate substitution]
\label{def:candidate-subst}
A \emph{candidate substitution} $\Omega$ is a partial map from parameters to
triples of reducibility candidates, derivations and closed terms
such that whenever $\Omega(X^p) = (\Rscr, \Pi, S)$, we have
\begin{itemize}
\item $S$ has the same syntactic type as $p$, 
\item $\Rscr$ is a reducibility candidate of type $S$, and
\item either one of the following holds:
\begin{itemize}
\item $p\,\vec x \defmu B\,p\,\vec x$ and $\Pi$ is a normalizable
derivation of $\Seq {B\,S\,\vec y}{S\,\vec y}$, or
\item $p\,\vec x \defnu B\,p\,\vec x$ and $\Pi$ is a normalizable
derivation of $\Seq {S\,\vec y}{B\,S\,\vec y}$. 
\end{itemize}
\end{itemize}
\end{definition}
We denote with $supp(\Omega)$ the \emph{support} of $\Omega$,
i.e., the set of parameters on which $\Omega$ is defined.
Each candidate substitution $\Omega$ determines a unique
parameter substitution $\Theta$, given by:
$$
\Theta(X^p) = (\Pi, S) \ \hbox{iff } \Omega(X^p) = (\Rscr, \Pi, S) \
\hbox{for some $\Rscr$}.
$$
We denote with $Sub(\Omega)$ the parameter substitution $\Theta$
obtained this way.  We say that a parameter $X^p$ is \emph{fresh for
  $\Omega$}, written $X^p \# \Omega,$ if $X^p \# Sub(\Omega)$.

\paragraph{Notation}
Since every candidate substitution has a corresponding
parameter substitution, we shall often treat a candidate
substitution as a parameter substitution. In particular,
we shall write $C\Omega$ to denote $C(Sub(\Omega))$
and $\Pi\Omega$ to denote $\Pi(Sub(\Omega))$. 

We are now ready to define the notion of parametric reducibility.  We
follow a similar approach  for $\FOLDN$~\cite{mcdowell00tcs}, where
families of reducibility sets are defined by the \emph{level} of
derivations, \ie the size of the types of
derivations. In defining a family (or families) of sets of derivations
at level $k$, we assume that reducibility sets at level $j < k$ are
already defined. The main difference with the notion of reducibility
for $\FOLDN$, aside from the use of parameters in the clause for
(co)induction rules (which do not exist in $\FOLDN$), is in the
treatment of the induction rules.

\begin{definition}[Parametric reducibility]
\label{def:param red} 
Let $\Fscr_k$ be the set of all formula of size $k$, \ie $ \{F \mid
|F| = k \} $.  The family of \emph{parametric reducibility sets}
$\RED_C[\Omega],$ where $C$ is a formula and $\Omega$ is a candidate
substitution, is defined by induction on the size of $C$ as follows.
For each $k$, the family of \emph{parametric reducibility sets of
  level $k$}
$$
\{\RED_C[\Omega] \}_{C \in \Fscr_k}
$$
is the smallest family of sets satisfying, 
for each $C \in \Fscr_k$:
\begin{description}
\item[P1] Suppose $C = X^p\,\vec u$ for some $\vec u$ and some parameter $X^p$. 
If $X^p \in supp(\Omega)$ then 
$\RED_C[\Omega] = \Rscr~\vec u,$ where $\Omega(X^p) = (\Rscr, \Pi_S, S)$. 
Otherwise, $\RED_C[\Omega] = \NM_{X^p}\,\vec u$. 
\end{description}

Otherwise, $C \not = X^p\,\vec u$, for any $\vec u$ and  $X^p$. 
Then a derivation $\Pi$ of type $C\Omega$ is in
$\RED_C[\Omega]$ if it is normalizable and one of the following holds:
\begin{description}
\item[P2] $\Pi$ ends with $mc$, and all its reducts are in
      $\RED_C[\Omega]$.
\item[P3] $\Pi$ ends with $\oimpR$, i.e., $C = B \oimp D$ and $\Pi$ is
  of the form: 
    $$
    \infer[\oimpR] 
          {\Seq \Gamma {B\Omega \oimp D\Omega}} 
          {\deduce{\Seq{\Gamma, B\Omega}{D\Omega}}{\Pi'} }
    $$
    and for every substitution $\rho$, $\Pi'\rho \in
    (\RED_{B\rho}[\Omega] \Rightarrow \RED_{D\rho}[\Omega])$.
    
\item[P4] $\Pi$ ends with $\indR$, i.e., 
    $$
    \infer[\indR, \hbox{where $p\,\vec x \defmu B\,p\,\vec x$}]
    {\Seq \Gamma {p\,\vec t} }
    {\deduce{\Seq \Gamma {B\,X^p\,\vec t}}{\Pi'}}
    $$
    without loss of generality, assume that $X^p \# \Omega$: for
    every reducibility candidate $(\Sscr : I)$, where $I$ is a closed
    term of the same syntactic type as $p$, for every normalizable
    derivation $\Pi_I$ of $\Seq {B\,I\,\vec y} {I\,\vec y}$, if for
    every $\vec u$ the following holds:
    $$
    \Pi_I[\vec u/ \vec y] \in (\RED_{(B\,X^p\,\vec u)}[\Omega, (\Sscr,
    \Pi_I, I)/X^p] \Rightarrow \Sscr~\vec u) 
    $$
    then 
    $$mc(\Pi'[(\Pi_I,I)/X^p], \Pi_I[\vec t/\vec y])  \in \Sscr\,\vec t$$

\item[P5] $\Pi$ ends with $\coindR$, i.e.,
  $$
  \infer[\coindR, \hbox{where $p\,\vec x \defnu B\,p\,\vec x$ }]
  {\Seq \Gamma {p\,\vec t} }
  {
    \deduce{\Seq \Gamma {I\,\vec t}}{\Pi'}
    &
    \deduce{\Seq {I\,\vec y}{B\,I\,\vec y}}{\Pi_I}
  }
  $$
  and there exist a parameter $X^p$ such that $X^p \# \Omega$ and a
  reducibility candidate $(\Sscr : I)$ such that $\Pi' \in \Sscr$ and
  $$
  \Pi_I[\vec u/ \vec y] \in (\Sscr\,\vec u \Rightarrow
  \RED_{B\,X^p\,\vec u}[\Omega, (\Sscr, \Pi_I, I)/X^p]) \ \hbox{for every
    $\vec u$.  }
  $$

\item[P6] $\Pi$ ends with any other rule and its major premise derivations 
  are in the parametric reducibility sets of the appropriate types.
\end{description}
We shall write $\RED_C$, instead of $\RED_C[\Omega]$, when the
$supp(\Omega)$ of a candidate substitution is the empty set.  A
derivation $\Pi$ of type $C$ is \emph{reducible} if $\Pi \in \RED_C$.
\end{definition}

Some comments and comparison with Girard's definition of parametric
reducibility for System F~\cite{girard89book} are in order, although
our technical setting is somewhat different from that of Girard:
\begin{itemize}
\item Condition {\bf P3} quantifies over $\rho$. This is needed to
  show that reducibility is closed under substitution (see
  Lemma~\ref{lm:red-subst}).  A similar quantification is used in the
  definition of reducibility for $\FOLDN$~\cite{mcdowell00tcs} for the
  same purpose.  In the same clause, we also quantify over derivations
  in $\RED_{B\rho}[\Omega]$, but since $B\rho$ has smaller size than
  $B\oimp D$, this quantification is legitimate and the definition is
  well-founded.  Note also the similar quantification in {\bf P4} and
  {\bf P5}, where the parametric reducibility set $\RED_{p\,\vec t} \,
  [\Omega]$ is defined in terms of $\RED_{(B\,X^p\,\vec t)} [\Omega]$.
  By Lemma~\ref{lm:level}, $|p\,\vec t| > |B\,X^p\,\vec t|$ so in both
  cases the set $\RED_{(B\,X^p\,\vec t)} [\Omega]$ is already defined
  by induction. It is clear by inspection of the clauses that the
  definition of parametric reducibility is well-founded.

\item Clauses {\bf P2} and {\bf P6} are needed to show that the notion
  of parametric reducibility is closed under left-rules, $id$ and
  $mc$, i.e., condition {\bf CR3} -- {\bf CR5}. This is also a point
  where our definition of parametric reducibility diverges from a
  typical definition of reducibility in natural deduction
  (e.g.,~\cite{girard89book}), where closure under reduction for
  ``neutral''
  terms 
  is a derived
  property.

\item 
  {\bf P4} (and dually \textbf{P5}) can be intuitively explained in
  terms of the second-order encoding of inductive definitions.  To
  simplify presentation, we restrict to the propositional case, so,
  {\bf P4} can be simplified as follows:
\begin{quote}
Suppose $\Pi$ ends with $\indR$, i.e., 
    $$
    \infer[\indR, \hbox{where $p \defmu B\,p$}]
    {\Seq \Gamma {p} }
    {\deduce{\Seq \Gamma {B\,X^p}}{\Pi'}}
    $$
    without loss of generality, assume that $X^p \# \Omega$: for
    every reducibility candidate $(\Sscr : I)$, where $I$ is a closed
    term of the same syntactic type as $p$, for every normalizable
    derivation $\Pi_I$ of $\Seq {B\,I} {I}$, if
    $
    \Pi_I \in (\RED_{B\,X^p}[\Omega, (\Sscr, \Pi_I, I)/X^p] \Rightarrow \Sscr) 
    $,
    then 
    $mc(\Pi'[(\Pi_I,I)/X^p], \Pi_I)  \in \Sscr$.
\end{quote}
Note that in the propositional $\Linc$, the set 
$$\RED_{B\,X^p}[\Omega, (\Sscr, \Pi_I, I)/X^p] \Rightarrow \Sscr$$
is equivalent to $\RED_{B\,X^p \oimp X^p}[\Omega, (\Sscr, \Pi_I,
I)/X^p],$ i.e., a set of reducible derivations of type $B\,I \oimp I$.
So, intuitively, $\Pi'$ can be seen as a higher-order function that
takes any function of type $B\,I \oimp I$ (i.e., the derivation
$\Pi_I$), and turns it into a derivation of type $I$ (i.e., the
derivation $mc(\Pi'[(\Pi_I,I)/X^p], \Pi_I)$), for all candidate
$(\Sscr : I)$. This intuitive reading matches the second-order
interpretation of $p$, i.e., $\forall I. (B\,I \oimp I) \oimp I$,
where the universal quantification is interpreted as the universal
type constructor and $\oimp$ is interpreted as the function type
constructor in System F.
\end{itemize}

\medskip
We shall now establish a list of properties of the parametric reducibility
sets that will be used in the main cut elimination proof. 
The main property that we are after is one which shows that a certain
set of derivations formed using a family of parametric reducibility sets 
actually forms a reducibility candidate. This will be important later
in constructing a reducibility candidate which acts as a co-inductive
``witness'' in the main cut elimination proof. 
The proofs of the following lemmas are mostly routine and rather tedious;
so we omit them here, but they can be found in Appendix~\ref{app:red}.

\begin{lemma}
  \label{lm:red-norm}
If $\Pi \in \RED_C[\Omega]$ then $\Pi$ is normalizable.
\end{lemma}

Since every $\Pi \in \RED_C[\Omega]$ is normalizable,
$nd(\Pi)$ is defined. This fact will be used implicitly in subsequent
proofs, i.e., we shall do induction on
$nd(\Pi)$ to prove properties of $\RED_C[\Omega]$. 

\begin{lemma}
  \label{lm:red-subst}
  If $\Pi \in \RED_C[\Omega]$ then for every substitution $\rho$, 
  $\Pi\rho \in \RED_{C\rho}[\Omega]$. 
\end{lemma}

\begin{lemma}
\label{lm:red vacuous}
Let $\Omega = [\Omega', (\Rscr, \Pi_S, S)/ X^p]$.
Let $C$ be a formula such that $X^p \# C$. 
Then for every $\Pi$,  $\Pi \in \RED_C[\Omega]$ if and only if
$\Pi \in \RED_C[\Omega']$. 
\end{lemma}

\begin{lemma}
\label{lm:red candidate}
Let $\Omega$ be a candidate substitution and $F$ be a closed term of
syntactic type $\alpha_1 \ra \cdots \ra \alpha_n \ra o$. 
Then the set
$$
\Rscr = \{\Pi \mid \Pi \in \RED_{F\,\vec u}[\Omega] \ \hbox{for some $\vec u$} \}
$$
is a reducibility candidate of type $F\Omega$.
\end{lemma}

\begin{lemma}
\label{lm:red param subst}
Let $\Omega$ be a candidate substitution and let $X^p$
be a parameter such that $X^p \# \Omega$. 
Let $S$ be a closed term of the same type as $p$ and let
$$
\Rscr = \{\Pi \mid \Pi \in \RED_{S\,\vec u}[\Omega] \ \hbox{for some $\vec u$} \}.
$$
Suppose $[\Omega, (\Rscr, \Psi, S\Omega)/ X^p]$ is a candidate
substitution, for some $\Psi$. 
Then
$$
\RED_{C[S/X^p]}[\Omega] = \RED_C[\Omega, (\Rscr, \Psi, S\Omega)/X^p].
$$
\end{lemma}

\subsection{Cut elimination}
\label{sec:ceproof}

We shall now show that every derivation is reducible, hence every
derivation can be normalized to a cut-free derivation.  But in order
to prove this, we need a slightly more general lemma, which states
that every derivation is in $\RED_C[\Omega]$ for a certain kind of
candidate substitution $\Omega$.  The precise definition is given
below.

\begin{definition}[Definitional closure]
  A candidate substitution $\Omega$ is \emph{definitionally closed} if
  for every $X^p \in supp(\Omega)$, if $\Omega(X^p) = (\Rscr, \Pi_S,
  S)$ then either one of the following holds:
  \begin{itemize}
  \item $p\,\vec x \defmu B\,p\,\vec x$, for some $B$ and for every
    $\vec u$ of the appropriate syntactic types:
$$
\Pi_S[\vec u/\vec x] \in \RED_{B\,X^p\,\vec u}\, [\Omega] \Rightarrow
\Rscr\,\vec u.
$$

\item $p\,\vec x \defnu B\,p\,\vec x$, for some $B$ and for every
  $\vec u$ of the appropriate syntactic types:
$$
\Pi_S[\vec u/\vec x] \in \Rscr\,\vec u \Rightarrow \RED_{B\,X^p\,\vec
  u}\, [\Omega].
$$
\end{itemize}
\end{definition}

The next two lemmas show that definitionally closed substitutions can
be extended in a way that preserves  definitional closure.

\begin{lemma}
  \label{lm:clo-ext-ind}
  Let $\Omega = [\Omega', (\Rscr,\Pi_S,S)/X^p]$ be a candidate
  substitution such that $p\,\vec x \defmu B\,p\,\vec x$, $\Omega'$ is
  definitionally closed, and for every $\vec u$ of the same types as
  $\vec x$,
$$\Pi_S[\vec u/\vec x] \in \RED_{B\,X^p\,\vec u}\, [\Omega] \Rightarrow
\Rscr\,\vec u$$  Then $\Omega$ is definitionally closed.
\end{lemma}
\begin{proof}
  Let $Y^q \in supp(\Omega)$. Suppose $\Omega(Y^q) = (\Sscr, \Pi_I,
  I)$.  
  We need to show that
$$
\Pi_I[\vec t/\vec x] \in \RED_{B\,Y^q\,\vec t}\, [\Omega] \Rightarrow
\Sscr\,\vec t
$$
for every $\vec t$ of the same types as $\vec x$.  If $Y^q = X^p$ then
this follows from the assumption of the lemma.  Otherwise, $Y^q \in
supp(\Omega')$, and by the definitional closure assumption on
$\Omega'$, we have
$$
\Pi_I[\vec t/\vec x] \in \RED_{B\,Y^q\,\vec t}\, [\Omega'] \Rightarrow
\Sscr\,\vec t
$$
for every $\vec t$.  Since  $X^p \#
(B\,Y^q\,\vec t)$ (recall that definition clauses cannot contain
occurrences of parameters), by Lemma~\ref{lm:red vacuous} we have
$\RED_{B\,Y^q\,\vec t}\, [\Omega'] = \RED_{B\,Y^q\,\vec t}\, [\Omega]$, and
therefore the result. \qed
\end{proof}

\begin{lemma}
  \label{lm:clo-ext-coind}
  Let $\Omega = [\Omega', (\Rscr,\Pi_S,S)/X^p]$ be a candidate
  substitution such that $p\,\vec x \defnu B\,p\,\vec x$, $\Omega'$ is
  definitionally closed, and for every $\vec u$ of the same types as
  $\vec x$,
$$\Pi_S[\vec u/\vec x] \in \Rscr \,\vec u \Rightarrow
\RED_{B\,X^p\,\vec u}\, [\Omega]$$  Then $\Omega$ is definitionally
closed.
\end{lemma}
\begin{proof}
  Analogous to the proof of Lemma~\ref{lm:clo-ext-ind}. \qed
\end{proof}

We are now ready to state the main lemma for cut elimination.

\begin{lemma}
  \label{lm:comp}
  Let $\Omega$ be a definitionally closed candidate substitution.  Let
  $\Pi$ be a derivation of $\Seq{B_1,\ldots,B_n,\Gamma}{C}$, and let
$$
\deduce{\Seq{\Delta_1}{B_1\Omega}}{\Pi_1}, \quad \ldots
\quad,~\deduce{\Seq{\Delta_n}{B_n\Omega}}{\Pi_n} 
$$
where $n \geq 0$, be derivations in, respectively,
$\RED_{B_1}\, [\Omega], \ldots, \RED_{B_n}\, [\Omega]$.  Then the derivation
$\Xi$
\begin{displaymath}
  \infer[\mc]{\Seq{\Delta_1,\ldots,\Delta_n,\Gamma\Omega}{C\Omega}}
  {\deduce{\Seq{\Delta_1}{B_1\Omega}}
    {\Pi_1}
    & \cdots
    & \deduce{\Seq{\Delta_n}{B_n\Omega}}
    {\Pi_n}
    & \deduce{\Seq{B_1\Omega,\ldots,B_n\Omega,\Gamma\Omega}{C\Omega}}
    {\Pi\Omega}}
\end{displaymath}
is in $\RED_{C}[\Omega]$.
\end{lemma}
\begin{proof}
  The proof is by induction on
$$
\Mscr(\Xi) = \langle \measure{\Pi}, \sum_{i=1}^n |B_i|,
ND(\Xi) \rangle
$$
where $ND(\Xi)$ is the multiset $\{nd(\Pi_1),\ldots,nd(\Pi_n) \}$
of normalization degrees of $\Pi_1$ to $\Pi_n$.
Note that the measure $\Mscr$ can be well-ordered using the
lexicographical ordering. We shall refer to this ordering as simply
$<$.  Note also that $\Mscr$ is insensitive to the order in which
$\Pi_i$ is given, thus when we need to distinguish one of the $\Pi_i$,
we shall refer to it as $\Pi_1$ without loss of generality.  The
derivation $\Xi$ is in $\RED_C[\Omega]$ if all its reducts are in
$\RED_C[\Omega]$.
 
\paragraph{\bf CASE I: n = 0} 
In this case, $\Xi$ reduces to $\Pi\Omega$, thus it is enough to show
that that $\Pi\Omega \in \RED_C[\Omega]$.  This is proved by case
analysis on $C$ and on the last rule of $\Pi$.

\paragraph{\bf I.1} Suppose $C = X^p \,\vec t$ for some parameter
$X^p$ and terms $\vec t$.

If $X^p \not \in supp(\Omega)$ then we need only to show that
$\Pi\Omega$ is normalizable. This follows mostly straightforwardly
from the induction hypothesis and Lemma~\ref{lm:red-norm}. The only
interesting case is when $\Pi$ ends with $\coindLP$ on some $Y^q\,\vec
u$ such that $Y^q \in supp(\Omega)$, i.e., $\Pi$ takes the form
$$
\infer[\coindLP.]  {\Seq {Y^q\,\vec u, \Gamma}{C}} { \deduce{\Seq
    {D\,Y^q\,\vec u, \Gamma}{C}}{\Pi'} }
$$
Suppose $\Omega(Y^q) = (\Rscr, \Pi_S,S)$. 
Then $\Pi\Omega = mc(mc(\idrv_{S\,\vec u}, \Pi_S[\vec u/\vec x]), \Pi'\Omega)$. 
By {\bf CR4} we have that $\idrv_{S\,\vec u} \in \Rscr$, so by the
definitional closure of $\Omega$ and \textbf{CR3}, we have $mc(\idrv_{S\,\vec u},
\Pi_S[\vec u/\vec x]) \in \RED_{D\,S\,\vec u}\, [\Omega]$.  Since
$\measure{\Pi'} < \measure{\Pi}$, and since $\Pi\Omega =
mc(mc(\idrv_{S\,\vec u}, \Pi_S[\vec u/\vec x]),\Pi'\Omega),$ by the
induction hypothesis, we have $\Pi\Omega \in \RED_C[\Omega]$, and
therefore, by Lemma~\ref{lm:red-norm}, $\Pi\Omega$ is normalizable.
Note that this case is actually independent of the form of $C$.

Otherwise, suppose $X^p \in supp(\Omega)$, and $\Omega(X^p) =
(\Rscr,\Pi_S,S)$.  Then there are several cases to consider, based on
the last rule of $\Pi$.  In all cases, we need to show that $\Pi\Omega
\in \Rscr\,\vec t$.  Note that since $\Pi\Omega$ is of type $S\,\vec
t$, $\Pi\Omega \in \Rscr$ implies that $\Pi\Omega \in \Rscr\,\vec t$.
So in the following in some cases we need to show only that
$\Pi\Omega\in\Rscr$.

\begin{itemize}
\item $\Pi$ ends with $init$: then  $\Pi\Omega$ also ends with
  $init$ and by {\bf CR4}, $\Pi\Omega \in \Rscr$.

\item $\Pi$ ends with $mc$: This follows from the induction hypothesis
  and Lemma~\ref{lm:red-norm}.

\item $\Pi$ ends with $\coindLP$: Suppose $\Pi$ ends with $\coindLP$
  acting on a formula $Y^q\,\vec u$. If $Y^q \not \in supp(\Omega)$,
  then this follows immediately from the induction hypothesis and {\bf
    CR5}. If $Y^q \in supp(\Omega)$, then we use the same arguments as
  shown above.

\item $\Pi$ ends with $subst$ or a left-rule other than $\coindLP$: Suppose the
  premise derivations of the rule are
$$
\left\{ \hbox{$\deduce{\Seq {\Gamma_i}{C_i}}{\Psi_i}$} \right\}_{i \in  I}
$$
for some index set $I$. Then $\Pi\Omega$ ends with the same left rule
and has premise derivations $\{\Psi_i\Omega \}_{i\in I}$.

By the induction hypothesis, $\Psi_i \in \RED_{C_i}[\Omega]$ for every
$i\in I$, and by Lemma~\ref{lm:red-norm}, each $\Psi_i$ is also
normalizable. The latter implies that $\Pi\Omega$ is normalizable.
Note that if $\Psi_i$ is a major premise derivation, then $C_i =
X^p\,\vec u$ for some $\vec u$, and we have $\Psi_i\Omega \in \Rscr$.
Therefore, by {\bf CR5}, we have that $\Pi\Omega \in \Rscr$.

\item Suppose $\Pi$ ends with $\indRP$:
$$
\infer[\indRP] {\Seq{\Gamma}{X^p\vec t}} { \deduce{\Seq \Gamma
    {D\,X^p\,\vec t}}{\Pi'} }
$$
where $p\,\vec x \defmu D\,p\,\vec x$.  
Then $\Pi\Omega = mc(\Pi'\Omega, \Pi_S[\vec t/\vec x]$. 
From the induction hypothesis, we have that $\Pi'\Omega \in
\RED_{D\,X^p\,\vec t}\, [\Omega]$.  This, together with the definitional
closure of $\Omega$, implies that $\Pi\Omega$ is indeed in
$\Rscr\,\vec t$.

\end{itemize}

\paragraph{\bf I.2:} Suppose $C \not = X^p \,\vec t$ for any parameter
$X^p$ and any terms $\vec t$.

Most subcases follow easily from the induction hypothesis,
Lemma~\ref{lm:red-norm} and Definition~\ref{def:param red}. The
subcases where $\Pi$ ends with a left rule follow the same lines of
arguments as in Case I.1 above.  We show here the non-trivial subcases
involving right-introduction rules:

\paragraph{\bf I.2.a} Suppose $\Pi$ ends with $\oimpR$, as shown below left.
Then $\Pi\Omega$ is as shown below right.
$$
\infer[\oimpR] {\Seq \Gamma {C_1 \oimp C_2}} { \deduce{\Seq {\Gamma,
      C_1}{C_2}}{\Pi'} }
\qquad
\infer[\oimpR] {\Seq {\Gamma\Omega} {C_1\Omega \oimp C_2\Omega}} {
  \deduce{\Seq {\Gamma\Omega, C_1\Omega}{C_2\Omega}}{\Pi'\Omega} }
$$
To show $\Pi\Omega \in \RED_C[\Omega]$, we need to show that
$\Pi\Omega$ is normalizable and that
\begin{equation}
  \label{eq:ce1}
  \Pi'\Omega\theta \in \RED_{C_1\theta}[\Omega] \Rightarrow \RED_{C_2\theta}[\Omega] 
\end{equation}
for every $\theta$.  Since $\measure{\Pi'} < \measure{\Pi}$, by the
induction hypothesis, $\Pi'\Omega \in \RED_{C_2}[\Omega]$.
Normalizability of $\Pi\Omega$ then follows immediately from this and
Lemma~\ref{lm:red-norm}. It remains to show that
Statement~\ref{eq:ce1} holds:

Let $\Psi$ be a derivation in $\RED_{C_1\theta}[\Omega]$.  Let $\Xi_1
= mc(\Psi, \Pi'\Omega\theta)$.  Note that since parameter substitution
commutes with eigenvariable substitution, $\Pi' \Omega \theta =
\Pi'\theta\Omega$.  Since $\measure{\Pi'\theta} \leq \measure{\Pi'} <
\measure{\Pi}$ (Lemma~\ref{lm:subst-height}), by induction hypothesis,
we have $\Xi_1 \in \RED_{C_2\theta}[\Omega]$.  In other words,
Statement~\ref{eq:ce1} holds for arbitrary $\theta$, and therefore by
Definition~\ref{def:param red}, $\Pi\Omega \in \RED_C[\Omega]$.

\paragraph{\bf I.2.b} Suppose $\Pi$ ends with $\indR$,
as shown below left, where $p\,\vec x \defmu D\,p\,\vec x$. 
We can assume w.l.o.g.\  that $X^p \# \Omega$. Then $\Pi\Omega$ is
as shown below right.
$$
\infer[\indR] {\Seq{\Gamma}{p\,\vec t}} { \deduce{\Seq \Gamma
    {D\,X^p\,\vec t}}{\Pi'} }
\qquad
\infer[\indR] {\Seq{\Gamma\Omega}{p\,\vec t}} { \deduce{\Seq
    {\Gamma\Omega}{D\,X^p \,\vec t}}{\Pi'\Omega} }
$$
To show that $\Pi\Omega \in \RED_C[\Omega]$, we need to show that
$\Pi\Omega$ is normalizable (as before this easily follows from the induction
hypothesis and Lemma~\ref{lm:red-norm}) and that
\begin{equation}
  \label{eq:caseI3-a}
  mc((\Pi'\Omega)[(\Pi_S,S)/X^p], \Pi_S[\vec t/\vec x]) \in \Rscr\,\vec t
\end{equation}
for every candidate $(\Rscr : S)$ and every $\Pi_S$ that satisfies:
\begin{equation}
  \label{eq:caseI3-b}
  \Pi_S[\vec u/\vec x] \in \RED_{D\,X^p\,\vec u}\, [\Omega, (\Rscr,\Pi_S,S)/X^p]
  \Rightarrow \Rscr\,\vec u
  \hbox{ for every $\vec u$.}
\end{equation}
Let $\Omega' = [\Omega, (\Rscr, \Pi_S, S)/X^p]$.  Note that since $X^p
\# \Omega$, we have
$
\Pi'\Omega[(\Pi_S,S)/X^p] = \Pi'\Omega'.
$
So Statement~\ref{eq:caseI3-a} above can be rewritten to
\begin{equation}
  \label{eq:caseI3-c}
  mc(\Pi'\Omega', \Pi_S[\vec t/\vec x]) \in \Rscr\,\vec t. 
\end{equation}
By Lemma~\ref{lm:clo-ext-ind}, we have that $\Omega'$ is
definitionally closed.  Therefore we can apply the induction
hypothesis to $\Pi'$ and $\Omega'$, obtaining $\Pi'\Omega' \in
\RED_{D\,X^p\,\vec t}\, [\Omega']$. This, together with definitional
closure of $\Omega'$, immediately implies Statement~\ref{eq:caseI3-c}
above, hence $\Pi\Omega$ is indeed in $\RED_C[\Omega]$.

\paragraph{\bf I.2.c} Suppose $\Pi$ ends with $\coindR$,
as shown below left, where $p\,\vec y \defnu D\,p\,\vec y$. Let $S' = S\Omega$.  Then
$\Pi\Omega$ is as shown below right.
$$
\infer[\coindR] {\Seq \Gamma {p\,\vec t}} { \deduce{\Seq \Gamma
    {S\,\vec t}}{\Pi'} & \deduce{\Seq {S\,\vec x}{D\,S\,\vec
      x}}{\Pi_S} }
\qquad
\infer[\coindR] {\Seq {\Gamma\Omega} {p\,\vec t}} { \deduce{\Seq
    {\Gamma\Omega}{S'\,\vec t}}{\Pi'\Omega} & \deduce{\Seq {S'\vec
      x}{D\,S'\,\vec x}}{\Pi_S\Omega} }
$$
Note that $\Pi\Omega$ is normalizable, by the induction hypothesis and
Lemma~\ref{lm:red-norm}.  To show that $\Pi\Omega \in \RED_C[\Omega]$
it remains to show that there exists a reducibility candidate $(\Rscr
: S')$ such that
\begin{description}
\item[(a)] $\Pi'\Omega \in \Rscr$, and
\item[(b)] $\Pi_S\Omega[\vec u/\vec x] \in \Rscr\,\vec u \Rightarrow
  \RED_{D\,X^p\,\vec u}\, [\Omega, (\Rscr, \Pi_S\Omega, S')/X^p]$ for a
  new $X^p \# \Omega$.
\end{description}
Let
$
\Rscr = \{\Psi \mid \Psi \in \RED_{S\,\vec u}\, [\Omega] \}.
$
By Lemma~\ref{lm:red candidate}, $\Rscr$ is a reducibility candidate
of type $S'$.  By the induction hypothesis, we have $\Pi'\Omega \in
\Rscr$, so $\Rscr$ satisfies {\bf (a)}.  Since substitution does not
increase the height of derivations, we have that $\measure{\Pi_S[\vec
  u/ \vec x]} \leq \measure{\Pi_S}$, and therefore, by applying the
induction hypothesis to $\Pi_S[\vec x/\vec u]$, we have
$
mc(\Psi, \Pi_S\Omega[\vec u/\vec x]) \in \RED_{D\,S\,\vec u}\, [\Omega]
$
for every $\Psi \in \RED_{S\,\vec u}\, [\Omega]$. In other words,
$$
\Pi_S\Omega[\vec u/\vec x] \in \RED_{S\,\vec u}\, [\Omega] \Rightarrow
\RED_{D\,S\,\vec u}\, [\Omega].
$$
Notice that $\RED_{S\,\vec u}\, [\Omega]$ is exactly $\Rscr\,\vec u$.  So
the above statement can be rewritten to
$$
\Pi_S\Omega[\vec u/\vec x] \in \Rscr\,\vec u \Rightarrow
\RED_{D\,S\,\vec u}\, [\Omega].
$$
By Lemma~\ref{lm:red param subst}, $\RED_{D\,S\,\vec u}\, [\Omega] =
\RED_{D\,X^p\,\vec u}\, [\Omega, (\Rscr, \Pi_S\Omega, S')/X^p]$, which
means that $\Rscr$ indeed satisfies condition {\bf (b)} above, and
therefore $\Pi\Omega \in \RED_C[\Omega]$.

\paragraph{\bf CASE II: $n > 0$} 
To show that $\Xi \in \RED_C[\Omega]$ in this case, we need to show
that all its reducts are in $\RED_C[\Omega]$ and that $\Xi$ is
normalizable.  The latter follows from the former by
Lemma~\ref{lm:red-norm} and Definition~\ref{def:norm}, so in the
following we need only to show the former.

Note that in this case, we do not need to distinguish cases based on
whether $C$ is headed by a parameter or not.  To see why, suppose $C =
X^p\,\vec t$ for some parameter $X^p$.  If $X^p \not \in supp(\Omega)$
then to show $\Xi \in \RED_C[\Omega]$ we need to show that it is
normalizable, which means that we need to show that all its reducts
are normalizable.  But since all reducts of $\Xi$ has the same type
$X^p\,\vec t$, showing their normalizability amounts to the same thing
as showing that they are in $\RED_C[\Omega]$.  If $X^p \in
supp(\Omega)$, then to show $\Xi \in \RED_C[\Omega]$ we need to show
that $\Xi \in \Rscr$. Then by {\bf CR3}, it is enough to show that all
reducts of $\Xi$ are in $\Rscr$, which is the same as showing that all
reducts of $\Xi$ are in $\RED_C[\Omega]$.

Since the applicable reduction rules to $\Xi$ are driven by the shape
of $\Pi\Omega$, and since $\Pi\Omega$ is determined by $\Pi$, we shall
perform case analysis on $\Pi$ in order to determine the possible
reduction rules that apply to $\Xi$, and show in each case that the
reduct of $\Xi$ is in the same parametric reducibility set.  There are
several main cases depending on whether $\Pi$ ends with a rule acting
on a cut formula $B_i$ or not.  Again, when we refer to $B_i$, without
loss of generality, we assume $i=1$.

In the following, we say that an instance of $\coindLP$ is
\emph{trivial} if it applies to a formula $Y^q \,\vec u$ for some
$\vec u$, but $Y^q \not \in supp(\Omega)$. Otherwise, we say that it
is non-trivial. 

\paragraph{\bf II.1} Suppose $\Pi$ ends with a left rule, other than
$\cL$, $\wL$ and a non-trivial $\coindLP$, on $B_1$ and $\Pi_1$ ends
with a right-introduction rule.  There are several subcases depending
on the logical rules that are applied to $B_1$. We show here the
non-trivial cases:
\begin{trivlist}
\item[\fbox{$\oimpR/\oimpL$}] Suppose $\Pi_1$ and $\Pi$ are
$$
\infer[\oimpR] {\Seq{\Delta_1}{B_1'\Omega \oimp B_1''\Omega}} {
  \deduce{\Seq{\Delta_1,B_1'\Omega}{B_1''\Omega}}{\Pi_1'} } \qquad
\infer[\oimpL.]  {\Seq{B_1'\oimp B_1'',B_2,\dots, B_n,\Gamma}{C}} {
  \deduce{\Seq{B_2,\dots,\Gamma}{B_1'}}{\Pi'} &
  \deduce{\Seq{B_1'',B_2,\dots,\Gamma}{C}}{\Pi''} }
$$
Let $\Xi_1 = mc(\Pi_2,\ldots,\Pi_n,\Pi'\Omega$. 
Then $\Xi_1 \in \RED_{B_1'}[\Omega]$ by induction hypothesis since
$\measure{\Pi'} < \measure{\Pi}$ and therefore $\Mscr(\Xi_1) <
\Mscr(\Xi)$.  Since $\Pi_1 \in \RED_{B_1}[\Omega]$, by
Definition~\ref{def:param red}, we have
$$
\Pi_1' \in \RED_{B_1'}[\Omega] \Rightarrow \RED_{B_1''}[\Omega]
$$
and therefore the derivation $\Xi_2= mc(\Xi_1,\Pi_1')$ with end sequent 
$\Seq{\Delta_1,\ldots,\Delta_n,\Gamma\Omega}{B_1''\Omega}$
is in $\RED_{B_1''}[\Omega]$.  Let $\Xi_3 = mc(\Xi_2, \Pi_2,\ldots,
\Pi_n,\Pi''\Omega)$.

The reduct of $\Xi$ in this case is the derivation $\Xi'$:
\begin{displaymath}
  \infer=[\cL]
  {\Seq{\Delta_1,\ldots,\Delta_n,\Gamma\Omega}{C\Omega}}
  {
    \deduce
    {\Seq{\Delta_1,\ldots, \Delta_n,\Gamma\Omega,
        \Delta_2,\ldots, \Delta_n,\Gamma\Omega}
      {C\Omega}
    }
    {\Xi_3}
  }
\end{displaymath}
By the induction hypothesis, we have $\Xi_3 \in \RED_{C}[\Omega]$, and
therefore, by Lemma~\ref{lm:red-norm}, it is normalizable.  By
Definition~\ref{def:norm}, this means that $\Xi'$ is normalizable and
by Definition~\ref{def:param red}, $\Xi' \in \RED_C[\Omega]$.

\item[\fbox{$\forallL/\forallR$}] Suppose $\Pi_1$ and $\Pi$ are
$$
\infer[\forallR] {\Seq{\Delta_1}{\forall x.B_1'\Omega}} {
  \deduce{\Seq{\Delta_1}{B_1'\Omega[y/x]}}{\Pi_1'} } \qquad \qquad
\infer[\forallL] {\Seq{\forall x.B_1',B_2,\dots,B_n,\Gamma}{C}} {
  \deduce{\Seq{B_1'[t/x],B_2,\dots,B_n,\Gamma}{C}}{\Pi'} }
$$
The reduct of $\Xi$ in this case is
$$
\Xi' = mc(\Pi_1'[t/y], \Pi_2, \ldots, \Pi_n, \Pi'\Omega).
$$
Since $\Pi_1' \in \RED_{B_1'[y/x]}[\Omega]$, by
Lemma~\ref{lm:red-subst} we have
$$\Pi_1'[t/y] \in \RED_{B_1'[t/x]}[\Omega]$$
Note that $\measure{\Pi'} < \measure{\Pi}$, so we can apply the
induction hypothesis to obtain $\Xi' \in \RED_{C}[\Omega]$.

\item[\fbox{$\eqR/\eqL$}] Suppose $\Pi_1$ and $\Pi$ are
$$
\infer[\eqR] {\Seq{\Delta_1}{s = t}} {} \qquad \qquad \infer[\eqL]
{\Seq{s = t,\dots,B_n,\Gamma}{C}} {\left\{\raisebox{-1.5ex} {
      \deduce{\Seq{B_2\rho,\dots,B_n\rho,\Gamma\rho}{C\rho}}
      {\Pi^\rho} } \right\}_\rho }
$$
Note that in this case $s$ must be the same term as $t$, and therefore
both are unifiable by any substitution. 
Let $\Pi'$ be the derivation:
$$
\infer[subst]
{\Seq{B_2,\ldots,B_n,\Gamma}{C}}
{
\left\{
\raisebox{-1.3ex}{
\deduce{\Seq{B_2\rho,\ldots,B_n\rho,\Gamma\rho}{C\rho}}{\Pi^\rho}
}
\right\}_\rho
}
$$
and let $\Xi_1 = mc(\Pi_2,\ldots,\Pi_n,\Pi'\Omega$.
Since $\measure{\Pi'} = \measure{\Pi}$ and since the total size of cut 
formulas in $\Xi_1$ is smaller than in $\Xi$, by the induction hypothesis, we have
$\Xi_1 \in \RED_C[\Omega]$.
Then the reduct of $\Xi$ in this case is the derivation $\Xi'$:
$$
\infer=[\wL]
{\Seq{\Delta_1,\Delta_2,\ldots,\Delta_n,\Gamma}{C}}
{
\deduce{\Seq{\Delta_2,\ldots,\Delta_n,\Gamma}{C}}{\Xi_1}
}
$$
which is also in $\RED_C[\Omega]$, by the definition of parametric reducibility.

\item[\fbox{$\indR/\indL$}] Suppose $\Pi_1$ and $\Pi$ are the derivations
$$
\infer[\indR] {\Seq {\Delta_1} {p\,\vec t}} { \deduce{\Seq
    {\Delta_1}{D\,X^p\,\vec t}}{\Pi_1'} } \qquad \infer[\indL]
{\Seq{p\,\vec{t}, \Gamma}{C}} {
  \deduce{\Seq{D\,S\,\vec{x}}{S\,\vec{x}}}{\Pi_S} &
  \deduce{\Seq{S\,\vec{t}, \Gamma}{C}}{\Pi'} }
$$
where $p\,\vec y \defmu D\,p\,\vec y$ and $X^p$ is a new parameter not
occuring in the end sequent of $\Pi_1$ (we can assume w.l.o.g.\ that
$X^p \# \Omega$ and that it does not occur either in the end sequent
of $\Pi$).  Then $\Pi\Omega$ is the derivation
$$
\infer[\indL] {\Seq{p\,\vec{t}, \Gamma\Omega}{C\Omega}} {
  \deduce{\Seq{D\,S'\,\vec{x}}{S'\,\vec{x}}}{\Pi_S\Omega} &
  \deduce{\Seq{S'\,\vec{t}, \Gamma\Omega}{C\Omega}}{\Pi'\Omega} }
$$
where $S' = S\Omega$.  
Let $\Xi_1 = mc(\Pi_1'[(\Pi_S\Omega,S')/X^p], \Pi_S\Omega[\vec t/\vec x]$. 
Then the reduct of $\Xi$ in this case is the derivation
$$
\Xi' = mc(\Xi_1,\Pi_2,\ldots,\Pi_n, \Pi'\Omega).
$$

Since $\measure{\Pi_S[\vec u/\vec x]} \leq \measure{\Pi_S} <
\measure{\Pi}$ by the induction hypothesis, we have
\begin{equation}
  \label{eq:caseII.1-a}
  \Pi_S\Omega[\vec u/\vec x] \in \RED_{D\,S\,\vec u}\, [\Omega] \Rightarrow \RED_{S\,\vec u}\, [\Omega].
\end{equation}
Let $\Rscr = \{\Psi \mid \Psi \in \RED_{S\,\vec u}\, [\Omega] \hbox{ for
  some $\vec u$} \}$.  Then by Lemma~\ref{lm:red candidate}, $\Rscr$
is a reducibility candidate of type $S'$.  Moreover, by
Lemma~\ref{lm:red param subst}, we have
$$
\RED_{D\,S\,\vec u}\, [\Omega] = \RED_{D\,X^p\,\vec u}\, [\Omega, (\Rscr,
\Pi_S\Omega,S')/X^p].
$$
This, together with Statement~\ref{eq:caseII.1-a} above, implies that
\begin{equation}
  \label{eq:caseII.1-b}
  \Pi_S\Omega[\vec u/\vec x] \in \RED_{D\,X^p\,\vec u}\, [\Omega, (\Rscr, \Pi_S\Omega,S')/X^p] \Rightarrow \Rscr \,\vec u
\end{equation}
for every $\vec u$.

Since $\Pi_1 \in \RED_{p\,\vec t}\, [\Omega]$, it follows from
Definition~\ref{def:param red} that for every reducibility candidate
$(\Sscr : I)$ and $\Pi_I$ such that
$$
\Pi_I[\vec u/\vec x] \in \RED_{D\,X^p\,\vec u}\, [\Omega, (\Sscr,
\Pi_I,I)/X^p] \Rightarrow \Sscr\,\vec u \hbox{ for every $\vec u$, }
$$
we have
$$
mc(\Pi_1'[(\Pi_I,I)/X^p], \Pi_I[\vec t/\vec x]) \in \Sscr\,\vec t.
$$
Substituting $\Rscr$ for $\Sscr$, $\Pi_S\Omega$ for $\Pi_I$ and $S'$
for $I$, and using Statement~\ref{eq:caseII.1-b} above, we obtain:
$$
\Xi_1 = mc(\Pi_1'[(\Pi_S\Omega,S')/X^p], \Pi_S\Omega[\vec t/\vec x])
\in \Rscr\,\vec t = \RED_{S\,\vec t}\, [\Omega].
$$
Since $\measure{\Pi'} < \measure{\Pi}$, we can then apply the
induction hypothesis to conclude that $\Xi' \in \RED_{C}[\Omega]$.

\item[\fbox{$\coindR/\coindL$}] Suppose $\Pi_1$ and $\Pi$ are
$$
\infer[\coindR] {\Seq{\Delta_1}{p\,\vec{t}}} {
  \deduce{\Seq{\Delta_1}{S\,\vec{t}}}{\Pi_1'} &
  \deduce{\Seq{S\,\vec{x}}{D\,S\,\vec{x}}}{\Pi_S} } \qquad \qquad
\infer[\coindL] {\Seq{p\,\vec{t}, B_2, \dots, \Gamma}{C}}
{\deduce{\Seq{D\,X^p\,\vec{t},B_2,\dots, \Gamma}{C}}{\Pi'}}
$$
where $p\,\vec{y}\defnu D\,p\,\vec{y}$ and $X^p$ is a parameter not
already occuring in the end sequent of $\Pi$ (and w.l.o.g.\ assume also
$X^p \# \Omega$ and $X^p$ not occuring in $\Delta_i$ or $B_i$).  Then
$\Pi\Omega$ is
$$
\infer[\coindL.]  {\Seq{p\,\vec{t}, B_2\Omega, \dots,
    \Gamma\Omega}{C\Omega}}
{\deduce{\Seq{D\,X^p\,\vec{t},B_2\Omega,\dots,
      \Gamma\Omega}{C\Omega}}{\Pi'\Omega}}
$$
Since $\Pi_1 \in \RED_{p\,\vec t}\, [\Omega]$, by
Definition~\ref{def:param red}, there exists a reducibility candidate
$(\Rscr : S)$ such that $\Pi_1' \in \Rscr$ and such that for every
$\vec u$,
$$
\Pi_S[\vec u/\vec x] \in \Rscr\,\vec u \Rightarrow \RED_{D\,X^p\,\vec
  u}\, [\Omega, (\Rscr,\Pi_S,S)/X^p)].
$$
Let $\Omega' = [\Omega, (\Rscr, \Pi_S,S)/X^p]$.  Then by
Lemma~\ref{lm:clo-ext-coind}, $\Omega'$ is definitionally closed.

Let $\Xi_1 = mc(\Pi_1', \Pi_S[\vec t/\vec x])$. By the definitional
closure of $\Omega'$, we have that $\Xi_1 \in \RED_{D\,X^p\,\vec
  t}\, [\Omega']$.

The reduct of $\Xi$ in this case is the derivation
$$
\Xi' = mc(\Xi_1,\Pi_2,\ldots,\Pi_n,\Pi'\Omega').
$$
Note that since $X^p$ does not occur in $\Delta_i$ or $B_i$, by
Lemma~\ref{lm:red vacuous}, we have that
$$
\Pi_i \in \RED_{B_i}[\Omega] = \RED_{B_i}[\Omega']
$$
for every $i\in\{2,\ldots,n\}$.  Therefore, by induction hypothesis,
we have that
$$
\Xi' \in \RED_C[\Omega'].
$$
But since $X^p$ is also new for $C$, we have $\RED_C[\Omega'] =
\RED_C[\Omega]$, and therefore
$$
\Xi' \in \RED_C[\Omega].
$$
\end{trivlist}

\paragraph{\bf II.2} $\Pi$ ends with a left rule, other than $\cL$,
$\wL$ and a non-trivial instance of $\coindLP$, acting on $B_1$, and
$\Pi_1$ ends with a left-rule or $subst$.

Note that in these cases, the reducts always end with a left-rule.
The proof for the following cases abide to the same pattern: we first
establish that the premise derivations of the reduct are either
normalizable or in certain reducibility sets. We then proceed to show
that the reduct itself is reducible by applying to the closure
conditions of reducibility under applications of left-rules.  For the
latter, we need to distinguish three cases depending on $C$: If $C =
X^p\,\vec t$ for some $X^p\in supp(\Omega)$, then closure under
left-rules is guaranteed by {\bf C5};
if $X^p \not \in supp(\Omega)$ then we need to show that the reduct
is normalizable, and the closure condition under left-rules is
guaranteed by the definition of normalizability. 
  Otherwise, $C$ is not headed by any parameter, and in this case,
the closure condition follows from {\bf  P6}.
We shall explicitly do these case analysis in one of the subcases
below, but will otherwise leave them implicit.  We show the
non-trivial subcases only; other cases can be proved by
straightforward applications of the induction hypothesis.

\begin{trivlist}
\item[\fbox{$\oimpL/\circ\Lscr$}] Suppose $\Pi_1$ is
$$
\infer[\oimpL] {\Seq{D_1 \oimp D_2,\Delta_1'}{B_1\Omega}} {
  \deduce{\Seq{\Delta_1'}{D_1}}{\Pi_1'} &
  \deduce{\Seq{D_2,\Delta_1'}{B_1\Omega}}{\Pi_1''} }
$$
Since $\Pi_1 \in \RED_{B_1}[\Omega]$, it follows from
Definition~\ref{def:param red} that $\Pi_1'$ is normalizable and
$\Pi_1'' \in \RED_{B_1}[\Omega]$.

Let $\Xi_1 = mc(\Pi_1'',\Pi_2,\ldots,\Pi_n,\Pi\Omega)$.  Since
$nd(\Pi_1'') < nd(\Pi_1)$, by induction hypothesis, $\Xi_1 \in
\RED_C[\Omega]$.

The reduct of $\Xi$ in this case is the derivation $\Xi'$:
\begin{displaymath}
  \infer[\oimpL]
  {\Seq{D_1 \oimp D_2,\Delta_1',\Delta_2,\ldots,\Gamma\Omega}{C\Omega}}
  {
    \infer=[\wL]
    {\Seq{\Delta_1',\ldots,\Gamma\Omega}{D_1}}
    {
      \deduce{\Seq{\Delta_1'}{D_1}}{\Pi_1'} 
    }
    &
    \deduce{\Seq{D_2,\Delta_1',\Delta_2,\ldots,\Gamma\Omega}{C\Omega}}
    {\Xi_1}
  }
\end{displaymath}
Since $\Pi_1'$ is normalizable, by Definition~\ref{def:norm} the left
premise derivation of $\Xi'$ is normalizable, and since reducibility
implies normalizability (Lemma~\ref{lm:red-norm}), the right premise
is also normalizable, hence $\Xi'$ is normalizable.

Now to show $\Xi' \in \RED_C[\Omega]$, we need to distinguish three
cases based on $C$:
\begin{itemize}
\item Suppose $C = X^p\,\vec t$ for some $X^p \in supp(\Omega)$ and
  $\Omega(X^p) = (\Rscr, \Pi_S,S)$. Then we need to show that $\Xi'
  \in \Rscr\,\vec t$. This follows from
  Definition~\ref{def:candidates}, more specifically, from {\bf CR5}
  and the fact that $\Xi_1 \in \RED_C[\Omega] = \Rscr\,\vec t$.

\item Suppose $C = X^p \,\vec t$ but $X^p \not \in supp(\Omega)$.
  Then we need to show that $\Xi'$ is normalizable. But this follows
  immediately from the normalizability of both of its premise
  derivations.

\item Suppose $C \not = X^p\,\vec t$ for any parameter $X^p$ and any
  terms $\vec t$. Since $\Xi_1 \in \RED_C[\Omega]$, by
  Definition~\ref{def:param red}, we have $\Xi' \in \RED_C[\Omega]$.
\end{itemize}

\item[\fbox{$\eqL/\circL$}] Suppose $\Pi_1$ is as shown below left.
Then the reduct of $\Xi$ in this case is shown below right, where
$\Xi^\rho = mc(\Pi_1^\rho,\Pi_2\rho,\ldots,\Pi_n\rho,\Pi\rho\Omega$.
$$
\infer[\eqL]{\Seq{s=t,\Delta_1'}{B_1\Omega}} {\left\{\raisebox{-1.5ex}
    {\deduce{\Seq{\Delta_1'\rho}{B_1\Omega\rho}} {\Pi^{\rho}}}
  \right\}_{\rho}}
\qquad
\infer[\eqL]
{\Seq{s=t,\Delta_1',\Delta_2,\dots,\Gamma\Omega}{C\Omega}}
{\left\{\raisebox{-1.5ex}
    {\deduce{\Seq{\Delta_1'\rho,\dots,\Gamma\Omega\rho}{C\Omega\rho}}
      {\Xi^{\rho}}} \right\}_{\rho}}
$$
$\Xi^\rho \in \RED_{C\rho}[\Omega]$ by the induction hypothesis
(since $nd(\Pi_1^\rho) < nd(\Pi_1)$ and the other measures are
non-increasing). Hence, the reduct of $\Xi$ is in  $\RED_C[\Omega]$
by the definition of parametric reducibility.

\item[\fbox{$\indL/\circL$}] Suppose $\Pi_1$ is
$$
\infer[\indL.]  {\Seq{p\,\vec{t}, \Delta_1'}{B_1\Omega}} {
  \deduce{\Seq{D\,S\,\vec{x}}{S\,\vec{x}}}{\Pi_S} &
  \deduce{\Seq{S\,\vec{t}, \Delta_1'}{B_1\Omega}}{\Pi_1'} }
$$
Since $\Pi_1 \in \RED_{B_1}[\Omega]$, we have that $\Pi_S$ is
normalizable and $\Pi_1' \in \RED_{B_1}[\Omega]$.  Let $\Xi_1$ be the
derivation
$$
mc(\Pi_1',\Pi_2,\ldots,\Pi_n,\Pi\Omega).
$$
Then $\Xi_1 \in \RED_{B_1}[\Omega]$ by the induction hypothesis, since
$nd(\Pi_1') < nd(\Pi_1)$.
Therefore the reduct of $\Xi$
$$
\infer[\indL] {\Seq{p\,\vec{u},\Delta_1',
    \dots,\Delta_n,\Gamma\Omega}{C\Omega}} {
  \deduce{\Seq{D\,S\,\vec{x}}{S\,\vec{x}}}{\Pi_S} &
  \deduce{\Seq{S\,\vec{u},\Delta_1',
      \ldots,\Delta_n,\Gamma\Omega}{C\Omega}}{\Xi_1} }
$$
is also in $\RED_{C}[\Omega]$.
 
\end{trivlist}

\paragraph{\bf II.3} $\Pi$ ends with a left rule, other than $\cL$,
$\wL$ and a non-trivial instance of $\coindLP$, acting on $B_1$, and
$\Pi_1$ ends with $mc$ or $init$: These cases follow straightforwardly
from the induction hypothesis.

\paragraph{\bf II.4} Suppose $\Pi$ ends with a non-trivial application
of $\coindLP$ on $B_1$. That is, $B_1 = X^p\,\vec t$, for some $X^p
\in supp(\Omega)$ and some $\vec t$, and $\Pi$ is
$$
\infer[\coindLP] {\Seq{X^p\,\vec t, B_2, \ldots, B_n, \Gamma}{C}} {
  \deduce{\Seq{D\,X^p\,\vec t, B_2,\ldots, B_n, \Gamma}{C}}{\Pi'} }
$$
where $p\,\vec x \defnu D\,p\,\vec x$.  Suppose $\Omega(X^p) = (\Rscr,
\Pi_S, S)$.  Then $\Pi\Omega$ is
$mc(mc(\idrv_{S\,\vec t}, \Pi_S[\vec t/\vec x]), \Pi'\Omega$.
Let $\Xi_1 = mc(\Pi_1, mc(\idrv_{S\,\vec t}, \Pi_S[\vec t/\vec x])$.
Note that $\Xi_1$ has exactly one reduct, that is,
$$
\Xi_2 = mc(mc(\Pi_1,\idrv_{S\,\vec t}), \Pi_S[\vec t/\vec x]).
$$
Note that $mc(\Pi_1,\idrv_{S\,\vec t})$ also has exactly one reduct,
namely, $\Pi_1$. Since $\Pi_1\in \RED_{X^p\,\vec t}\, [\Omega] =
\Rscr\,\vec t$, this means, by {\bf CR3}, that $mc(\Pi_1,\idrv_{S\,\vec t})$ is in
$\Rscr\,\vec t$.  Since $\Omega$ is definitionally closed, we have
that $\Xi_2 \in \RED_{D\,X^p\,\vec t}\, [\Omega]$.  And since $\Xi_2$ is
the only reduct of $\Xi_1$, this also means that, by
Definition~\ref{def:param red}, $\Xi_1 \in \RED_{D\,X^p\,\vec
  t}\, [\Omega]$.

The reduct of $\Xi$, \ie the derivation $mc(\Xi_1,\Pi_2,\ldots,\Pi_n,\Pi'\Omega)$
is therefore in $\RED_C[\Omega]$ by the induction hypothesis.

\paragraph{\bf II.5} Suppose $\Pi$ ends with $\wL$ or $\cL$ acting on
$B_1$, or $init$. Then $\Pi\Omega$ also ends with the same rule.  The
cut reduction rule that applies in this case is either $-/\wL$,
$-/\cL$ or $-/init$. In these cases, parametric reducibility of the
reducts follow straightforwardly from the assumption (in case of
$init$), the induction hypothesis and Definition~\ref{def:param red}.

\paragraph{\bf II.6} Suppose $\Pi$ ends with $mc$.  Then $\Pi\Omega$
also ends with $mc$. The reduction rule that applies in this case is
the reduction $-/mc$.  Parametric reducibility of the reduct in this
case follows straightforwardly from the induction hypothesis and
Definition~\ref{def:param red}.

\paragraph{\bf II.7} Suppose $\Pi$ ends with $subst$ or a rule acting on a
formula other than a cut formula.
Most cases follow straightforwardly from the induction hypothesis,
Lemma~\ref{lm:red-norm} and Lemma~\ref{lm:red-subst} (which is needed
in the reduction case $-/\eqL$ and $-/subst$).  We show the interesting subcases
here:

\begin{trivlist}

\item[\fbox{$-/\indRP$}] Suppose $\Pi$ ends with a non-trivial $\indRP$,
  i.e., $\Pi$ is
$$
\infer[\indRP] {\Seq{B_1,\ldots,B_n,\Gamma}{X^p\,\vec t}} {
  \deduce{\Seq{B_1,\ldots,B_n,\Gamma}{D\,X^p\,\vec t}}{\Pi'} }
$$
where $p\,\vec x \defmu D\,p\,\vec x$ and $X^p \in supp(\Omega)$.
Suppose $\Omega(X^p) = (\Rscr,\Pi_S,S)$. Then $\Pi\Omega$ is
the derivation $mc(\Pi'\Omega, \Pi_S[\vec t/\vec x]$.
The reduct of $\Xi$ in this case is the derivation
$$
\Xi' = mc(mc(\Pi_1,\ldots,\Pi_n,\Pi'\Omega), \Pi_S[\vec t/\vec x]).
$$
By the induction hypothesis, we have
$
mc(\Pi_1,\ldots, \Pi_n,\Pi'\Omega) \in \RED_{D\,X^p\,\vec t}\, [\Omega].
$
This, together with the definitional closure of $\Omega$, implies that
$\Xi' \in \Rscr\,\vec t = \RED_{X^p\,\vec t}\, [\Omega]$.

\item[\fbox{$-/\indR$}] Suppose $\Pi$ is
$$
\infer[\indR] {\Seq{B_1,\ldots,B_n,\Gamma}{p\,\vec t}} {
  \deduce{\Seq{B_1,\ldots,B_n,\Gamma}{D\,X^p\,\vec t}}{\Pi'} }
$$
where $p\,\vec y \defmu D\,p\,\vec y$.  Without loss of generality we
can assume that $X^p$ is chosen to be sufficiently fresh (e.g., not
occurring in $\Omega$, $\Delta_1$, $B_1$, etc.).  
Let $\Xi_1 = mc(\Pi_1,\ldots,\Pi_n,\Pi'\Omega$. 
Then the reduct of $\Xi$ is the derivation $\Xi'$
$$
\infer[\indR .]  {\Seq{\Delta_1,\ldots,\Delta_n,\Gamma\Omega}{p\,\vec
    t}} {
  \deduce{\Seq{\Delta_1,\ldots,\Delta_n,\Gamma\Omega}{D\,X^p\,\vec
      t}}{\Xi_1} }
$$
To show that $\Xi' \in \RED_C[\Omega]$, we first need to show that it
is normalizable.  This follows straightforwardly from the induction
hypothesis (which shows that $\Xi_1 \in \RED_{D\,X^p\,\vec
  t}\, [\Omega]$) and Lemma~\ref{lm:red-norm}.  It then remains to show
that
$$
\Xi_2 = mc(\Xi_1[(\Pi_S,S)/X^p], \Pi_S[\vec t/\vec x]) \in \Rscr\,\vec
t
$$
for every reducibility candidate $(\Rscr : S)$ and every $\Pi_S$ such
that
\begin{equation}
  \label{eq:II.7.a.1}
  \Pi_S[\vec u/\vec x] \in \RED_{D\,X^p\,\vec u}\, [\Omega, (\Rscr,\Pi_S,S)/X^p] \Rightarrow \Rscr\,\vec u,
  \hbox{ for every $\vec u$.}
\end{equation}

So suppose $(\Rscr,\Pi_S,S)$ satisfies Statement~\ref{eq:II.7.a.1}
above.  Let $\Omega' = [\Omega, (\Rscr,\Pi_S,S)/X^p]$. 
By Lemma~\ref{lm:clo-ext-ind}, $\Omega'$ is definitionally closed.  Note
that since we assume that $X^p$ is a fresh parameter not occuring in
$B_i$, we have
$\RED_{B_i}[\Omega] = \RED_{B_i}[\Omega']$ 
by Lemma~\ref{lm:red vacuous}, and
$
\Pi_i[(\Pi_S,S)/X^p] = \Pi_i \in \RED_{B_i}[\Omega']
$
by Lemma~\ref{lm:param subst vacuous}, for every $i \in
\{1,\ldots,n\}$.  Therefore, by the induction hypothesis we have:
$$
\Xi_1[(\Pi_S,S)/X^p] = mc(\Pi_1,\ldots,\Pi_n,\Pi'\Omega') \in
\RED_{D\,X^p\,\vec t}\, [\Omega'].
$$
This, together with the definitional closure of $\Omega'$, implies that
$
\Xi_2 \in \Rscr\,\vec t.
$

\item[\fbox{$-/\coindLP$}] Suppose $\Pi$ ends with a non-trivial $\coindLP$,
  i.e., $\Pi$ is
$$
\infer[\coindLP] {\Seq{B_1,\ldots,B_n, X^p\,\vec t, \Gamma'}{C}} {
  \deduce{\Seq{B_1,\ldots,B_n, D\,X^p\,\vec t, \Gamma'}{C}}{\Pi'} }
$$
where $p\,\vec x \defnu D\,p\,\vec x$ and $X^p \in supp(\Omega)$.
Suppose $\Omega(X^p) = (\Rscr, \Pi_S,S)$.  Then $\Pi\Omega$ is
$$
mc(mc(\idrv_{S\,\vec t}, \Pi_S[\vec t/\vec x]),\Pi'\Omega).
$$
Let $\Xi_1 = mc(\idrv_{S\,\vec t}, \Pi_S[\vec t/\vec x])$.  By {\bf
  CR4}, $\idrv_{S\,\vec t} \in \Rscr \,\vec t$, and therefore, by
definitional closure of $\Omega$, we have $\Xi_1 \in
\RED_{D\,X^p\,\vec t}\, [\Omega]$.  The reduct of $\Xi$ in this case is
$$
mc(\Xi_1,\Pi_1,\ldots,\Pi_n,\Pi'\Omega)
$$
which is in $\RED_C[\Omega]$ by the induction hypothesis.

\item[\fbox{$-/\coindR$}] Suppose $\Pi$ is
$$
\infer[\coindR] {\Seq{B_1,\dots,B_n,\Gamma}{p\,\vec{t}}} {
  \deduce{\Seq{B_1,\dots,B_n,\Gamma}{S\,\vec{t}}}{\Pi'} &
  \deduce{\Seq{S\,\vec{x}}{D\,S\,\vec{x}}}{\Pi_S} } \enspace 
$$
where $p\,\vec{y} \defnu D\,p\,\vec{y}$.  Let $S' = S\Omega$.  The
derivation $\Pi\Omega$ in this case is 
$$
\infer[\coindR]
{\Seq{B_1\Omega,\dots,B_n\Omega,\Gamma\Omega}{p\,\vec{t}}} {
  \deduce{\Seq{B_1\Omega,\dots,B_n\Omega,\Gamma\Omega}{S'\,\vec{t}}}{\Pi'\Omega}
  & \deduce{\Seq{S'\,\vec{x}}{D\,S'\,\vec{x}}}{\Pi_S\Omega} }
$$
Let $\Xi_1$ be the derivation $mc(\Pi_1,\ldots,\Pi_n,\Pi'\Omega)$.  By
the induction hypothesis, $\Xi_1 \in \RED_{S\,\vec t}\, [\Omega]$ and
$\Pi_S\Omega \in \RED_{D\,S\,\vec x}\, [\Omega]$, hence both $\Xi_1$ and
$\Pi_S\Omega$ are also normalizable by Lemma~\ref{lm:red-norm}.  The
 reduct of $\Xi$ is the derivation $\Xi'$
$$
\infer[\coindR .]
{\Seq{\Delta_1,\dots,\Delta_n,\Gamma\Omega}{p\,\vec{t}}} {
  \deduce{\Seq{\Delta_1,\dots,\Delta_n,\Gamma\Omega}{S'\,\vec{t}}}{\Xi_1}
  & \deduce{\Seq{S'\,\vec{x}}{D\,S'\,\vec{x}}}{\Pi_S\Omega} }
$$
Let $X^p$ be a parameter fresh for  $\Omega$,
$\Gamma$, $\Delta_i$ and $B_i$.

To show that $\Xi' \in \RED_C[\Omega]$ we must first show that it is
normalizable.  This follows from immediately from normalizability of
$\Xi_1$ and $\Pi_S\Omega$.  Then we need to find a reducibility
candidate $(\Rscr : S')$ such that
\begin{description}
\item[(a)] $\Xi_1 \in \Rscr$, and
\item[(b)] $\Pi_S\Omega[\vec u/\vec x] \in \Rscr\,\vec u \Rightarrow
  \RED_{D\,X^p\,\vec u}\, [\Omega,(\Rscr,\Pi_S,S)/X^p]$.
\end{description}
Let
$
\Rscr = \{\Psi \mid \Psi \in \RED_{S\,\vec u}\, [\Omega] \}.
$
As in case {\bf I.2.c}, we show, using Lemma~\ref{lm:red candidate},
that $\Rscr$ is a reducibility candidate of type $S'$.  By the
induction hypothesis, we have $\Xi_1 \in \Rscr$, so $\Rscr$ satisfies
{\bf (a)}.  Using the same argument as in case {\bf I.2.c} we can show
that $\Rscr$ also satisfies {\bf (b)}, \ie by appealing to the
induction hypothesis, applied to $\Pi_S$.
\end{trivlist}
\qed
\end{proof}

\begin{corollary}
  Every derivation is reducible.
\end{corollary}
\begin{proof}
  The proof follows from Lemma~\ref{lm:comp}, by setting $n=0$ and 
  $\Omega$ to the empty candidate substitution. \qed
\end{proof}
Since reducibility implies cut-elimination and since every cut-free
derivation can be turned into a $subst$-free derivation
(Lemma~\ref{lm:subst-elimination}), it follows that every proof can be
transformed into a cut-free and $subst$-free derivation.

\begin{corollary}
  \label{cor:cut-elimination}
  Given a fixed definition, a sequent has a derivation in $\Linc$ if
  and only if it has a cut-free and $subst$-free derivation.
\end{corollary}

The consistency of $\Linc$ is an immediate consequence of
cut-elimination.  By consistency we mean the following: given a fixed
definition and an arbitrary formula $C$, it is not the case
that both $C$ and $C\oimp \bot$ are provable.

\begin{corollary}
  \label{cor:consistency}
  The logic $\Linc$ is consistent.
\end{corollary}

\section{Related work and conclusions}
\label{sec:lrel}
 
Of course, there is a long association between mathematical logic and
inductive definitions~\cite{Acz77} and in particular with
proof-theory, starting with the Takeuti's conjecture, the earliest
relevant entry for our purposes being Martin-L\"of's original
formulation of the theory of \emph{iterated inductive
  definitions}~\cite{martin-lof71sls}.  From the representation of
algebraic types~\cite{Bohm85} and the introduction of (co)inductive
types in system F~\cite{Mendler87,Geuvers92}, (co)induction/recursion
became mainstream and made it into type-theoretic proof assistants such as
Coq~\cite{PaulinMohring93}, first via a primitive recursive operator,
but eventually in the let-rec style of functional programming
languages, as in Gimenez's \emph{Calculus of Infinite
  Constructions}~\cite{Gim96phd}.  Unlike works in these
type-theoretic settings, we put less emphasis on proof terms and
strong normalization; in fact, our cut elimination procedure is
actually a form of weak normalization, in the sense that our procedure
only guarantees termination with respect to a particular strategy,
i.e, by reducing the lowest cuts in a derivation tree.  Our notion of
equality, which internalizes unification in its left introduction
rule, departs from the more traditional notion of equality.  As a
consequence of these differences, it is not at all obvious that strong
normalization proofs for term calculi with (co-)inductive types can be
adapted straightforwardly to our setting.


Baelde and Miller have recently introduced an extension of mulitplicative-additive
linear logic with least and greatest fixed points~\cite{baelde07lpar}, called $\mu$MALL. 
In that work, cut elimination is proved indirectly via a second-order 
encoding of the least and the greatest fixed point operators 
into higher-order linear logic and via an appeal to completeness of
focused proofs for higher-order linear logic.
Such an encoding can also be used for proving cut elimination
for $\Linc$, but as we noted earlier, our main concern here is
to provide a basis for cut elimination for (orthogonal) extensions
of $\Linc$ with the $\nabla$-quantifier, for which there are
currently no known encodings into higher-order (linear) logic.
Baelde has also given a direct cut-elimination proof for
$\mu$MALL~\cite{Baelde09CoRR}. The proof uses a notion of orthogonality in
the definition of reducibility, defined via classical negation, 
so it is not clear if it can be adapted straightforwardly to 
the intuitionistic setting like ours. 

Circular proofs are also connected with the proof-theory of
{fixed point logics} and {process
  calculi}~\cite{Santocanale02,Sprenger03}, as well as in traditional
sequent calculi such as in~\cite{BrotherstonS07}.  The issue is the
equivalence between systems with local vs.\ global induction, that is,
between fixed point rules vs.\ well-founded and guarded induction (\ie
circular proofs). In the traditional sequent calculus, it is unknown
whether every global inductive proof can be translated into a local
one.




In higher order logic (co)inductive definitions are usually obtained
via the Tarski set-theoretic fixed point construction, as realized for
example in Isabelle/HOL~\cite{Paulson97}.  As we mentioned before,
those approaches are at odd with HOAS even at the level of the
syntax. This issue has originated a research field in its own and we
only mention the main contenders: in the {Twelf}
system~\cite{Schurmann09}
the LF type theory is used to encode deductive systems as judgments
and to specify meta-theorems as relations (type families) among them;
a logic programming-like interpretation provides an operational
semantics to those relations, so that an external check for totality
(incorporating termination, well-modedness and
coverage~\cite{SchurmannP03,Pientka05}) verifies that the given
relation is indeed a realizer for that theorem. Coinduction is still
unaccounted for and may require a switch to a different operational
semantics for LF.  There exists a second approach to reasoning in LF
that is built on the idea of devising an explicit (meta-)meta-logic
($\Momega$) for reasoning (inductively) about the
framework~\cite{S00}. It can be seen as a constructive first-order
inductive type theory, whose quantifiers range over possibly open LF
objects.  In this calculus it is possible to express and inductively
prove meta-logical properties of an object level system.  $\Momega$
can be also seen as a dependently-typed functional programming
language, and as such it has been refined into the
\emph{Delphin} programming language~\cite{PosSch08}.  In a similar
vein \emph{Beluga}~\cite{Pientka08} is based on context modal
logic~\cite{NanevskiTOCL}, which provides a basis for a different
foundation for programming with HOAS and dependent types. Because all
of these systems are programming languages, we refrain from a deeper
discussion. We only note that systems like Delphin or Beluga separate
data from computations. This means they are always based on eager
evaluation, whereas co-recursive functions should be interpreted
lazily.  Using standard techniques such as \emph{thunks} to simulate
lazy evaluation in such a context seems problematic (Pientka, personal
communication).

\emph{Weak higher-order abstract syntax}~\cite{despeyroux94lpar} is an
approach that strives to co-exist with an inductive setting
.  The problem of negative occurrences in datatypes is handled
by replacing them with a new type.  Similarly for 
hypothetical judgments
, although  \emph{axioms} are needed to reason about
them, to mimic what is inferred by the cut rule in
our architecture.  Miculan \etal's framework~\cite{HonsellMS01}
embraces this \emph{axiomatic} approach extending Coq with the
``theory of contexts'' (ToC).  The theory includes axioms for the the
reification of key properties of names akin to
\emph{freshness}. Furthermore, higher-order induction and recursion
schemata on expressions are also assumed.
\emph{Hybrid}~\cite{Ambler02,Hybrid} is a $\lambda$-calculus on top of
Isabelle/HOL which provides the user with a \emph{Full} HOAS syntax,
compatible with a classical (co)-inductive setting. $\Linc$ improves
on the latter on several counts. First it disposes of Hybrid notion of
\emph{abstraction}, which is used to carve out the ``parametric''
function space from the full HOL function space.  Moreover it is not
restricted to second-order abstract syntax, as the current Hybrid
version is (and as ToC cannot escape from being). Finally, at higher
types, reasoning via $\eqL$ and fixed points is more powerful than
inversion, which does not exploit higher-order unification.

\emph{Nominal logic} gives a different foundation to programming and
reasoning with \emph{names}. It can be presented as a first-order
theory~\cite{pitts03ic}, which includes primitives for variable
renaming and freshness, and a (derived) ``new'' freshness
quantifier.  It is endowed of  natural principles of 
structural induction and recursion over
syntax~\cite{Pitts06}.  Urban \etal have engineered a \emph{nominal
  datatype package} inside Isabelle/HOL~\cite{Nominal}
analogous to the standard datatype package but defining equivalence
classes of term constructors. 
Co-induction/recursion on nominal datatypes is not available, but to
be fair it is also currently absent from Isabelle/HOL.



\medskip 

We have presented a proof theoretical treatment of both induction and
co-induction in a sequent calculus compatible with HOAS encodings. The
proof principle underlying the explicit proof rules is basically fixed
point (co)induction. However, the formulation of the rules is inspired
by a second-order encoding of least and greatest fixed points.  We
have developed a new cut elimination proof, radically different from
previous proofs (\cite{mcdowell00tcs,tiu04phd}), using a
reducibility-candidate technique \`a la Girard.

Consistency of the logic is an easy consequence of cut-elimination.
Our proof system is, as far as we know, the first which incorporates a
co-induction proof rule with a direct cut elimination proof. This
schema can be used as a springboard towards cut elimination procedures
for more expressive (conservative) extensions of $\Linc$, for example
in the direction of $\FOLNb$~\cite{miller05tocl}, or more recently,
the logic $LG^\omega$~\cite{Tiu07} by Tiu and the logic $\Gscr$ by
Gacek \etal~\cite{gacek08lics}.



An interesting problem is the connection with circular
proofs, which is particularly attractive from the viewpoint of
proof search, both inductively and co-inductively. This could be
realized by directly proving a cut-elimination result for a logic
where circular proofs, under termination and guardedness conditions
completely replace (co)inductive rules. Indeed, the question whether
``global'' proofs are equivalent to ``local'' proofs
\cite{BrotherstonS07} is still unsettled.

\medskip

\textbf{Acknowledgements} The $\Linc$ logic was developed in
collaboration with Dale Miller. We thank David Baelde for his comments
to a draft of this paper. 


\appendix

\section{The complete set of cut reduction rules}
\label{app:reduc}

\paragraph{Essential cases:}

  \begin{trivlist}
  \item[\fbox{$\landR/\landL$}] If $\Pi_1$ and $\Pi$ are
    \begin{displaymath}
      \infer[\landR]{\Seq{\Delta_1}{B_1' \land B_1''}}
      {\deduce{\Seq{\Delta_1}{B_1'}}
        {\Pi_1'}
	& \deduce{\Seq{\Delta_1}{B_1''}}
        {\Pi_1''}}
      \qquad\qquad
      \infer[\landL]{\Seq{B_1' \land B_1'',B_2,\ldots,B_n,\Gamma}{C}}
      {\deduce{\Seq{B_1',B_2,\ldots,B_n,\Gamma}{C}}
        {\Pi'}}
      \enspace ,
    \end{displaymath}
    then $\Xi$ reduces to $mc(\Pi_1',\Pi_2,\ldots,\Pi_n,\Pi'$.
    The case for the other $\landL$ rule is symmetric.

  \item[\fbox{$\lorR/\lorL$}] Suppose $\Pi_1$ and $\Pi$ are
    \begin{displaymath}
      \infer[\lorR]{\Seq{\Delta_1}{B_1' \lor B_1''}}
      {\deduce{\Seq{\Delta_1}{B_1'}}
        {\Pi_1'}}
      \quad
      \infer[\lorL]{\Seq{B_1' \lor B_1'',B_2,\ldots,B_n,\Gamma}{C}}
      {\deduce{\Seq{B_1',B_2,\ldots,B_n,\Gamma}{C}}
        {\Pi'}
	& \deduce{\Seq{B_1'',B_2,\ldots,B_n,\Gamma}{C}}
        {\Pi''}}      
    \end{displaymath}
    Then $\Xi$ reduces to $mc(\Pi_1',\Pi_2,\ldots,\Pi'$.
    The case for the other $\lorR$ rule is symmetric.

  \item[\fbox{$\oimpR/\oimpL$}] Suppose $\Pi_1$ and $\Pi$ are
    \begin{displaymath}
      \infer[\oimpR]{\Seq{\Delta_1}{B_1' \oimp B_1''}}
      {\deduce{\Seq{B_1',\Delta_1}{B_1''}}
        {\Pi_1'}}
      \qquad
      \infer[\oimpL]{\Seq{B_1' \oimp B_1'',B_2,\ldots,B_n,\Gamma}{C}}
      {\deduce{\Seq{B_2,\ldots,B_n,\Gamma}{B_1'}}
        {\Pi'}
	& \deduce{\Seq{B_1'',B_2,\ldots,B_n,\Gamma}{C}}
        {\Pi''}}
      \enspace 
    \end{displaymath}
    Let $\Xi_1 = mc(mc(\Pi_2,\ldots,\Pi_n,\Pi'),\Pi_1'$. 
    Then $\Xi$ reduces to \settowidth{\infwidthi}
    {$\Seq{\Delta_1,\ldots,\Delta_n,\Gamma,\Delta_2,\ldots,\Delta_n,\Gamma}{C}$}
    \begin{displaymath}
      \infer=[\cL]
      {\Seq{\Delta_1,\ldots,\Delta_n,\Gamma}{C}}
      {
        \infer[\mc]
        {\Seq{\Delta_1,\ldots,\Delta_n,\Gamma, \Delta_2,\ldots,\Delta_n,\Gamma}{C}}
        {
          \raisebox{-2.5ex}{\deduce{\Seq{\ldots}{B_1''}}{\Xi_1}}
          & 
          \left\{\raisebox{-1.5ex}{\deduce{\Seq{\Delta_i}{B_i}}{\Pi_i}}\right\}_{i \in \{2..n\}}
          & \raisebox{-2.5ex}{\deduce{\Seq{B_1'',\{B_i\}_{i \in \{2..n\}},\Gamma}{C}}{\Pi''}}
        }
      }
      \enspace 
    \end{displaymath}

  \item[\fbox{$\forallR/\forallL$}] If $\Pi_1$ and $\Pi$ are
    \begin{displaymath}
      \infer[\forallR]{\Seq{\Delta_1}{\forall x.B_1'}}
      {\deduce{\Seq{\Delta_1}{B_1'[y/x]}}
        {\Pi_1'}}
      \qquad\qquad\qquad
      \infer[\forallL]{\Seq{\forall x.B_1',B_2,\ldots,B_n,\Gamma}{C}}
      {\deduce{\Seq{B_1'[t/x],B_2,\ldots,B_n,\Gamma}{C}}
        {\Pi'}}
      \enspace ,
    \end{displaymath}
    then $\Xi$ reduces to $mc(\Pi_1'[t/y],\Pi_2,\ldots,\Pi_n,\Pi'$.

  \item[\fbox{$\existsR/\existsL$}] If $\Pi_1$ and $\Pi$ are
    \begin{displaymath}
      \infer[\existsR]{\Seq{\Delta_1}{\exists x.B_1'}}
      {\deduce{\Seq{\Delta_1}{B_1'[t/x]}}
        {\Pi_1'}}
      \qquad\qquad\qquad
      \infer[\existsL]{\Seq{\exists x.B_1',B_2,\ldots,
          B_n,\Gamma}{C}}
      {\deduce{\Seq{B_1'[y/x],B_2,\ldots,B_n,
	    \Gamma}{C}}
        {\Pi'}}
      \enspace ,
    \end{displaymath}
    then $\Xi$ reduces to $mc(\Pi_1',\Pi_2,\ldots,\Pi'[t/y]$.

  \item[\fbox{$\indR/\indL$}] 
Suppose $\Pi_1$ and $\Pi$ are, respectively,  
$$
\infer[\indR]
{\Seq {\Delta_1}{p\,\vec t}}
{
 \deduce{\Seq{\Delta_1}{D\,X^p\,\vec t}}{\Pi_1'}
}
\qquad
\infer[\indL] {\Seq{p\,\vec{t}, B_2,\dots,B_n,\Gamma}{C}} {
  \deduce{\Seq{D\,S\,\vec{y}}{S\,\vec{y}}}{\Pi_S} &
  \deduce{\Seq{S\,\vec{t}, B_2,\dots,B_n, \Gamma} {C}}{\Pi'} }
$$
where $p\,\vec{x} \defmu D\,p\,\vec{x}$ and $X^p$ is a new
parameter.  Then $\Xi$ reduces to
$$
mc(mc(\Pi_1'p[(\Pi_S,S)/X^p], \Pi_S[\vec t/\vec y]), \Pi_2,\ldots,\Pi_n,\Pi').
$$


\item[\fbox{$\coindR/\coindL$}] Suppose $\Pi_1$ and $\Pi$ are 
$$
\infer[\coindR] 
{\Seq{\Delta_1}{p\,\vec{t}}} 
{
  \deduce{\Seq{\Delta_1}{S\,\vec{t}}}{\Pi_1'} &
  \deduce{\Seq{S\,\vec{y}}{D\,S\,\vec{y}}}{\Pi_S} 
} 
\qquad \qquad
\infer[\coindL] 
{\Seq{p\,\vec{t}, \dots, \Gamma}{C}}
{
  \deduce{\Seq{D\,X^p\,\vec{t},\dots, \Gamma}{C}}{\Pi'}
}
$$
where $p\,\vec y \defnu D\,p\,\vec y$ and $X^p$ is a new parameter.
Then $\Xi$ reduces to
$$
mc(mc(\Pi_1', \Pi_S[\vec t/\vec y]), \Pi_2,\ldots,\Pi_n,\Pi'[(\Pi_S,S)/X^p]).
$$


\item[\fbox{$\eqR/\eqL$}] 
Suppose $\Pi_1$ and $\Pi$ are
  \begin{displaymath}
    \infer[\eqR]{\Seq{\Delta_1}{s = t}}
    {}
    \qquad\qquad\qquad
    \infer[\eqL]{\Seq{s=t,B_2,\ldots,B_n,\Gamma}{C}}
    {\left\{\raisebox{-1.5ex}
        {\deduce{\Seq{B_2\rho,\ldots,B_n\rho,\Gamma\rho}
            {C\rho}}
          {\Pi^\rho}}
      \right\}_\rho}
    \enspace 
  \end{displaymath}
  Note that in this case, $\rho$ in $\Pi$ ranges over all
  substitution, as any substitution is a unifier of $s$ and $t$.  Let
  $\Xi_1$ be the derivation
  $mc(\Pi_2,\ldots,\Pi_n,subst(\{\Pi^\rho\}_\rho)$. Then $\Xi$
  reduces to
  $$
  \infer=[\wL]
  {\Seq{\Delta_1,\Delta_2,\ldots,\Delta_n,\Gamma}{C}}
  {\deduce{\Seq{\Delta_2,\ldots,\Delta_n,\Gamma}{C}}{\Xi_1}}
  $$


\end{trivlist}

\paragraph{Left-commutative cases:}

In the following cases, we suppose that $\Pi$ ends with a left rule,
other than $\{\cL, \wL\}$, acting on $B_1$.

\begin{trivlist}

\item[\fbox{$\bulletL/\circL$}] Suppose $\Pi_1$ is as below left, 
  where $\bulletL$ is any left rule except $\oimpL$, $\eqL$, or
  $\indL$.  Let $\Xi^i = mc(\Pi_1^i,\Pi_2,\ldots,\Pi_n,\Pi$.
  Then $\Xi$ reduces to the derivation given below right. 
  \begin{displaymath}
    \infer[\bulletL]
          {\Seq{\Delta_1}{B_1}}
          {\left\{\raisebox{-1.5ex}{\deduce{\Seq{\Delta_1^i}{B_1}}
              {\Pi_1^i}}\right\}_i
          }
    \qquad
    \infer[\bulletL]
          {\Seq{\Delta_1,\Delta_2,\ldots,\Delta_n,\Gamma}{C}}
          {
            \left\{
            \raisebox{-1.3ex}{
              \deduce{\Seq{\Delta_1^i,\Delta_2,\ldots,\Delta_n,\Gamma}{C}}{\Xi^i}
            }
            \right\}_i
          }
  \end{displaymath}

\item[\fbox{$\oimpL/\circL$}] Suppose $\Pi_1$ is
  \begin{displaymath}
    \infer[\oimpL]{\Seq{D_1' \oimp D_1'',\Delta_1'}{B_1}}
    {\deduce{\Seq{\Delta_1'}{D_1'}}
      {\Pi_1'}
      & \deduce{\Seq{D_1'',\Delta_1'}{B_1}}
      {\Pi_1''}}
    \enspace 
  \end{displaymath}
  Let $\Xi_1 = mc(\Pi_1'',\Pi_2,\ldots,\Pi_n, \Pi$. 
  Then $\Xi$ reduces to
  \begin{displaymath}
    \infer[\oimpL]
    {\Seq{D_1' \oimp D_1'',\Delta_1',\Delta_2,\ldots,\Delta_n,\Gamma}{C}}
    {
      \infer=[\wL]
      {\Seq{\Delta_1',\Delta_2,\ldots,\Delta_n,\Gamma}{D_1'}}
      {\deduce{\Seq{\Delta_1'}{D_1'}}{\Pi_1'} }
      & 
      \deduce{\Seq{D_1'',\Delta_1',\Delta_2,\ldots,\Delta_n,\Gamma}{C}}{\Xi_1}
    }
    \enspace 
  \end{displaymath}

\item[\fbox{$\indL/\circL$}] Suppose $\Pi_1$ is 
$$
\infer[\indL] {\Seq{p\,\vec{t}, \Delta_1'}{B_1}} {
  \deduce{\Seq{D\,S\,\vec{y}}{S\,\vec{y}}}{\Pi_S} &
  \deduce{\Seq{S\,\vec{t}, \Delta_1'}{B_1}}{\Pi_1'} }
$$      
where $p\,\vec{y} \defmu D\,p\,\vec{y}$.  
Let $\Xi_1 = mc(\Pi_1',\Pi_2,\ldots,\Pi_n,\Pi$.
Then $\Xi$ reduces to
$$
\infer[\indL] {\Seq{p\,\vec{t}, \Delta_1',\dots,\Delta_n}{C}} {
  \deduce{\Seq{D\,S\,\vec{y}}{S\,\vec{y}}}{\Pi_S} &
  \deduce{\Seq{S\,\vec{t}, \Delta_1',\dots,\Delta_n,\Gamma}{C}}{\Xi_1}
}
$$

\item[\fbox{$\eqL/\circL$}] Suppose $\Pi_1$ is as below left. 
  Let $\Xi^\rho = mc(\Pi_1^\rho, \Pi_2\rho,\ldots,\Pi_n\rho, \Pi\rho$.
  Then $\Xi$ reduces to the derivation given below right. 
  \begin{displaymath}
    \infer[\eqL .]{\Seq{s=t,\Delta_1'}{B_1}}
    {\left\{\raisebox{-1.5ex}
        {\deduce{\Seq{\Delta_1'\rho}{B_1\rho}}
          {\Pi_1^{\rho}}}
      \right\}}
    \qquad
    \infer[\eqL .]
        {\Seq{s=t,\Delta_1',\Delta_2,\ldots,\Delta_n,\Gamma}{C}}
        {
          \left\{
          \raisebox{-1.3ex}{
          \deduce{\Seq{\Delta_1'\rho,\Delta_2\rho,\ldots,\Delta_n\rho,\Gamma\rho}{C\rho}}{\Xi^\rho}
          }
          \right\}
        }
  \end{displaymath}


\item[\fbox{$subst/\circL$}] Suppose $\Pi_1$ is $subst(\{\Pi_1^\rho\}_\rho$.
  Then $\Xi$ reduces to $$subst(\{mc(\Pi_1^\rho,\Pi_2\rho,\ldots,\Pi_n\rho,\Pi\rho)\}_\rho$$

\end{trivlist}

\paragraph{Right-commutative cases:}

\begin{trivlist}

\item[\fbox{$-/\circL$}] Suppose $\Pi$ is as given below left, where 
  where $\circL$ is any left rule other than $\oimpL$, $\eqL$, or
  $\indL$ acting on a formula other than $B_1, \ldots, B_n$.
  Let $\Xi^i = mc(\Pi_1,\ldots,\Pi_n,\Pi^i$.
  Then $\Xi$ reduces to the derivation given below right. 
  \begin{displaymath}
    \infer[\circL]{\Seq{B_1,\ldots,B_n,\Gamma}{C}}
    {\left\{\raisebox{-1.5ex}{\deduce{\Seq{B_1,\ldots,B_n,\Gamma^i}{C}}
          {\Pi^i}}\right\}_i}
    \qquad
    \infer[\circL]
          {\Seq{\Delta_1,\ldots,\Delta_n,\Gamma}{C}}
          {
            \left\{
            \raisebox{-1.3ex}{
              \deduce{\Seq{\Delta_1,\ldots,\Delta_n,\Gamma^i}{C}}{\Xi^i}
            }
            \right\}_i
          }
  \end{displaymath}

\item[\fbox{$-/\oimpL$}] Suppose $\Pi$ is
  \begin{displaymath}
    \infer[\oimpL]{\Seq{B_1,\ldots,B_n, D' \oimp D'',\Gamma'}{C}}
    {\deduce{\Seq{B_1,\ldots,B_n,\Gamma'}{D'}}
      {\Pi'}
      & \deduce{\Seq{B_1,\ldots,B_n,D'',\Gamma'}{C}}
      {\Pi''}}
    \enspace 
  \end{displaymath}
  Let $\Xi_1 = mc(\Pi_1,\ldots,\Pi_n,\Pi')$ and let
  $\Xi_2 = mc(\Pi_1,\ldots,\Pi_n,\Pi''$.
  Then $\Xi$ reduces to
  \begin{displaymath}
    \infer[\oimpL]{\Seq{\Delta_1,\ldots,\Delta_n,D' \oimp D'',\Gamma'}{C}}
    {\deduce{\Seq{\Delta_1,\ldots,\Delta_n,\Gamma'}{D'}}
      {\Xi_1}
      & \deduce{\Seq{\Delta_1,\ldots,\Delta_n,D'',\Gamma'}{C}}
      {\Xi_2}}
    \enspace 
  \end{displaymath}


\item[\fbox{$-/\indL$}] Suppose $\Pi$ is 
$$
\infer[\indL] {\Seq{B_1,\dots,B_n, p\,\vec{t},\Gamma'} {C}} {
  \deduce{\Seq{D\,S\,\vec{y}}{S\,\vec{y}}}{\Pi_S} &
  \deduce{\Seq{B_1,\dots,B_n, S\,\vec{t}, \Gamma'}{C}}{\Pi'} }
\enspace ,
$$      
where $p\,\vec{y} \defmu D\,p\,\vec{y}$.  
Let $\Xi_1 = mc(\Pi_1,\ldots,\Pi_n,\Pi'$. 
Then $\Xi$ reduces to
$$
\infer[\indL] {\Seq{\Delta_1,\dots,\Delta_n, p\,\vec{t},\Gamma'}{C}} {
  \deduce{\Seq{D\,S\,\vec{y}}{S\,\vec{y}}}{\Pi_S} &
  \deduce{\Seq{\Delta_1,\dots,\Delta_n, S\,\vec{t}, \Gamma'}{C}}{\Xi_1}
} \enspace 
$$

\item[\fbox{$-/\eqL$}] Suppose $\Pi$ is as shown below left.
Let $\Xi^\rho = mc(\Pi_1\rho,\ldots,\Pi_n\rho,\Pi^\rho$.
Then $\Xi$ reduces to the derivation below right.
  \begin{displaymath}
    \infer[\eqL]{\Seq{B_1,\ldots,B_n,s=t,\Gamma'}{C}}
    {\left\{\raisebox{-1.5ex}
        {\deduce{\Seq{B_1\rho,\ldots, B_n\rho,\Gamma'\rho}{C\rho}}
          {\Pi^{\rho}}}\right\}}
    \qquad
    \infer[\eqL]
    {\Seq{\Delta_1,\ldots,\Delta_n,s=t,\Gamma'}{C}}
    {
      \left\{
      \raisebox{-1.3ex}
      {
        \deduce{\Seq{\Delta_1\rho,\ldots,\Delta_n\rho,\Gamma'\rho}{C\rho}}{\Xi^\rho}
      }
      \right\}
    }
  \end{displaymath}

\item[\fbox{$-/ subst$}] If $\Pi = subst(\{\Pi^\rho\}_\rho)$ 
then $\Xi$ reduces to $subst(\{mc(\Pi_1\rho,\ldots,\Pi_n\rho,\Pi^\rho)\}_\rho$.

\item[\fbox{$-/\circR$}] If $\Pi$ is as below left, where 
  where $\circR$ is any right rule except $\coindR$, then $\Xi$ reduces
  to the derivation below right, where $\Xi^i = mc(\Pi_1,\ldots,\Pi_n,\Pi^i$. 
  \begin{displaymath}
    \infer[\circR]{\Seq{B_1,\ldots,B_n,\Gamma}{C}}
    {\left\{\raisebox{-1.5ex}{
          \deduce{\Seq{B_1,\ldots,B_n,\Gamma^i}{C^i}}
          {\Pi^i}}\right\}_i}
    \qquad
    \infer[\circR]
          {\Seq{\Delta_1,\ldots,\Delta_n,\Gamma}{C}}
          {
            \left\{
            \raisebox{-1.3ex}{
              \deduce{\Seq{\Delta_1,\ldots,\Delta_n,\Gamma^i}{C^i}}{\Xi^i}
              }
            \right\}_i
          }
  \end{displaymath}

\item[\fbox{$-/\coindR$}] Suppose $\Pi$ is 
$$
\infer[\coindR] {\Seq{B_1,\dots,B_n,\Gamma}{p\,\vec{t}}} {
  \deduce{\Seq{B_1,\dots,B_n,\Gamma} {S\,\vec{t}}}{\Pi'} &
  \deduce{\Seq{S\,\vec{y}}{D\,S\,\vec{y}}}{\Pi_S} } \enspace ,
$$
where $p\,\vec{y} \defnu D\,p\,\vec{y}$.  
Let $\Xi_1 = mc(\Pi_1,\ldots,\Pi_n,\Pi'$. 
Then $\Xi$ reduces to
$$
\infer[\coindR] {\Seq{\Delta_1,\dots,\Delta_n,\Gamma}{p\,\vec{t}}} {
  \deduce{\Seq{\Delta_1,\dots,\Delta_n,\Gamma} {S\,\vec{t}}}{\Xi_1} &
  \deduce{\Seq{S\,\vec{y}}{D\,S\,\vec{y}}}{\Pi_S} } \enspace
$$
\end{trivlist}

\paragraph{Multicut cases:}

\begin{trivlist}

\item[\fbox{$\mc/\circL$}] If $\Pi$ ends with a left rule, other than $\cL$ and
  $\wL$, acting on $B_1$ and $\Pi_1$ ends with a multicut
  and reduces to $\Pi_1'$, then $\Xi$ reduces to $mc(\Pi_1',\Pi_2,\ldots,\Pi_n,\Pi$.

\item[\fbox{$-/\mc$}] Suppose $\Pi$ is
  \begin{displaymath}
    \infer[\mc]{\Seq{B_1,\ldots,B_n,\Gamma^1,\ldots,\Gamma^m,\Gamma'}{C}}
    {\left\{\raisebox{-1.5ex}{\deduce{\Seq{\{B_i\}_{i \in I^j},\Gamma^j}{D^j}}
          {\Pi^j}}\right\}_{j \in \{1..m\}}
      & \raisebox{-2.5ex}{\deduce{\Seq{\{D^j\}_{j \in \{1..m\}},\{B_i\}_{i \in I'},\Gamma'}{C}}
        {\Pi'}}}
    \enspace ,
  \end{displaymath}
  where $I^1,\ldots,I^m,I'$ partition the formulas $\{B_i\}_{i \in
    \{1..n\}}$ among the premise derivations $\Pi_1$, \ldots,
  $\Pi_m$,$\Pi'$.  For $1 \leq j \leq m$ let $\Xi^j$ be
  \begin{displaymath}
    \infer[\mc]{\Seq{\{\Delta_i\}_{i \in I^j},\Gamma^j}{D^j}}
    {\left\{\raisebox{-1.5ex}{\deduce{\Seq{\Delta_i}{B_i}}
          {\Pi_i}}\right\}_{i \in I^j}
      & \raisebox{-2.5ex}{\deduce{\Seq{\{B_i\}_{i \in I^j},\Gamma^j}{D^j}}
        {\Pi^j}}}
    \enspace 
  \end{displaymath}
  Then $\Xi$ reduces to
  \begin{displaymath}
    \infer[\mc]{\Seq{\Delta_1,\ldots,\Delta_n,\Gamma^1,\ldots\Gamma^m,\Gamma'}{C}}
    {\left\{\raisebox{-1.5ex}{\deduce{\Seq{\ldots}{D^j}}
          {\Xi^j}}\right\}_{j \in \{1..m\}}
      & \left\{\raisebox{-1.5ex}{\deduce{\Seq{\Delta_i}{B_i}}
          {\Pi_i}}\right\}_{i \in I'}
      & \raisebox{-2.5ex}{\deduce{\Seq{\ldots}{C}}
        {\Pi'}}}
    \enspace 
  \end{displaymath}

\end{trivlist}

\paragraph{Structural cases:}

\begin{trivlist}
\item[\fbox{$-/\cL$}] If $\Pi$ is as shown below left, then
$\Xi$ reduces to the derivation shown below right, 
where $\Xi_1 = mc(\Pi_1,\Pi_1,\Pi_2,\ldots,\Pi_n,\Pi')$. 
  \begin{displaymath}
    \infer[\cL]{\Seq{B_1,B_2,\ldots,B_n,\Gamma}{C}}
    {\deduce{\Seq{B_1,B_1,B_2,\ldots,B_n,\Gamma}{C}}
      {\Pi'}}
    \qquad
    \infer=[\cL]
    {\Seq{\Delta_1,\Delta_2,\ldots,\Delta_n,\Gamma}{C}}
    {
      \deduce{\Seq{\Delta_1,\Delta_1,\Delta_2,\ldots,\Delta_n,\Delta_n,\Gamma}{C}}{\Xi_1}
    }
  \end{displaymath}

\item[\fbox{$-/\wL$}] If $\Pi$ is as shown below left,
then $\Xi$ reduces to the derivation shown below right, where
$\Xi_1 = mc(\Pi_2,\ldots,\Pi_n,\Pi'$.
  \begin{displaymath}
    \infer[\wL]{\Seq{B_1,B_2,\ldots,B_n,\Gamma}{C}}
    {\deduce{\Seq{B_2,\ldots,B_n,\Gamma}{C}}
      {\Pi'}}
    \qquad
    \infer[\wL]
    {\Seq{\Delta_1,\Delta_2,\ldots,\Delta_n,\Gamma}{C}}
    {
      \deduce{\Seq{\Delta_2,\ldots,\Delta_n,\Gamma}{C}}{\Xi_1}
    }
  \end{displaymath}

\end{trivlist}

\paragraph{Axiom cases:}

\begin{trivlist}
\item[\fbox{$\init/\circL$}] Suppose $\Pi$ ends with a left-rule acting on
  $B_1$ and $\Pi_1$ ends with the $\init$ rule. Then it must be the
  case that $\Delta_1 = \{B_1\}$ and $\Xi$ reduces to
  $mc(\Pi_2,\ldots,\Pi_n,\Pi$.

\item[\fbox{$-/\init$}] If $\Pi$ ends with the $\init$ rule, then $n = 1$,
  $\Gamma$ is the empty multiset, and $C$ must be a cut formula, i.e.,
  $C = B_1$. Therefore $\Xi$ reduces to $\Pi_1$.
\end{trivlist}

\section{Proofs for Section~\ref{sec:reduc} and Section~\ref{sec:red}}
\label{app:red}

\begin{lemmacp}{lm:reduct_subst}
  Let $\Pi$ be a derivation ending with a $\mc$
  and let $\theta$ be a substitution.  If $\Pi\theta$ reduces to $\Xi$
  then there exists a derivation $\Pi'$ such that $\Xi = \Pi'\theta$
  and $\Pi$ reduces to $\Pi'$.
\end{lemmacp}
\begin{proof}
  Observe that the redexes of a derivation are not affected by
  eigenvariable substitution, since the cut reduction rules are
  determined by the last rules of the premise derivations, which are
  not changed by substitution. Therefore, any cut reduction rule that
  is applied to $\Pi\theta$ to get $\Xi$ can also be applied to
  $\Pi$. Suppose that $\Pi'$ is the reduct of $\Pi$ obtained this
  way. In all cases, except for the cases where the reduction rule
  applied is either $\indR/\indL$, $\coindL/ \coindR$, or those
  involving $\eqL$, it is a matter of routine to check that
  $\Pi'\theta = \Xi$. For the reduction rules $\indR/\indL$ and
  $\coindL/ \coindR$, we need Lemma~\ref{lm:param subst} which shows
  that eigenvariable substitution commutes with parameter
  substitution.  We show here the case involving $\eqL$. The only
  interesting case is the reduction $\eqL/\eqR$. For simplicity, we
  show the case where $\Pi$ ends with $mc$ with three premises; it is
  straightforward to adapt the following analysis to the more general
  case.  So suppose $\Pi$ is the derivation:
$$
\infer[mc]
{\Seq{\Delta_1, \Delta_2, \Gamma}{C} }
{
\infer[\eqR]
{\Seq{\Delta_1}{t = t}}{}
&
\deduce{\Seq{\Delta_2}{B}}{\Pi_2}
& 
\infer[\eqL]
{\Seq{t=t, B, \Gamma}{C}}
{
\left\{
\raisebox{-1.5ex}{ \deduce{\Seq{B\rho, \Gamma\rho}{C\rho}}{\Pi^\rho} }
\right\}_\rho
}
}
$$
According to Definition~\ref{def:subst}, the derivation $\Pi\theta$ is 
$$
\infer[mc]
{\Seq{\Delta_1\theta,\Delta_2\theta, \Gamma\theta}{C\theta} }
{
\infer[\eqR]
{\Seq{\Delta_1\theta}{t\theta = t\theta}}{}
& 
\deduce{\Seq{\Delta_2\theta}{B\theta}}{\Pi_2\theta}
&
\infer[\eqL]
{\Seq{t\theta=t\theta, B\theta, \Gamma\theta}{C\theta}}
{
\left\{
\raisebox{-1.5ex}{ \deduce{\Seq{B\theta\rho', \Gamma\theta\rho'}{C\theta\rho'}}{\Pi^{(\theta\circ \rho')} }}
\right\}_{\rho'}
}
}
$$
Let $\Psi = mc(\Pi_2\theta, subst(\{\Pi^{(\theta\circ \rho)} \}_\rho)$.
The reduct of $\Pi\theta$ in this case (modulo the different
order in which the weakening steps are applied) is:  
$$
\infer=[\wL]
{\Seq{\Delta_1\theta, \Delta_2\theta, \Gamma\theta}{C\theta}}
{
\deduce{\Seq{\Delta_2\theta,\Gamma\theta}{C\theta}}{\Psi}
}
$$
Let us call this derivation $\Xi$.

Let $\Psi' = mc(\Pi_2, subst(\{\Pi^\rho\}_\rho)$. 
The above reduct can be matched by the following reduct of $\Pi$ (using the same
order of applications of the weakening steps): 
$$
\infer=[\wL]
{\Seq{\Delta_1, \Delta_2, \Gamma}{C}}
{
\deduce{\Seq{\Delta_2,\Gamma}{C}}{\Psi'}
}
$$
Let us call this derivation $\Pi'$. 
By Definition~\ref{def:subst}, we have $\Psi' = \Psi\theta$, and
obviously, also $\Xi = \Pi'\theta$.
\qed
\end{proof}

\begin{lemmacp}{lm:red-norm}
If $\Pi \in \RED_C[\Omega]$ then $\Pi$ is normalizable.
\end{lemmacp}
\begin{proof}
By case analysis on $C$.  
If $C = X^p\,\vec u$ for some $\vec u$ and $X^p \in supp(\Omega)$
then 
$\Pi \in \Rscr$, where $\Omega(X^p) = (\Rscr, \Pi_S, S)$, 
hence it is normalizable by Definition~\ref{def:candidates}
(specifically, condition {\bf CR1}). 
Otherwise, $\Pi$ is normalizable by Definition~\ref{def:param red}. \qed
\end{proof}


\begin{lemmacp}{lm:red-subst}
  If $\Pi \in \RED_C[\Omega]$ then for every substitution $\rho$, 
  $\Pi\rho \in \RED_{C\rho}[\Omega]$. 
\end{lemmacp}
\begin{proof}
By induction on $|C|$ with sub-induction on $nd(\Pi)$. 

Suppose $C = X^q\,\vec u$, for some $\vec u$ and some
$X^q \in supp(\Omega)$, and suppose $\Omega(X^q) = (\Rscr, \Pi_S, S)$. 
Then $\Pi \in \Rscr$ by Definition~\ref{def:param red}. 
By Definition~\ref{def:candidates} ({\bf CR0}) we also have
$\Pi\rho \in \Rscr$. 
Otherwise, suppose $X^q \not \in supp(\Omega)$. 
Then $\Pi \in \NM_{X^q}$ by Definition~\ref{def:param red}.
By Lemma~\ref{lm:subst-norm}, we have $\Pi\rho \in \NM_{X^q}$,
therefore $\Pi\rho \in \RED_{C\rho}[\Omega]$. 

Otherwise, $C \not = X^q \,\vec u$ for any $\vec u$ and any parameter
$X^q$.  In this case, to apply the inner induction hypothesis, we need
to show that $\Pi\rho$ is normalizable, which follows immediately from
Lemma~\ref{lm:red-norm} and Lemma~\ref{lm:subst-norm}. We distinguish
several cases based on the last rule of $\Pi$:

\begin{itemize}
\item Suppose $\Pi$ ends with $mc$, i.e., 
$\Pi = mc(\Pi_1,\ldots,\Pi_n,\Pi')$ for some $\Pi_1,\ldots,\Pi_n$
and $\Pi'$. 
By Lemma~\ref{lm:reduct_subst}, every reduct of $\Pi\rho$, say $\Xi$, 
is the result of applying $\rho$ to a reduct of $\Pi$.  
By the inner induction hypothesis (on
the normalization degree), every reduct of $\Pi\rho$ 
is in $\RED_{C\rho}[\Omega]$, and therefore $\Pi\rho$ is 
also in $\RED_{C\rho}[\Omega]$ by Definition~\ref{def:param red} ({\bf P2}).

\item Suppose $\Pi$ ends with $\oimpR$, with the premise derivation
$\Pi'$. In this case, $C = B \oimp D$ for some $B$ and $D$.
Since $\Pi \in \RED_C[\Omega]$, we have that (\textbf{P3})
\begin{equation}
\label{eq:param1}
\Pi'\theta \in (\RED_{B\theta}[\Omega] \Rightarrow \RED_{D\theta}[\Omega])
\end{equation}
for every $\theta$. 
We need to show that 
$\Pi'\rho\delta \in (\RED_{B\rho\delta}[\Omega] \Rightarrow
\RED_{D\rho\delta}[\Omega])$ 
for every $\delta$. Note that by Lemma~\ref{lm:subst-drv-comp},
$\Pi'\rho\delta = \Pi'(\rho \circ \delta)$, so 
this is just an instance of Statement~\ref{eq:param1} above.

\item $\Pi$ ends with $\indR$ or $\coindR$: This follows from Definition~\ref{def:param red}
and the fact that reducibility candidates are closed under substitution
(condition {\bf CR0} in Definition~\ref{def:candidates}). In the case where 
$\Pi$ ends with $\indR$, we also need the fact
that eigenvariable substitution commutes with parameter substitution
(Lemma~\ref{lm:param subst commutes}). In the case where $\Pi$ ends 
with $\coindR$, to establish $\Pi\rho \in \RED_{C\rho}[\Omega]$,
we can use the same reducibility candidate which is used to establish
$\Pi \in \RED_C[\Omega]$.

\item $\Pi$ ends with a rule other than $mc$, $\oimpR$, $\indR$ or $\coindR$:
This case follows straightforwardly from the induction hypothesis.
\end{itemize}
\qed
\end{proof}

\begin{lemmacp}{lm:red vacuous}
Let $\Omega = [\Omega', (\Rscr, \Pi_S, S)/ X^p]$.
Let $C$ be a formula such that $X^p \# C$. 
Then for every $\Pi$,  $\Pi \in \RED_C[\Omega]$ if and only if
$\Pi \in \RED_C[\Omega']$. 
\end{lemmacp}
\begin{proof}
By induction on $|C|$ with sub-induction on $nd(\Pi)$. 

Suppose $C = Y^q\,\vec u$ for some $Y^q \in supp(\Omega)$ 
and suppose $\Omega(Y^q) = (\Rscr', \Pi_I, I)$. Since $X^p \# C$,
this means that $Y^q \in supp(\Omega')$ and $\Omega'(Y^q) = \Omega(Y^q)$. 
Then obviously, $\Pi \in \RED_C[\Omega]$ iff $\Pi \in \RED_C[\Omega']$. 
If $Y^q \not \in supp(\Omega)$, then
obviously $\RED_C[\Omega] = \NM_{Y^q}\,\vec u = \RED_C[\Omega']$. 

Otherwise, suppose $C \not = Y^q\,\vec u$, and  $\Pi \in \RED_C[\Omega]$. 
The latter implies that $\Pi$ is normalizable.
We show, by induction on $nd(\Pi)$ that $\Pi \in \RED_C[\Omega']$. 
In most cases, this follows straightforwardly from the induction
hypothesis. We show the interesting cases here:
\begin{itemize}
\item Suppose $\Pi$ ends with $\oimpR$, i.e., $C = B \oimp D$ for some $B$ and $D$
and $\Pi$ is of the form:
$$
\infer[\oimpR]
{\Seq{\Gamma}{B\Omega \oimp D\Omega}}
{
 \deduce{\Seq{\Gamma, B\Omega}{D\Omega}}{\Pi'}
}
$$
Note that since $X^p \# C$, we have that $B\Omega=B\Omega'$ and
$D\Omega = D\Omega'$. 
Since $\Pi \in \RED_C[\Omega]$, we have
$$
\Pi'\rho \in (\RED_{B\rho}[\Omega] \Rightarrow \RED_{D\rho}[\Omega])
$$
for every $\rho$. 
Since $|B| < |C|$ and $|D| < |D|$, by the (outer) induction hypothesis,
we have $\RED_{B\rho}[\Omega] = \RED_{B\rho}[\Omega']$
and $\RED_{D\rho}[\Omega] = \RED_{D\rho}[\Omega']$. 
Therefore, we also have that 
$$
\Pi'\rho \in (\RED_{B\rho}[\Omega'] \Rightarrow \RED_{D\rho}[\Omega'])
$$
for every $\rho$. This means, by Definition~\ref{def:param red}, that 
$\Pi \in (\RED_C[\Omega']$. 

\item Suppose $\Pi$ ends with $\indR$:
$$
\infer[\indR]
{\Seq {\Gamma} {q\,\vec t}}
{\deduce{\Seq \Gamma {D\,Y^q\,\vec t}}{\Pi'}}
$$
where $q\,\vec x \defmu D\,q\,\vec x$ and $Y^q$ is a new parameter.
Since we identify derivations which differ only in the choice of
internal variables and parameters, we can assume without loss of generality
that $Y^q \# \Omega$. Note that since the body of a definition cannot
contain occurrences of parameters, we also have $X^p \# D\,Y^q\,\vec t$. 
Suppose $\Sscr$ is a reducibility candidate of type $I$, for some
closed term $I$ of the same syntactic type as $q$, and suppose
$\Pi_I$ is a normalizable derivation of $\Seq {D\,I\,\vec y}{I\,\vec y}$
such that
\begin{equation}
\label{eq:vac1}
\Pi_I[\vec u/\vec y] \in (
\RED_{(D\,Y^q\,\vec u)}[\Omega', (\Sscr,\Pi_I, I)/Y^q] \Rightarrow \Sscr
\,\vec u) 
\end{equation}
for every $\vec u$ of the appropriate types. To show that $\Pi \in (\RED_C[\Omega']$
we need to show that 
$$mc(\Pi'[(\Pi_I,I)/Y^q], \Pi_I[\vec t/ \vec y]) \in \Sscr\,\vec t$$ 

Since $|(D\,Y^q\,\vec u)| < |p\,\vec t|$ by Lemma~\ref{lm:level}, we have, by the outer induction hypothesis, 
$$
\RED_{(D\,Y^q\,\vec u)}[\Omega', (\Sscr,\Pi_I, I)/Y^q] 
=
\RED_{(D\,Y^q\,\vec u)}[\Omega, (\Sscr,\Pi_I, I)/Y^q] 
$$
Hence, by Statement~\ref{eq:vac1}, we also have
$$
\Pi_I[\vec u/\vec y] \in (
\RED_{(D\,Y^q\,\vec u)}[\Omega, (\Sscr,\Pi_I, I)/Y^q] \Rightarrow \Sscr \,\vec u)
$$
for arbitrary $\vec u$.  
Now since $\Pi \in (\RED_C[\Omega]$ (from the assumption), this means
that 
$$
mc(\Pi'[(\Pi_I,I)/Y^q], \Pi_I[\vec t/ \vec y]) \in \Sscr\,\vec t
$$
and therefore $\Pi$ is indeed in $\RED_C[\Omega']$. 

\item Suppose $\Pi$ ends with $\coindR$:
$$
\infer[\coindR]
{\Seq \Gamma {q\,\vec t}}
{
  \deduce{\Seq \Gamma {I\,\vec t}}{\Pi'}
  &
  \deduce
  {\Seq {I\,\vec y}{B\,I\,\vec y}}
  {\Pi_I}
}
$$
where $q\,\vec x \defnu B\,q\,\vec x$. 
Since $\Pi \in \RED_C[\Omega]$, by Definition~\ref{def:param red} ({\bf P4}),
there exist a parameter $Y^q$ such
that $Y^q \# \Omega$ and a reducibility candidate $(\Sscr : I)$
such that $\Pi' \in \Sscr$ and 
\begin{equation}
\label{eq:vac2}
\Pi'[\vec u/\vec y] \in (\Sscr\,\vec u \Rightarrow 
\RED_{B\,Y^q\,\vec u}[\Omega, (\Sscr,\Pi_I,I)/Y^q])
\end{equation}
for every $\vec u$. 
To show $\Pi \in \RED_C[\Omega']$ we need to find a reducibility
candidate satisfying {\bf P4}. We simply use $\Sscr$ as that candidate.
It remains to show that
$$
\Pi'[\vec u/\vec y] \in (\Sscr\,\vec u \Rightarrow 
\RED_{B\,Y^q\,\vec u}[\Omega', (\Sscr,\Pi_I,I)/Y^q])
$$
This follows from Statement (\ref{eq:vac2}) above and
the outer induction hypothesis, since 
$$\RED_{B\,Y^q\,\vec u}[\Omega, (\Sscr,\Pi_I,I)/Y^q]
= \RED_{B\,Y^q\,\vec u}[\Omega', (\Sscr,\Pi_I,I)/Y^q]$$
\end{itemize}

The converse, i.e., $\Pi \in \RED_C[\Omega']$
implies $\Pi \in \RED_C[\Omega]$, can be proved analogously.
In particular, in the case where $\Pi$ ends with $\coindR$,
we rely on the fact that the choice of the new parameter
$Y^q$ is immaterial, as long as it is new, so
we can assume without loss of generality that $Y^q \not = X^p$. 
\qed
\end{proof}

\begin{lemmacp}{lm:red candidate}
  Let $\Omega$ be a candidate substitution and $F$  a closed term of
  type $\alpha_1 \ra \cdots \ra \alpha_n \ra
  o$. Then the set
$ 
\Rscr = \{\Pi \mid \Pi \in \RED_{F\,\vec u}[\Omega] \ \hbox{for some
  $\vec u$} \} 
$
is a reducibility candidate of type $F\Omega$.
\end{lemmacp}
\begin{proof}
  Suppose $F = X^p$ for some $X^p \in supp(\Omega)$ and suppose
  $\Omega(X^p) = (\Sscr, \Pi, F)$.  Then in this case, we have $\Rscr
  = \Sscr$, so $\Rscr$ is a reducibility candidate of type $F$ by
  assumption.  If $F = X^p$ but $X^p \not \in supp(\Omega)$ then in
  this case $\Rscr = \NM_{X^p}$, and by Lemma~\ref{lm:norm red},
  $\Rscr$ is also a reducibility candidate.

  Otherwise, $F \not = X^p$ for any parameter $X^p$.  We need to show
  that $\Rscr$ satisfies {\bf CR0} - {\bf CR5}.  {\bf CR0} follows
  from Lemma~\ref{lm:red-subst}.  {\bf CR1} follows from
  Lemma~\ref{lm:red-norm}, and the rest follow from
  Definition~\ref{def:param red}.  \qed
\end{proof}

\begin{lemmacp}{lm:red param subst}
Let $\Omega$ be a candidate substitution and let $X^p$
be a parameter such that $X^p \# \Omega$. 
Let $S$ be a closed term of the same type as $p$ and let
$$
\Rscr = \{\Pi \mid \Pi \in \RED_{S\,\vec u}[\Omega] \ \hbox{for some $\vec u$} \}.
$$
Suppose $[\Omega, (\Rscr, \Psi, S\Omega)/ X^p]$ is a candidate
substitution, for some $\Psi$. 
Then
$$
\RED_{C[S/X^p]}[\Omega] = \RED_C[\Omega, (\Rscr, \Psi, S\Omega)/X^p].
$$
\end{lemmacp}
\begin{proof}
By induction on $|C|$. 
If $C = X^p \,\vec u$, then
$$\RED_C[\Omega, (\Rscr, \Psi,S\Omega)/X^p] = \Rscr\,\vec u = \RED_{S\,\vec u}[\Omega]$$
by assumption. The other cases where $C$ is $Y^q\,\vec u$
for some parameter $Y^q \not = X^p$ are straightforward. 
So suppose $C \not = Y^q\,\vec u$ for any $\vec u$ and any parameter $Y^q$. 
We show that for every $\Pi$, $\Pi \in \RED_{C[S/X^p]}[\Omega]$
iff $\Pi \in \RED_C[\Omega, (\Rscr, \Psi, S\Omega)/X^p]$. 
Note that if $X^p$ does not occur in $C$ then
$C[S/X^p] = C$, and by Lemma~\ref{lm:red vacuous} we have
$$
\RED_{C[S/X^p]}[\Omega] = \RED_C[\Omega]=\RED_C[\Omega, (\Rscr, \Psi, S\Omega)/X^p].
$$
So assume that $X^p$ is not vacuous in $C$. 
Let $\Omega' = [\Omega, (\Rscr, \Psi, S\Omega)/X^p]$. 
\begin{itemize}
\item Suppose $\Pi \in \RED_{C[S/X^p]}[\Omega]$. 
Then $\Pi$ is normalizable. We show, by induction on $nd(\Pi)$,
that $\Pi \in \RED_C[\Omega']$. 
Most cases follow immediately from the induction hypothesis.
The only interesting case is when $\Pi$ ends with $\oimpR$,
where $C = B \oimp D$, for some $B$ and $D$,  and $\Pi$ takes
the form:
$$
\infer[\oimpR]
{\Seq \Gamma {B[S/X^p]\Omega \oimp D[S/X^p]\Omega}}
{
  \deduce{\Seq {\Gamma, B[S/X^p]\Omega}{D[S/X^p]\Omega}}{\Pi'}
}
$$
Since $\Pi \in \RED_{C[S/X^p]}[\Omega]$, we have that
$$
\Pi'\rho \in (\RED_{B[S/X^p]\rho}[\Omega] \Rightarrow \RED_{D[S/X^p]\rho}[\Omega])
$$
for every $\rho$.
By the outer induction hypothesis (on the size of $C$), we have 
$$
\Pi'\rho \in (\RED_{B\rho}[\Omega'] \Rightarrow \RED_{D\rho}[\Omega'])
$$
hence $\Pi \in \RED_{C}[\Omega']$. 

\item The converse, i.e., $\Pi \in \RED_C[\Omega']$
implies $\Pi \in \RED_{C[S/X^p]}[\Omega]$, can be proved analogously.
\end{itemize}
\qed
\end{proof}

\end{document}